%% file: adversarial.tex
\documentclass[12pt]{article}
\pdfoutput=1

\usepackage[auth-sc]{authblk}

\title{\MakeUppercase{An Adversarial Approach to Structural Estimation}}

\author[1]{Tetsuya Kaji}
\author[2]{Elena Manresa}
\author[1]{Guillaume Pouliot}
\affil[1]{University of Chicago}
\affil[2]{New York University}

\date{\normalsize\today}

\usepackage[margin=1.25in]{geometry}
\usepackage[onehalfspacing,nodisplayskipstretch]{setspace} 

\usepackage{amsthm,amsmath,amssymb,calc}
\usepackage{thmtools}
\usepackage{mathrsfs}
\usepackage{mathtools}
\usepackage{bm}
\usepackage{bbm}

\usepackage{titling}
\usepackage[tiny]{titlesec}
\usepackage{chngcntr}

\usepackage[authoryear]{natbib}
\usepackage{breakcites}
\RequirePackage[colorlinks,citecolor=blue,urlcolor=blue]{hyperref}
\RequirePackage{cleveref}
\usepackage{hyphenat} 

\usepackage{booktabs} 
\usepackage{chapterbib} 
\usepackage{enumerate}
\usepackage{rotating}
\usepackage{xcolor}

\usepackage{graphicx,epstopdf,epsfig}
\usepackage{caption,subcaption}
\usepackage{pgfplots}

\usepackage[T1]{fontenc}
\usepackage{textcomp} 
\usepackage{lmodern}

\usepackage{enumitem}

\usepackage{algorithm}
\usepackage[noend]{algpseudocode}

\allowdisplaybreaks

\pretitle{\vspace{-5em}\begin{center}\fontsize{13}{15}\selectfont}
\posttitle{\par\end{center}}
\preauthor{\vspace{-1em}\begin{center}
	\normalsize \scshape \lineskip 0.5em%
	\begin{tabular}[t]{c}}
\postauthor{\end{tabular}\par\end{center}}
\predate{\begin{center}\large}
\postdate{\par\end{center}\vspace{-1em}}

\makeatletter
\newcommand\new@setfontsize[3]{%
    \ifx \protect \@typeset@protect \let \@currsize #1\fi \fontsize {#2}{#3}\selectfont
}
\let\orig@setfontsize\@setfontsize
\let\orig@cases\cases
\let\endorig@cases\endcases

\renewenvironment{cases}{%
    \let\@setfontsize\new@setfontsize
    \setstretch{\setspace@singlespace}%
    \let\setfontsize\orig@setfontsize
    \orig@cases
}{%
    \endorig@cases
}
\makeatother

\titleformat{\section}{\normalfont\centering}{\thesection}{1em}{\MakeUppercase}
\titleformat*{\subsection}{\itshape\centering}

\usepackage{etoolbox}
\usepackage{cleveref}


\theoremstyle{plain}
\newtheorem{thm}{Theorem}
\newtheorem*{thm*}{Theorem}
\newtheorem{prop}[thm]{Proposition}
\newtheorem{lem}{Lemma}
\newtheorem{coro}[thm]{Corollary}

\theoremstyle{definition}
\newtheorem*{defn}{Definition}
\declaretheorem[name=Assumption,qed={\hfill$\square$}]{asm}
\declaretheorem[name=Example,qed={\hfill$\square$}]{exa}
\newtheorem*{exa*}{Example}

\theoremstyle{remark}
\newtheorem*{rem}{Remark}

\renewcommand\thmcontinues[1]{continued}

\crefname{thm}{Theorem}{Theorems}
\crefname{prop}{Proposition}{Propositions}
\crefname{lem}{Lemma}{Lemmas}
\crefname{coro}{Corollary}{Corollaries}
\crefname{add}{Addendum}{Addendums}
\crefname{asm}{Assumption}{Assumptions}
\crefname{alg}{Algorithm}{Algorithms}
\crefname{proc}{Procedure}{Procedures}
\crefname{exa}{Example}{Examples}
\crefname{section}{Section}{Sections}
\crefname{subsection}{Section}{Sections}
\crefname{appendix}{Appendix}{Appendices}

\crefrangelabelformat{exa}{#3#1#4--#5\crefstripprefix{#1}{#2}#6}

\DeclareMathOperator{\Var}{Var}

\DeclareMathOperator{\erfc}{erfc}
\DeclareMathOperator*{\conv}{\mathchoice{%
	\,\longrightarrow\,}{
	\rightarrow}{
	\rightarrow}{
	\rightarrow}
}

\DeclareMathOperator*{\maxg}{max\vphantom{g}} 

\def\argmax{\mathop{\arg\max}}	
\def\argmin{\mathop{\arg\min}}		

\makeatletter
\def\blfootnote{\gdef\@thefnmark{}\@footnotetext}
\makeatother

\allowdisplaybreaks[3]

\makeatletter
\newcommand*\bigcdot{\mathpalette\bigcdot@{.5}}
\newcommand*\bigcdot@[2]{\mathbin{\vcenter{\hbox{\scalebox{#2}{$\m@th#1\bullet$}}}}}
\makeatother

\counterwithin{subsection}{section}

\begin{document}

\include{maintext}

\include{supptext}

\end{document}

%% file: maintext.tex
\defcitealias{dfj}{DFJ}

\maketitle

\begin{abstract}
We propose a new simulation\hyp{}based estimation method, adversarial estimation, for structural models. The estimator is formulated as the solution to a minimax problem between a generator (which generates simulated observations using the structural model) and a discriminator (which classifies whether an observation is simulated). The discriminator maximizes the accuracy of its classification while the generator minimizes it. We show that, with a sufficiently rich discriminator, the adversarial estimator attains parametric efficiency under correct specification and the parametric rate under misspecification. We advocate the use of a neural network as a discriminator that can exploit adaptivity properties and attain fast rates of convergence. We apply our method to the elderly's saving decision model and show that our estimator uncovers the bequest motive as an important source of saving across the wealth distribution, not only for the rich.
\bigskip

\vspace{4pt}

\vspace{4pt}
\end{abstract}

\blfootnote{We thank Mariacristina De Nardi and John Jones for sharing the data and codes for the empirical application and for very helpful discussion. We also thank Isaiah Andrews, Manuel Arellano, Stephane Bonhomme, Aureo De Paula, Costas Meghir, Chris Hansen, Koen Jochmans, Whitney Newey, Luigi Pistaferri, Bernard Salanie, Dennis Kristensen, Anna Mikusheva, Zhenling Jiang, Xintong Han, and Daniel Waldinger, as well as numerous participants in conferences and venues for helpful discussion. Elsie Hoffet, Yijun Liu, Ignacio Ciggliutti, and Marcela Barrios provided superb research assistance. We gratefully acknowledge the support of the NSF by means of the Grant SES\hyp{}1824304 and the Richard N.\ Rosett Faculty Fellowship and the Liew Family Faculty Fellowship at the University of Chicago Booth School of Business.}

\section{Introduction}
Structural estimation is a useful tool to %
learn about the effects of policies that are yet to be implemented.
Structural models are naturally articulated as parametric models and, as such, may be estimated using maximum likelihood (MLE). However, likelihood functions are sometimes too complex to evaluate or may not exist in closed form.
This has spurred large literature on simulation\hyp{}based estimation methods.

A prominent example of such methods is the simulated method of moments (SMM) \citep{m1989}.
If we want identification and estimation of the parameters to rely on specific features, SMM is a natural tool as long as such features can be expressed as moments.
At the same time, a naive strategy to stack many moments is known to yield poor finite sample properties \citep{altonji1996small}. This tradeoff is especially pronounced in models with rich heterogeneity, where the number of moments may grow rapidly with the number of covariates. %
While this problem may be resolved if we can reduce the moments to a handful of informative ones, such a choice is often not obvious.

This paper proposes a new simulation\hyp{}based estimation method, which we call {\em adversarial estimation}.
It is inspired by the {\em generative adversarial networks (GAN)}, a machine learning algorithm developed by \citet{goodfellow2014generative} to generate realistic images.
We adopt their adversarial framework to estimate the structural parameters that generate realistic economic data.
The proposed estimator achieves efficiency under correct specification and the parametric rate under misspecification.
Thus, our method is useful in applications where the likelihood is not computable but simulation is feasible and it can be a more efficient alternative to SMM.

The generative adversarial estimation framework is a minimax game between two components---the {\em discriminator} and the {\em generator}---over classification accuracy:
\[
	\min_{\{\text{\em generator}\}}\max_{\{\text{\em discriminator}\}}\text{\em classification accuracy}.
\]
The generator is an algorithm that produces the simulated data; its objective is to find a data\hyp{}generating process that confuses the discriminator.
The discriminator is a classification algorithm that distinguishes the observed data from the simulated data; it takes an observation as input and classifies whether it comes from observed data or simulated data; its objective is to maximize the accuracy of its classification.

In the original GAN, both the discriminator and the generator are given as neural networks (hence the name).
In this paper, we take the generator to be the structural model we intend to estimate and the discriminator to be an arbitrary classification algorithm (while our primary choice is a neural network).
For classification accuracy, we employ the cross\hyp{}entropy loss, following \citet{goodfellow2014generative}.%
\footnote{There are also other losses considered in the literature.
In machine learning, they concern high\hyp{}dimensional data such as images, sounds, and texts, %
and the Wasserstein distance has gained huge popularity for its ability to measure the distance of disjoint probability distributions.
It is also used in economic applications \citep{imbensGANs}.}

Interestingly, our framework casts a bridge between SMM and MLE.
When we use a logistic discriminator, the resulting estimator is asymptotically equivalent to optimally\hyp{}weighted SMM (\cref{sec:logistic}).
When we use the oracle discriminator, the resulting estimator is equivalent to MLE under the condition that the simulation sample size increases faster than the actual sample size.
Of particular interest is the middle case, in which the oracle discriminator is not available but a sufficiently rich discriminator capable of approximating it is used.
Under some conditions, the resulting estimator enjoys the desirable properties of both SMM and MLE: the user has the flexibility to choose moments if desired, a closed\hyp{}form likelihood is not required, and the asymptotic efficiency is attained.

We illustrate the theoretical properties of our estimator in simulations using
simple models. We show that the curvature of the classification accuracy is comparable to that of the log likelihood function for a suitable
choice of discriminator. In addition, we show that the estimator can achieve
the parametric rate under misspecification, and finally that compared to SMM,
the adversarial estimator suffers less from the small sample bias. We also showcase
the implementation of the method using a Roy Model with two occupations over
two time periods.

Using the adversarial estimation framework, we revisit investigation of the elderly's saving motives in \citet{dfj}.
Understanding different channels of saving motives is vital in evaluating social insurances.
We aim to disentangle three reasons to save: survival risk, medical expense risk, and bequest motive.
The structural model is dynamic and agents face heterogeneous risk by gender, age, health status, and permanent income. %
We demonstrate the capacity of adversarial estimation with two specifications: first with the inputs representing similar identifying variation as \citet{dfj}, and second the inputs augmented with gender and health.

The rest is organized as follows.
\cref{sec:2} defines the adversarial estimation framework.
\cref{sec:toy} illustrates the estimator with simple examples.
\cref{statistics} develops the asymptotic properties.
\cref{sec:application} applies the method to \citet{dfj}.

\section{Adversarial Estimation Framework} \label{sec:2}

The adversarial estimation has two main components: simulation and discrimination.
The simulation component is the same as other simulation\hyp{}based estimation methods, such as SMM or indirect
inference, but the discriminator component is new. The essence of the adversarial
framework is to find a parameter value for which the corresponding simulated
data is indistinguishable from the real data according to the discriminator. We
now describe each component in turn.

Suppose we have data $\{X_i\}_{i=1}^n$ drawn i.i.d.\ from an unknown distribution $P_0$.
Suppose we have a fully parametric model $\{P_\theta:\theta\in\Theta\}$ for which the likelihood is not tractable but simulation is feasible.%
\footnote{This is the case for many structural models in economics involving dynamic optimal decision making.}
Our target is the parameter $\theta$ that best describes the distribution of the data $P_0$ through the model $P_\theta$.

We formalize the simulation process as follows: for a given $\theta$, and a given sample size $m$, we can obtain a sample of simulated observations, $\{X_{i,\theta}\}_{i=1}^m$, according to model $P_\theta$ by taking draws $\{\tilde{X}_i\}_{i=1}^m$ from a known distribution $\tilde{P}_0$ and applying a transformation $T_\theta$ to them, $X_{i,\theta}=T_\theta(\tilde{X}_i)$.%
\footnote{If $P_\theta$ is discrete, e.g., Bernoulli with parameter $\theta$, we can generate $\tilde{X}_i\sim U[0,1]$ and apply the inverse transform sampling, e.g., $X_{i,\theta}=\mathbbm{1}\{\tilde{X}_i\geq 1-\theta\}$. %
}

For illustration, take the example of a normal location model with known variance $1$ and unknown mean $\theta$, $P_\theta=N(\theta,1)$.
When we generate a simulated observation $X_{i,\theta}$ from $P_\theta$, we can generate a standard normal observation $\tilde{X}_i\sim N(0,1)$ and convert it into $P_\theta$ through $X_{i,\theta}=\theta+\tilde{X}_i$.
We now turn to the discriminator. The discriminator is the novelty in the estimation framework and is the key component in the construction of the objective function for the adversarial estimator.
For some $\theta$ and $x$, consider the problem of assessing whether $x$ is from $P_\theta$ or $P_0$.
If $P_\theta$ is very different from $P_0$, it should be easy to distinguish realizations of $P_\theta$ from those of $P_0$. If they are close, it should be harder.
The idea, therefore, is to pick a classification algorithm that takes a value $x$ and predicts which distribution it came from, and to search for the value of $\theta$ for which the algorithm can classify the least.

If we had access to the probability density functions corresponding to
$P_0$ and $P_\theta$, it would be easy to assign the provenance of $x$ according to the
likelihood of $x$ for each distribution. This suggests
an estimation strategy based on the search of $\theta$ for which the probability that
any draw $X_{i,\theta}$ is drawn from $P_0$ versus $P_\theta$ is $0.5$. Since we do not have access
to the probability distributions, this strategy is infeasible. However, we can take
advantage of the availability of samples $\{X_i\}_{i=1}^n$ and $\{X_{i,\theta}\}_{i=1}^m$ to estimate the
extent to which, for a given $\theta$, these two distributions are different. In particular,
we use the fitted predictions of a discrete choice model (called the discriminator), where the dependent variable is $1$ if the data is real and $0$ if it is
simulated, and the explanatory variables are $X_i$ if the data is real,
and $X_{i,\theta}$ if it is simulated. When $\theta$ is a poor candidate to describe the observed
data, the predictions will be either close to $1$ or close to $0$. However,
as $\theta$ becomes a better candidate to describe the real data, the distribution of
the prediction will concentrate around $1/2$.

Formally, classification is defined as a function $D:\mathcal{X}\to[0,1]$ such that $D(x)$ represents the likelihood of $x$ being an actual observation; $D(x)=1$ means that $x$ is classified as ``actual'' with certainty; $D(x)=0$ that $x$ is classified as ``simulated'' with certainty.
Denote by $\mathcal{D}_n$ the class of classification functions we consider.
The dependence on $n$ allows us to use a richer classification algorithm as the sample size gets larger.
The choice of $\mathcal{D}_n$ is an important one for the researcher as it impacts the properties of the estimator. While any class of binary choice models would work, certain choices will have attractive properties, as we discuss below.

The {\em adversarial estimator} is defined by the following minimax problem:%
\footnote{Minimization and maximization need not be solved exactly (\cref{asm:neyman,asm:discriminator}).}
\[
	\hat{\theta}=\argmin_{\theta\in\Theta}\maxg_{\vphantom{\theta}D\in\mathcal{D}_n} \frac{1}{n}\sum_{i=1}^n\log D(X_i)+\frac{1}{m}\sum_{i=1}^m\log(1-D(X_{i,\theta})).
\]
Since $D$ is between $0$ and $1$, both $\log D$ and $\log(1-D)$ are nonpositive.
If $\{X_i\}$ and $\{X_{i,\theta}\}$ are very different from each other, the discriminator may be able to find $D$ that assigns $1$ on the support of $\{X_i\}$ and $0$ on the support of $\{X_{i,\theta}\}$, in which case the inner maximization attains the value of zero.
Meanwhile, regardless of the values of $\{X_i\}$ and $\{X_{i,\theta}\}$, the discriminator can always attain the classification accuracy of $2\log(1/2)$ by setting $D\equiv1/2$.%
\footnote{This is of course provided that a constant function $1/2$ is in $\mathcal{D}_n$, which is usually the case.}
In general, therefore, the inner maximization will give a number between $2\log(1/2)$ and $0$, and the closer it is to $2\log(1/2)$, the less able the discriminator is to classify the observations.

When we let $n$ and $m$ grow, we obtain the population counterpart of the problem
\[
	\min_{\theta\in\Theta}\max_{\vphantom{\theta}D\in\mathcal{D}_n}\,\mathbb{E}_{X_i\sim P_0}[\log D(X_i)]+\mathbb{E}_{X_{i,\theta}\sim P_\theta}[\log(1-D(X_{i,\theta}))].
\]
If there is no restriction on $\mathcal{D}_n$ (so any function $D:\mathcal{X}\to[0,1]$ is allowed), the optimum classifier for the population inner maximization is known to be
\[
	D_\theta(x)\coloneqq\frac{p_0(x)}{p_0(x)+p_\theta(x)},
\]
where $p_0$ and $p_\theta$ are the densities of $P_0$ and $P_\theta$ with respect to some common dominating measure \citep[Proposition 1]{goodfellow2014generative}.
We call this $D_\theta$ the {\em oracle discriminator}.
If the model is correctly specified, then $\theta_0$ is the unique solution to the outer minimization \citep[Theorem 1]{goodfellow2014generative}.
In the normal location model, if we assume $P_0=N(0,1)$, the oracle discriminator is given by
\(
	D_\theta(x)=\Lambda(\frac{1}{2}\theta^2-\theta x)
	=\Lambda(-\theta(x-\frac{1}{2}\theta))
\).
Since $\Lambda$ is a standard logistic cdf, $\Lambda(0)=1/2$, $\lim_{t\to\infty}\Lambda(t)\to 1$, and $\lim_{t\to-\infty}\Lambda(t)\to 0$. Therefore, if $\theta<0$, positive deviation of $x$ from $\theta/2$ is classified as more likely an actual observation, and negative deviation as less likely; if $\theta=0$, whatever value of $x$ has an equal chance of being actual.

The choice of $\mathcal{D}_n$ gives rise to a few special cases.
First, if we use the oracle discriminator $D_\theta$ in lieu of maximization, the resulting estimator for $\theta$ becomes efficient under correct specification and $m\gg n$ \citep[Proposition 1]{kmp2021}.
In the normal location model, we see that as $m\to\infty$, the oracle estimator solves
\[
	\hat{\theta}=\argmin_{\theta\in\Theta}\frac{1}{n}\sum_{i=1}^n\log\Lambda\Bigl(\frac{1}{2}\theta^2-\theta X_i\Bigr)+\mathbb{E}_\theta\Bigl[\log\Bigl(1-\Lambda\Bigl(\frac{1}{2}\theta^2-\theta X_{i,\theta}\Bigr)\Bigr)\Bigr].
\]
The FOC combined with the first\hyp{}order Taylor expansion of $\Lambda$ around $0$ yields
\[
	0=\frac{1}{n}\sum_{i=1}^n(\theta-X_i)\Bigl[1-\Lambda\Bigl(\frac{\theta^2}{2}-\theta X_i\Bigr)\Bigr]-\mathbb{E}_\theta\Bigl[(\theta-X_{i,\theta})\Lambda\Bigl(\frac{\theta^2}{2}-\theta X_{i,\theta}\Bigr)\Bigr]
	\approx\frac{1}{2n}\sum_{i=1}^n(\theta-X_i).
\]
Therefore, $\hat{\theta}$ is approximately the sample average, which is the MLE.

Second, if we use the logistic discriminator, the cross\hyp{}entropy loss can be interpreted as the log likelihood of the logistic regression where the actual observations are labeled $1$ and the simulated $0$.%
\footnote{When $n\neq m$, the two sets of observations are weighted differently.}
The resulting estimator for $\theta$ is then asymptotically equivalent to the optimally\hyp{}weighted SMM with moments $\mathbb{E}[X_i]$ under $m\gtrsim n$ (\cref{sec:logistic}).
In practice, we may use a sieve of discriminators that can represent oracle $D_\theta$ asymptotically, e.g., the sieve of neural networks or the sieve of logistic discriminators with an increasing number of polynomials of $X$.
In fact, we can regard $D_\theta$ as the nuisance parameter estimated in the inner maximization.
\cref{statistics} presents conditions under which the estimation of $D_\theta$ via nonparametric estimation makes the adversarial estimator efficient.

The estimation algorithm is given in \Cref{alg:adv}.
As is customary in simulation\hyp{}based methods, we use the same shocks $\{\tilde{X}_i\}_{i=1}^m$ to generate $\{X_{i,\theta}\}_{i=1}^m$ across different $\theta$.
Note that since the optimal discriminator depends on $\theta$, we need to solve the inner maximization for each candidate value of $\theta$.
When the transformation $T_\theta$ is differentiable in $\theta$, line \ref{alg:ln:6} can be performed by gradient descent.

\begin{algorithm}[t]
\caption{Adversarial estimation} \label{alg:adv}
\hspace*{\algorithmicindent} \textbf{Input:} Actual data $\{X_i\}_{i=1}^n$, distribution of random shocks $\tilde{P}_0$, structural transformation map $T_\theta$, simulation sample size $m$, discriminator $\mathcal{D}_n$ \\
\hspace*{\algorithmicindent} \textbf{Output:} Estimate $\hat{\theta}$
\begin{algorithmic}[1]
\State Sample $\tilde{X}_i\sim\tilde{P}_0$ for $i=1,\dots,m$ \Comment{$\{\tilde{X}_i\}_{i=1}^m$ is drawn this time only}
\State $\hat{\theta}\leftarrow$ initial value %
\State $X_{i,\hat{\theta}}\leftarrow T_{\hat{\theta}}(\tilde{X}_i)$ for $i=1,\dots,m$
\State $\mathbb{M}_{\hat{\theta}}\leftarrow \max_{D\in\mathcal{D}_n}\frac{1}{n}\sum_{i=1}^n\log D(X_i)+\frac{1}{m}\sum_{i=1}^m\log(1-D(X_{i,\hat{\theta}}))$.
\Repeat
	\State $\tilde{\theta}\leftarrow$ new candidate value \Comment{e.g., by gradient descent or simplex method} \label{alg:ln:6}
	\State $X_{i,\tilde{\theta}}\leftarrow T_{\tilde{\theta}}(\tilde{X}_i)$ for $i=1,\dots,m$ \Comment{same $\{\tilde{X}_i\}_{i=1}^m$ as above}
	\State $\mathbb{M}_{\tilde{\theta}}\leftarrow \max_{D\in\mathcal{D}_n}\frac{1}{n}\sum_{i=1}^n\log D(X_i)+\frac{1}{m}\sum_{i=1}^m\log(1-D(X_{i,\tilde{\theta}}))$.
	\If{$\mathbb{M}_{\tilde{\theta}}<\mathbb{M}_{\hat{\theta}}$} \label{alg:ln:9}
		\State $\hat{\theta}\leftarrow \tilde{\theta}$
		\State $\mathbb{M}_{\hat{\theta}}\leftarrow\mathbb{M}_{\tilde{\theta}}$
	\EndIf
\Until{$\hat{\theta}$ converges}
\end{algorithmic}
\end{algorithm}

The asymptotic distribution of the adversarial estimator depends on the choice of $\mathcal{D}_n$.
If the discriminator is logistic, the asymptotic variance of the adversarial estimator coincides with SMM (\cref{sec:logistic}).
If $\mathcal{D}_n$ is a nonparametric discriminator, under some conditions, the asymptotic variance will be a function of the score and Hessian of the likelihood (\cref{thm:theta:dist,asm:spec}).
When the likelihood is intractable, estimating this asymptotic variance formula is not an easy task; we recommend using bootstrap in which we resample both $\{X_i\}_{i=1}^n$ and $\{\tilde{X}_i\}_{i=1}^m$ with replacement, even though it might be computationally costly.

\section{Illustration with Simple Examples} \label{sec:toy}

We overview our estimator with simple examples to provide insights.
For the case where our method is of practical interest, see %
\cref{sec:application}.
The first example we consider is a logistic location model in which the mean is unknown and the variance is known.
We illustrate three points using this example:
(1) the adversarial estimator achieves parametric efficiency under correct specification;
(2) the adversarial estimator is asymptotically normal under model misspecification;
(3) the adversarial estimator is less sensitive to the curse of dimensionality compared to SMM. %
Next, we consider a Roy model with two occupations over two periods of time.
This example illustrates the whole procedure of estimation and inference in a case when the likelihood is intractable.

We write $\mathbb{L}_\theta\coloneqq-\frac{1}{2n}\sum_{i=1}^n\log p_\theta(X_i)$ for minus half the log likelihood and $\mathbb{M}_\theta(D)\coloneqq\frac{1}{n}\sum_{i=1}^n\log D(X_i)+\frac{1}{m}\sum_{i=1}^m\log(1-D(X_{i,\theta}))$ for the sample objective function.
Also, we let %
$\phi(x)\coloneqq\frac{1}{\sqrt{2\pi}}\exp(-x^2/2)$
be the standard normal pdf and $\Phi(x)\coloneqq\int_{-\infty}^x\phi(t)dt$ the standard normal cdf.

\subsection{Logistic Location Model}

\subsubsection{Efficiency} \label{sec:ill1}

Suppose we have $n$ i.i.d.\ observations $X_1,\dots,X_n$ from the standard logistic distribution with pdf $p_0(x)=\Lambda(x)(1-\Lambda(x))$.
Our structural model is the logistic distribution with unit scaling, i.e., $p_\theta(x)=\Lambda(x-\theta)(1-\Lambda(x-\theta))$.
The oracle discriminator is given by
\(
	D_\theta(x)
	=\Lambda(-\theta-2\log(1+e^{-x})+2\log(1+e^{-(x-\theta)}))
\).
The synthetic data is generated as $X_{i,\theta}=T_\theta(\tilde{X}_i)\coloneqq\theta+\tilde{X}_i$ where $\tilde{X}_i$ follows the standard logistic distribution.
We set $n=m=300$ and run 500 replications.

To yield a discriminator capable of representing the oracle, we consider
\(
	D_\lambda(x)=\Lambda(\lambda_0-2\log(1+e^{-x})+2\log(1+e^{-x+\lambda_1}))
\)
parameterized by $\lambda\in\mathbb{R}^2$.
This class of discriminator is ``correctly specified'' in the sense that the oracle discriminator is given by $\lambda_\theta\coloneqq(-\theta,\theta)^\top$; thus, it allows us to obliterate the approximation error of the nonparametric estimator for $D_\theta$ and focus on the aspects conducive to efficiency.
Nonetheless, we also present results with a nonparametric estimator, a shallow neural network, at the end of this section.

\begin{figure}
\centering
\begin{subfigure}[t]{0.32\textwidth}
\centering
\includegraphics[page=1]{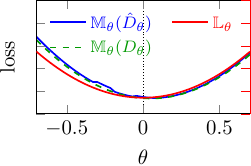}
\caption{Curvature of cross\hyp{}entropy loss and log likelihood.}
\label{fig:logistic:ortho}
\end{subfigure}
\ 
\begin{subfigure}[t]{0.32\textwidth}
\centering
\includegraphics[page=1]{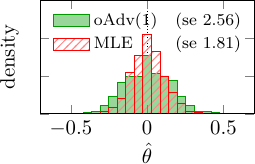}
\caption{Oracle adversarial estimator and MLE.}
\label{fig:logistic:1}
\end{subfigure}
\ 
\begin{subfigure}[t]{0.32\textwidth}
\centering
\includegraphics[page=2]{fig_logistic.pdf}
\caption{Adversarial estimator and MLE.}
\label{fig:logistic:2}
\end{subfigure}
\caption{The logistic location model. The curvature of oracle and estimated cross\hyp{}entropy losses matches the log likelihood (\subref{fig:logistic:ortho}). This makes the adversarial estimator comparable with MLE (\subref{fig:logistic:2}) and as good as the oracle estimator (\subref{fig:logistic:1}). The standard errors (se) are multiplied by $\sqrt{n}$. The vertical dots indicate the true parameter $\theta_0$.}
\label{fig:logistic}
\end{figure}

An intuition behind efficiency is that the curvature of $\mathbb{M}_\theta(\hat{D}_\theta)$ at $\theta_0$ is proportional to the Fisher information.
\Cref{fig:logistic:ortho} illustrates this point.
First, the curvature of $\mathbb{L}_\theta$ is a quarter of the Fisher information, and so is the curvature of the oracle loss $\mathbb{M}_\theta(D_\theta)$ (\cref{lem:lan:misspec}).
Second, the estimated loss $\mathbb{M}_\theta(\hat{D}_\theta)$ traces $\mathbb{M}_\theta(D_\theta)$ very well.
As a result, the curvature of $\mathbb{M}_\theta(\hat{D}_\theta)$ also becomes a quarter of the Fisher.
This is somewhat surprising given that $\hat{D}_\theta$ is estimated separately for each $\theta$ (\Cref{alg:adv}, line 8); the plot of $\mathbb{M}_\theta(\hat{D}_\theta)$ could have been zigzag if maximization was noisy each time.

An important practice that effects ``smooth'' $\mathbb{M}_\theta(\hat{D}_\theta)$ is to use a deterministic algorithm for the inner maximization.
Here, we use Matlab's \texttt{fminsearch} for maximization, which employs a deterministic algorithm.
However, if some stochastic optimization is to be used, we advise that the random seed be reset to the same value each time maximization is carried out.
For a logistic discriminator with differentiable $T_\theta$, \cref{sec:asm:neyman} shows that the estimated loss $\mathbb{M}_\theta(\hat{D}_\theta)$ will be smooth in $\theta$ if $\{\tilde{X}_i\}$ are fixed and the exact maximum is attained at the inner step for each $\theta$.

With the curvature of $\mathbb{M}_\theta(\hat{D}_\theta)$ matching $\mathbb{M}_\theta(D_\theta)$, the asymptotic variance of the adversarial estimator is $1+n/m$ times the inverse Fisher (\cref{thm:theta:efficiency}, \cref{sec:dist}).
In this example, the theoretical asymptotic standard deviation of MLE is 1.73 while of the adversarial estimator is 2.45, which are closely reproduced in \Cref{fig:logistic:1,fig:logistic:2}.

Similar results hold when $m$ is increased (figures omitted); the curvatures of $\mathbb{M}_\theta(D_\theta)$ and $\mathbb{M}_\theta(\hat{D}_\theta)$ match closely with $\mathbb{L}_\theta$, and the adversarial estimator gets closer to MLE.
For example, when $m=3{,}000$ (so $m=10 n$), the standard error of the adversarial estimator decreases to 2.00 (theoretically 1.94).

\begin{figure}
\centering
\begin{subfigure}[t]{0.32\textwidth}
\centering
\includegraphics[page=2]{fig_logistic_ortho.pdf}
\caption{$m=n$. The curve of $\mathbb{M}_\theta(\hat{D}_\theta)$ matches the oracle.}
\label{fig:logistic2:ortho1}
\end{subfigure}
\ 
\begin{subfigure}[t]{0.32\textwidth}
\centering
\includegraphics[page=1]{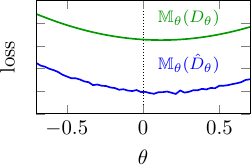}
\caption{$m=2n$. The level is off, but the curvature is right.}
\label{fig:logistic2:ortho2}
\end{subfigure}
\ 
\begin{subfigure}[t]{0.32\textwidth}
\centering
\includegraphics[page=1]{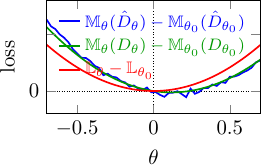}
\caption{$m=2n$. Demeaned (\subref{fig:logistic2:ortho2}) to highlight the curvature.}
\label{fig:logistic2:ortho3}
\end{subfigure}
\caption{Use of a neural network discriminator on the logistic location model.}
\label{fig:logistic2}
\end{figure}

To see how a nonparametric discriminator fares, we also try a shallow neural network discriminator.
The input is a one\hyp{}dimensional observation $X$; there are three nodes in one hidden layer with a hyperbolic tangent activation function; the output is a sigmoid function.
The neural network discriminator is trained for each $\theta$ using Matlab's \texttt{train} function, which is deterministic.
\Cref{fig:logistic2:ortho1} shows that the estimated loss $\mathbb{M}_\theta(\hat{D}_\theta)$ still gives a good approximation to $\mathbb{M}_\theta(D_\theta)$.
It is notable that as we increase $m$, the {\em level} of $\mathbb{M}_\theta(\hat{D}_\theta)$ becomes off from $\mathbb{M}_\theta(D_\theta)$, but the {\em curvature} is still correctly estimated (\Cref{fig:logistic2:ortho2}).
If we adjust the level, it becomes clear the curvature matches that of the log likelihood (\Cref{fig:logistic2:ortho3}).
According to our theory, the quality of the adversarial estimator hinges on the curvature of $\mathbb{M}_\theta(\hat{D}_\theta)$ but {\em not} on the level of $\mathbb{M}_\theta(\hat{D}_\theta)$ being close to $\mathbb{M}_\theta(D_\theta)$.
Thus, the resulting estimator is very close to the oracle (figures omitted).

We also examine if bootstrap works on the adversarial estimator.
The bootstrap consists of 500 replications with resampling both $\{X_i\}_{i=1}^n$ and $\{\tilde{X}_i\}_{i=1}^m$ with replacement but holding fixed the specification of the discriminator.
The bootstrap standard error for the logistic discriminator is $2.29$ and for the neural network discriminator $2.52$, which are close to the theoretical limit $2.45$.

\subsubsection{Normality under Misspecification} \label{sec:toy:misspec}

We now move to explore how the adversarial estimator behaves under misspecification.
Suppose we misspecify the model to be a normal location family with unit variance, $p_\theta(x)=\frac{1}{\sqrt{2\pi}}\exp(-\frac{(x-\theta)^2}{2})$, while the true distribution is still the standard logistic distribution that has variance $\pi^2/3\approx 3.3$.
The oracle discriminator is
\(
	D_\theta(x)
	=\Lambda(\log\sqrt{2\pi}-x+\tfrac{1}{2}(x-\theta)^2-2\log(1+e^{-x}))
\).
Here, we use the correctly specified discriminator
\(
	D_\lambda(x)=\Lambda(\lambda_0+\lambda_1 x+\lambda_2 x^2+\lambda_3\log(1+e^{-x}))
\)
parameterized by $\lambda\in\mathbb{R}^4$.

\begin{figure}
\centering
\begin{subfigure}[t]{0.32\textwidth}
\centering
\includegraphics[page=1]{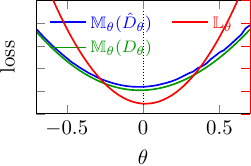}
\caption{Loss and quasi\hyp{}log likelihood.}
\label{fig:misspec:ortho}
\end{subfigure}
\ 
\begin{subfigure}[t]{0.32\textwidth}
\centering
\includegraphics[page=1]{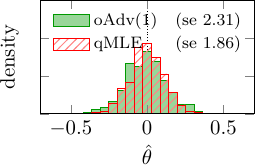}
\caption{Oracle adversarial estimator and quasi\hyp{}MLE.}
\label{fig:misspec:1}
\end{subfigure}
\ 
\begin{subfigure}[t]{0.32\textwidth}
\centering
\includegraphics[page=3]{fig_misspec.pdf}
\caption{Adversarial estimator and quasi\hyp{}MLE.}
\label{fig:misspec:2}
\end{subfigure}
\caption{The normally\hyp{}misspecified logistic location model. The adversarial estimator is comparable with quasi\hyp{}MLE.}
\label{fig:misspec}
\end{figure}

\Cref{fig:misspec:ortho} shows that the curvature of $\mathbb{L}_\theta$ is much steeper than $\mathbb{M}_\theta(D_\theta)$ due to misspecification (particularly to misspecification of variance).
However, the estimated loss $\mathbb{M}_\theta(\hat{D}_\theta)$ still estimates the curvature of the oracle loss correctly.
\Cref{fig:misspec:1} shows that the oracle adversarial estimator is approximately normal and comparable with quasi\hyp{}MLE.
A slight inflation of the variance is due to the fact that the adversarial estimator uses the synthetic data and gets affected by their randomness while quasi\hyp{}MLE does not.
\Cref{fig:misspec:2} shows that the adversarial estimator is very close to the oracle one.
The theoretical asymptotic standard deviation of the adversarial estimator implied by \cref{thm:theta:dist} is 2.27 while of quasi\hyp{}MLE is 1.81.
The observations for the increased synthetic sample size $m$ and for the neural network discriminator are analogous to \cref{sec:ill1} and hence omitted for space.

\subsubsection{Comparison with SMM} \label{sec:toy:smm}

Finally, we compare the adversarial estimator with SMM. As discussed, the adversarial estimator with a logistic discriminator is asymptotically equivalent to SMM.
However,
it is known that stacking up many moments yields poor finite\hyp{}sample performance of SMM.
To compare our estimator in this regard, the logistic location model is a particularly interesting one.
Unlike the normal distribution, the sample average is not a sufficient statistic for the mean of a logistic distribution.
Indeed, the collection of order statistics is known to be a minimal sufficient statistic.
Technically speaking, therefore, the higher\hyp{}order moments $\mathbb{E}[X_i^2]$, $\mathbb{E}[X_i^3]$, $\dots$ do contribute in identifying the mean.
In light of this, we carry out the following exercise.

For SMM, we consider matching (1) three moments $\mathbb{E}[X_i]$, $\mathbb{E}[X_i^2]$, $\mathbb{E}[X_i^3]$, (2) seven moments $\mathbb{E}[X_i],\dots,\mathbb{E}[X_i^7]$, and (3) eleven moments $\mathbb{E}[X_i],\dots,\mathbb{E}[X_i^{11}]$.
Since the optimally\hyp{}weighted SMM beats the unweighted SMM in all cases in our simulation, we only present the optimally\hyp{}weighted SMM for comparison; the weights are estimated with the real data.
For the adversarial estimator, we use the same set of moments as the inputs to the discriminator.
In particular, the discriminator is the logistic classifier of the form
\(
	D_\lambda(x)=\Lambda(\lambda_0+\lambda_1 x+\cdots+\lambda_d x^d)
\)
for $d=3,7,11$ parameterized by $\lambda\in\mathbb{R}^{1+d}$.
In contrast to the one in \cref{sec:ill1}, this discriminator is ``misspecified'' but is good enough to yield a reasonable estimator for $\theta$.
As discussed in \cref{sec:logistic}, the optimally\hyp{}weighted SMM is asymptotically equivalent to the adversarial estimator with this choice of the discriminator.
However, their finite\hyp{}sample properties are subject to debate.
For this exercise, we decrease the sample sizes to $n=m=200$ to emphasize the finite\hyp{}sample performance.

\begin{figure}
\centering
\begin{subfigure}[t]{0.32\textwidth}
\centering
\includegraphics[page=2]{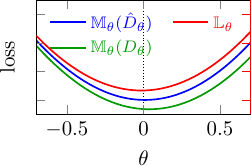}
\caption{Loss with 3 moments.}
\label{fig:curse:ortho:1}
\end{subfigure}
\ 
\begin{subfigure}[t]{0.32\textwidth}
\centering
\includegraphics[page=4]{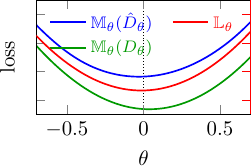}
\caption{Loss with 7 moments.}
\label{fig:curse:ortho:2}
\end{subfigure}
\ 
\begin{subfigure}[t]{0.32\textwidth}
\centering
\includegraphics[page=6]{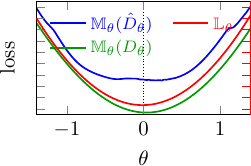}
\caption{Loss with 11 moments.}
\label{fig:curse:ortho:3}
\end{subfigure}
\caption{The logistic location model with increasing numbers of inputs. The curvature of the cross\hyp{}entropy loss is very close to the log likelihood up to 7 moments and is still good for 11 moments.}
\label{fig:curse:ortho}
\end{figure}

\Cref{fig:curse:ortho} shows the plots of the cross\hyp{}entropy loss and the log likelihood for varying numbers of inputs.
It is noteworthy that the curvature of the estimated loss $\mathbb{M}_\theta(\hat{D}_\theta)$ is very close to the oracle one up to seven moments.
We see nonnegligible deviation of the curvature for eleven moments, but as we see below, it is still sharp enough to yield a much better estimator than SMM.

The first row of \Cref{fig:curse} shows the histogram of the optimally\hyp{}weighted SMM.
The horizontal scales of the figures are adjusted to match the distribution of SMM; MLE is the same for all figures and serves as the reference point.
We see that the precision of SMM deteriorates quickly as the number of moments increases.
For eleven moments, the standard error is eight times as large as MLE.
The second row of \Cref{fig:curse} presents the adversarial estimator.
Even for seven inputs, the adversarial estimator is as tight as MLE, and for eleven moments, it is still comparable (three\hyp{}times larger standard error).
This shows that the adversarial estimator is less sensitive to the number of moments compared to SMM.
This can especially be an advantage when we do not know which moments to match.

\begin{figure}
\centering
\begin{subfigure}[t]{0.32\textwidth}
\centering
\includegraphics[page=8]{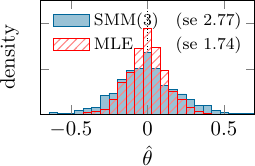}
\caption{SMM with 3 moments.}
\label{fig:curse:1}
\end{subfigure}
\ 
\begin{subfigure}[t]{0.32\textwidth}
\centering
\includegraphics[page=10]{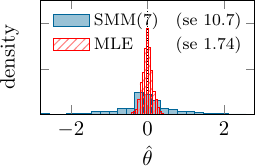}
\caption{SMM with 7 moments.}
\label{fig:curse:2}
\end{subfigure}
\ 
\begin{subfigure}[t]{0.32\textwidth}
\centering
\includegraphics[page=12]{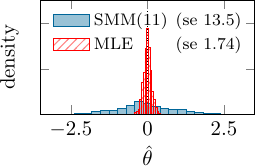}
\caption{SMM with 11 moments.}
\label{fig:curse:3}
\end{subfigure}
\\
\begin{subfigure}[t]{0.32\textwidth}
\centering
\includegraphics[page=14]{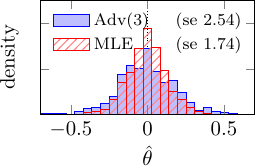}
\caption{Adversarial estimator with 3 moments.}
\label{fig:curse:1}
\end{subfigure}
\ 
\begin{subfigure}[t]{0.32\textwidth}
\centering
\includegraphics[page=16]{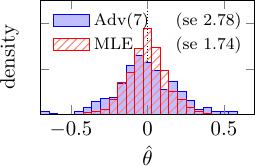}
\caption{Adversarial estimator with 7 moments.}
\label{fig:curse:2}
\end{subfigure}
\ 
\begin{subfigure}[t]{0.32\textwidth}
\centering
\includegraphics[page=18]{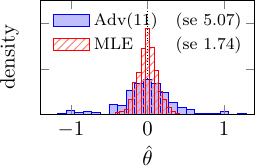}
\caption{Adversarial estimator with 11 moments.}
\label{fig:curse:3}
\end{subfigure}
\caption{The logistic location model with increasing numbers of inputs. Precision of the optimally\hyp{}weighted SMM rapidly deteriorates as the number of moments increases. The adversarial estimator is much less sensitive. The standard errors (se) are multiplied by $\sqrt{n}$.}
\label{fig:curse}
\end{figure}

We also note that since the moments are highly correlated, the estimation of the discriminator gives warnings of multicollinearity, but it does not impair the quality of the subsequent estimator $\hat{\theta}$.
This is insightful for a more general neural network discriminator since neural network weights are not identified uniquely.
This observation is in line with our theory that depends on the quality of the estimator $\hat{D}_\theta$ for $D_\theta$ but not on the quality of the estimator $\hat{\lambda}_\theta$ for $\lambda_\theta$.

The improvement of our method relative to SMM is analogous to the improvement of empirical likelihood to GMM \citep{i2002}.
SMM, like GMM, suffers from substantial bias when the number of moments is large; our method, like empirical likelihood, has better finite\hyp{}sample and large\hyp{}sample properties at the expense of computational cost.
The idea of both comes from treating the nuisance component as a kind of a nonparametric maximum likelihood problem.
Meanwhile, both SMM and GMM retain the advantage of simplicity to easily accommodate time series settings.

\subsection{The Roy Model}

We consider the following model of comparative advantage, for which the likelihood is not available under some configurations of the parameter values.
Suppose there are two sectors and two periods.
In each period, an agent chooses the sector to work in to maximize her present and discounted future expected wages.
The wage $w_{i1s}$ for agent $i$ in period 1 in sector $s$ is determined by $\log w_{i1s}=\mu_s+\varepsilon_{i1s}$, and the wage $w_{i2s}$ for agent $i$ in period 2 in sector $s$ by $\log w_{i2s}=\mu_s+\gamma_s\mathbbm{1}\{d_{i1}=s\}+\varepsilon_{i2s}$ where $d_{i1}$ is the sector choice of agent $i$ in period 1.
The parameter $\mu_s$ represents the base wage in sector $s$ and $\gamma_s$ the returns to experience in sector $s$.
The error terms are observable to the agent in respective periods (so she observes $\varepsilon_{i1\bigcdot}$ in period 1 and $\varepsilon_{i2\bigcdot}$ in period 2) and distribute as
\[
	\begin{bsmallmatrix}\vphantom{0^0_0}\varepsilon_{i11}\\\vphantom{0^0_0}\varepsilon_{i12}\\\vphantom{0^0_0}\varepsilon_{i21}\\\vphantom{0^0_0}\varepsilon_{i22}\end{bsmallmatrix}=N\!\left(
	\begin{bsmallmatrix}\vphantom{0^0_0}0\\\vphantom{0^0_0}0\\\vphantom{0^0_0}0\\\vphantom{0^0_0}0\end{bsmallmatrix},
	\begin{bsmallmatrix}
	\sigma_1^2&\rho_s\sigma_1\sigma_2&\rho_t\sigma_1^2&\rho_s\rho_t\sigma_1\sigma_2\\
	\rho_s\sigma_1\sigma_2&\sigma_2^2&\rho_s\rho_t\sigma_1\sigma_2&\rho_t\sigma_2^2\\
	\rho_t\sigma_1^2&\rho_s\rho_t\sigma_1\sigma_2&\sigma_1^2&\rho_s\sigma_1\sigma_2\\
	\rho_s\rho_t\sigma_1\sigma_2&\rho_t\sigma_2^2&\rho_s\sigma_1\sigma_2&\sigma_2^2
	\end{bsmallmatrix}\right)\!.
\]
Observable to us is the quartet $X_i=(\log w_{i1},d_{i1},\log w_{i2},d_{i2})$ of realized log wages and sector choices in both periods.
They are functions of above variables by $w_{i1}=w_{i1d_{i1}}$, $d_{i1}=\argmax_{s\in\{1,2\}}w_{i1s}+\beta\mathbb{E}[w_{i2}\mid d_{i1}=s]$, $w_{i2}=w_{i2d_{i2}}$, and $d_{i2}=\argmax_{s\in\{1,2\}}w_{i2s}$ where $\beta$ is the discount factor.
We fix $\beta=0.9$, so $\beta$ is not a free parameter.

\subsubsection{Comparison with MLE} \label{sec:toy:mle}

As a first exercise, we show that the adversarial estimator has a computational advantage over MLE. To this end,
we fix $\rho_t=0$ to have a tractable likelihood.
Thus, the parameter of interest is $\theta=(\mu_1,\mu_2,\gamma_1,\gamma_2,\sigma_1,\sigma_2,\rho_s)$.
The true value is $\theta_0=(1.8,2,0.5,0,1,1,0.5)$.
We set the sample sizes at $n=m=300$.

Although the likelihood is available, the correct functional form of $D_\theta$ is not easy to derive.
So, we skip the correctly specified discriminator and use the neural network discriminator for the feasible adversarial estimator.
The neural network has one hidden layer with 10 nodes with a hyperbolic tangent activation function.
The input is $X_i$ without transformation. The output layer uses a sigmoid function.

\begin{figure}
\centering
\begin{subfigure}[t]{0.32\textwidth}
\centering
\includegraphics[page=1]{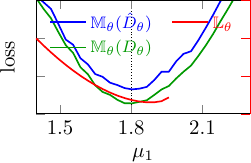}
\caption{Loss against $\mu_1$.}
\label{fig:roy:ortho:mu1}
\end{subfigure}
\ 
\begin{subfigure}[t]{0.32\textwidth}
\centering
\includegraphics[page=2]{fig_roy_ortho.pdf}
\caption{Loss against $\mu_2$.}
\label{fig:roy:ortho:mu2}
\end{subfigure}
\ 
\begin{subfigure}[t]{0.32\textwidth}
\centering
\includegraphics[page=3]{fig_roy_ortho.pdf}
\caption{Loss against $\gamma_1$.}
\label{fig:roy:ortho:gam1}
\end{subfigure}
\\
\begin{subfigure}[t]{0.32\textwidth}
\centering
\includegraphics[page=4]{fig_roy_ortho.pdf}
\caption{Loss against $\gamma_2$.}
\label{fig:roy:ortho:gam2}
\end{subfigure}
\ 
\begin{subfigure}[t]{0.32\textwidth}
\centering
\includegraphics[page=5]{fig_roy_ortho.pdf}
\caption{Loss against $\sigma_1$.}
\label{fig:roy:ortho:sig1}
\end{subfigure}
\ 
\begin{subfigure}[t]{0.32\textwidth}
\centering
\includegraphics[page=6]{fig_roy_ortho.pdf}
\caption{Loss against $\sigma_2$.}
\label{fig:roy:ortho:sig2}
\end{subfigure}
\caption{The loss for the Roy model. In the region where $\mathbb{L}_\theta$ is not plotted, the real data $X$ is not supported on the corresponding model $P_\theta$, so $\mathbb{L}_\theta=\infty$. The figure for $\rho_s$ is omitted.}
\label{fig:roy:ortho}
\end{figure}

Note that if $w_{i11}+\beta\mathbb{E}[w_{i2}\mid d_{i1}=1]<\beta\mathbb{E}[w_{i2}\mid d_{i1}=2]$, there is no way that agent $i$ chooses sector 1 in period 1.
Therefore, if we see a pair $(w_{i1},d_{i1})=(w_{i11},1)$ that satisfies this inequality for a particular $\theta$, this observation is not supported by $P_\theta$.
This is indeed a common phenomenon.
\Cref{fig:roy:ortho} plots the loss and the log likelihood against each parameter, holding all other parameters to the truth.
The range of the figures reflects the range of MLE and the adversarial estimator.
In this ``relevant'' region, we see that $\mathbb{L}_\theta$ sometimes breaks off; this is because the discontinued part does not support the real data so $\mathbb{L}_\theta$ is infinity.

Aside from possible inefficiency, this is not a problem for MLE insofar as the likelihood maximizer can be found.
However, there may be a trouble when the initial value of $\theta$ does not support the real data.
In fact, if we do not pick the initial value carefully, Matlab's \texttt{fminsearch} wanders around the unsupported region and returns a meaningless value after the evaluation count hits the limit. %
Meanwhile, \Cref{fig:roy:ortho} indicates that such a problem does not occur for the cross\hyp{}entropy loss; indeed, $\mathbb{M}_\theta(D_\theta)$ extends a nice curve throughout the ``unsupported'' region.
The key is in the robustness of the sample Jensen--Shannon divergence
\[
	\frac{1}{2}\mathbb{M}_\theta(D_\theta)=\frac{1}{2n}\sum_{i=1}^n\log\frac{p_0(X_i)}{p_0(X_i)+p_\theta(X_i)}+\frac{1}{2m}\sum_{i=1}^m\log\frac{p_\theta(X_{i,\theta})}{p_0(X_{i,\theta})+p_\theta(X_{i,\theta})}.
\]
When a single observation $X_i$ is not on the support of $p_\theta$, the corresponding fraction is $1$, which does not ruin the sum so we can still calculate a meaningful distance using remaining observations; hence the curve continues.
Moreover, even if the entire observations $\{X_i\}$ go outside the support, the divergence still works as long as (some of) synthetic data are on the support of $p_0$ and the second sum is informative.
It is only when both the entire real sample $\{X_i\}$ and the synthetic sample $\{X_{i,\theta}\}$ are outside the supports of $p_\theta$ and $p_0$ respectively that the Jensen--Shannon divergence gets fixated at $0$ and loses guidance on $\theta_0$.%
\footnote{If the supports of $p_0$ and $\{p_\theta\}$ are fully disjoint, the Jensen--Shannon projection $\theta_0$ is not defined.}
This is the intuition why the adversarial estimator does not suffer from the support issue in the Roy model.
We can also see this as a virtue of estimating the likelihood ratio as opposed to the raw likelihood.

This is not to say, however, that the cross\hyp{}entropy loss works for any kind of disjointly supported distributions.
When GAN is used to generate high\hyp{}dimensional data like images, the cross\hyp{}entropy loss is known to be very hard to train, partly because of the severe disjoint support problem.
An alternative is the Wasserstein loss that puts a nice gradation on the distance between completely disjoint distributions.

\Cref{fig:roy:ortho} also illustrates that, despite having discrete observables (sector choices), the objective functions are very smooth thanks to continuous observables (wages), so there is no need for smoothing in contrast to \cref{sec:toy:misspec}.
The resulting estimators are comparable with MLE just as in the previous examples (\Cref{fig:roy}).

\begin{figure}
\centering
\begin{subfigure}[t]{0.32\textwidth}
\centering
\includegraphics[page=1]{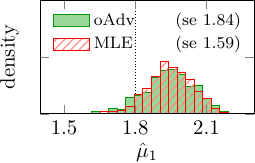}
\caption{Oracle adversarial estimator for $\mu_1$.}
\label{fig:roy:o:mu1}
\end{subfigure}
\ 
\begin{subfigure}[t]{0.32\textwidth}
\centering
\includegraphics[page=3]{fig_roy.pdf}
\caption{Oracle adversarial estimator for $\mu_2$.}
\label{fig:roy:o:mu2}
\end{subfigure}
\ 
\begin{subfigure}[t]{0.32\textwidth}
\centering
\includegraphics[page=5]{fig_roy.pdf}
\caption{Oracle adversarial estimator for $\gamma_1$.}
\label{fig:roy:o:gam1}
\end{subfigure}
\\
\begin{subfigure}[t]{0.32\textwidth}
\centering
\includegraphics[page=2]{fig_roy.pdf}
\caption{Adversarial estimator for $\mu_1$.}
\label{fig:roy:mu1}
\end{subfigure}
\ 
\begin{subfigure}[t]{0.32\textwidth}
\centering
\includegraphics[page=4]{fig_roy.pdf}
\caption{Adversarial estimator for $\mu_2$.}
\label{fig:roy:mu2}
\end{subfigure}
\ 
\begin{subfigure}[t]{0.32\textwidth}
\centering
\includegraphics[page=6]{fig_roy.pdf}
\caption{Adversarial estimator for $\gamma_1$.}
\label{fig:roy:gam1}
\end{subfigure}
\caption{The oracle adversarial estimator and the adversarial estimator for the Roy model with $\rho_t=0$. Figures for other parameters are omitted.}
\label{fig:roy}
\end{figure}

\subsubsection{Case with Intractable Likelihood} \label{sec:toy:case}

\begin{figure}
\centering
\begin{subfigure}[t]{0.244\textwidth}
\centering
\includegraphics[page=1]{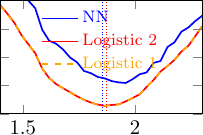}
\caption{Loss against $\mu_1$.}
\label{fig:case3:ortho:mu1}
\end{subfigure}
\begin{subfigure}[t]{0.244\textwidth}
\centering
\includegraphics[page=2]{fig_case3_ortho.pdf}
\caption{Loss against $\mu_2$.}
\label{fig:case3:ortho:mu2}
\end{subfigure}
\begin{subfigure}[t]{0.244\textwidth}
\centering
\includegraphics[page=3]{fig_case3_ortho.pdf}
\caption{Loss against $\gamma_1$.}
\label{fig:case3:ortho:gam1}
\end{subfigure}
\begin{subfigure}[t]{0.244\textwidth}
\centering
\includegraphics[page=4]{fig_case3_ortho.pdf}
\caption{Loss against $\gamma_2$.}
\label{fig:case3:ortho:gam2}
\end{subfigure}
\\
\begin{subfigure}[t]{0.244\textwidth}
\centering
\includegraphics[page=5]{fig_case3_ortho.pdf}
\caption{Loss against $\sigma_1$.}
\label{fig:case3:ortho:sig1}
\end{subfigure}
\begin{subfigure}[t]{0.244\textwidth}
\centering
\includegraphics[page=6]{fig_case3_ortho.pdf}
\caption{Loss against $\sigma_2$.}
\label{fig:case3:ortho:sig2}
\end{subfigure}
\begin{subfigure}[t]{0.244\textwidth}
\centering
\includegraphics[page=7]{fig_case3_ortho.pdf}
\caption{Loss against $\rho_t$.}
\label{fig:case3:ortho:rhot}
\end{subfigure}
\begin{subfigure}[t]{0.244\textwidth}
\centering
\includegraphics[page=8]{fig_case3_ortho.pdf}
\caption{Loss against $\rho_s$.}
\label{fig:case3:ortho:rhos}
\end{subfigure}
\caption{The first logistic loss does not identify $\rho_t$, while the second logistic does. The neural network loss indicates orthogonality, albeit a bit rough.}
\label{fig:case3:ortho}
\end{figure}

Now, we illustrate the whole procedure of estimation and inference using the Roy model with intractable likelihood.
Let us consider the same model as \cref{sec:toy:mle} without assuming $\rho_t=0$%
, so the parameter of interest is $\theta=(\mu_1,\mu_2,\gamma_1,\gamma_2,\sigma_1,\sigma_2,\rho_t,\rho_s)$.
The true values are the same as before.
We first pre\hyp{}estimate the model with a logistic discriminator and then estimate it with a neural network discriminator using the logistic estimator as the initial value.
Since it is naturally speculated that identification comes from the moments of the log wages, 
we consider the logistic discriminator of the form $D_\lambda(\log w_1,d_1,\log w_2,d_2)=\Lambda(\lambda_0+\lambda_1\log w_1+\lambda_2 d_1+\lambda_3\log w_2+\lambda_4 d_2+\lambda_5(\log w_1)^2+\lambda_6(\log w_2)^2+\lambda_7\log w_1\log w_2)$.

As the curvature of the second logistic loss is quite sharp, we may in practice stop here and go with the logistic estimator.
For illustration, we move on to the neural network discriminator with the same configuration as \cref{sec:toy:mle}.
The loss is plotted as the blue line in \Cref{fig:case3:ortho}.
The vertical blue dotted lines indicate the neural network estimator.%
\footnote{Note that the global minimizer is not the same as the local minimizers of the figures since the parameters are fixed at the logistic estimator.} %

\begin{figure}
\centering
\begin{subfigure}[t]{0.244\textwidth}
\centering
\includegraphics[page=1]{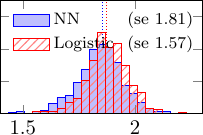}
\caption{$\hat{\mu}_1$}
\label{fig:case3:mu1}
\end{subfigure}
\begin{subfigure}[t]{0.244\textwidth}
\centering
\includegraphics[page=2]{fig_case3.pdf}
\caption{$\hat{\mu}_2$}
\label{fig:case3:mu2}
\end{subfigure}
\begin{subfigure}[t]{0.244\textwidth}
\centering
\includegraphics[page=3]{fig_case3.pdf}
\caption{$\hat{\gamma}_1$}
\label{fig:case3:gam1}
\end{subfigure}
\begin{subfigure}[t]{0.244\textwidth}
\centering
\includegraphics[page=4]{fig_case3.pdf}
\caption{$\hat{\gamma}_2$}
\label{fig:case3:gam2}
\end{subfigure}
\\
\begin{subfigure}[t]{0.244\textwidth}
\centering
\includegraphics[page=5]{fig_case3.pdf}
\caption{$\hat{\sigma}_1$}
\label{fig:case3:sig1}
\end{subfigure}
\begin{subfigure}[t]{0.244\textwidth}
\centering
\includegraphics[page=6]{fig_case3.pdf}
\caption{$\hat{\sigma}_2$}
\label{fig:case3:sig2}
\end{subfigure}
\begin{subfigure}[t]{0.244\textwidth}
\centering
\includegraphics[page=7]{fig_case3.pdf}
\caption{$\hat{\rho}_t$}
\label{fig:case3:rhot}
\end{subfigure}
\begin{subfigure}[t]{0.244\textwidth}
\centering
\includegraphics[page=8]{fig_case3.pdf}
\caption{$\hat{\rho}_s$}
\label{fig:case3:rhos}
\end{subfigure}
\caption{The bootstrap samples and bootstrap standard errors (multiplied by $\sqrt{n}$).}
\label{fig:case3}
\end{figure}

Next, we use bootstrap to compute the standard errors.
We resample both the actual data $\{X_i\}_{i=1}^n$ and the simulation shocks $\{\tilde{X}_i\}_{i=1}^m$ with replacement, pre\hyp{}estimate the model with the second logistic discriminator, and then estimate the model with the neural network discriminator.
\Cref{fig:case3} shows the bootstrap samples of the logistic estimator (red) and the neural network estimator (blue).
Due to some degree of roughness of the neural network loss and difficulty in identifying the global minimum, the neural network sample is overall more dispersed than the logistic estimator.
However, the neural network estimator is still comparable and sometimes produces even tighter estimates (for $\sigma_1$ and $\sigma_2$).
Note that the neural network takes as inputs the raw quartet but not the higher\hyp{}order moments.
So, the neural network with one hidden layer of 10 nodes ``figures out'' the correct moments to match and produces an estimator comparable with (and partly superior to) the logistic discriminator whose inputs were deliberately chosen.

\Cref{tbl:case3} presents the estimates and the standard errors (not multiplied by $\sqrt{n}$).
Along with the adversarial estimator, we present the results of SMM.
SMM matches the same seven moments as the inputs to the second logistic discriminator: $\mathbb{E}[\log w_{i1}]$, $\mathbb{E}[d_{i1}]$, $\mathbb{E}[\log w_{i2}]$, $\mathbb{E}[d_{i2}]$, $\mathbb{E}[(\log w_{i1})^2]$, $\mathbb{E}[(\log w_{i2})^2]$, and $\mathbb{E}[\log w_{i1}\log w_{i2}]$.
The optimal weights are estimated with the actual data.
We see that the adversarial estimator with the logistic discriminator is slightly more precise than SMM.
\begin{table}
\caption{Estimates and bootstrap standard errors for the Roy model for one replication.}
\label{tbl:case3}
\centering
\small
\begin{tabular}{lcccccccc}
\toprule \midrule
& $\mu_1$ & $\mu_2$ & $\gamma_1$ & $\gamma_2$ & $\sigma_1$ & $\sigma_2$ & $\rho_t$ & $\rho_s$ \\
\midrule
Logistic $D$ & 1.87 & 1.82 & 0.31 & 0.09 & 1.05 & 0.94 & $\mathllap{-}$0.06 & 0.49 \\
& (0.09)& (0.11) & (0.11) & (0.13) & (0.09) & (0.14) & (0.09) & (0.12) \\
Neural network $D$ & 1.86 & 1.81 & 0.35 & 0.01 & 1.04 & 1.28 & 0.07 & 0.50 \\
& (0.10) & (0.15) & (0.15) & (0.18) & (0.08) & (0.13) & (0.14) & (0.14) \\
SMM & 1.88 & 1.81 & 0.31 & 0.09 & 1.04 & 0.95 & $\mathllap{-}$0.05 & 0.49 \\
& (0.10) & (0.13) & (0.13) & (0.17) & (0.11) & (0.14) & (0.09) & (0.15) \\
\midrule
Truth & 1.80 & 2.00 & 0.50 & 0.00 & 1.00 & 1.00 & 0.00 & 0.50 \\
\bottomrule
\end{tabular}
\end{table}

\subsection{Challenges of the Adversarial Estimator}

Not every aspect of our method is superior to alternatives.
First, the theoretical results in this paper do not cover time series data.
The Roy model has a dynamic choice of individuals, but we have i.i.d.\ observations of individuals.
This is not to say that the adversarial framework cannot be extended thereto, but it would require a careful design of the discriminator to incorporate the structure of serial correlation.

Second, the adversarial estimator can be time\hyp{}consuming.
A logistic discriminator is as fast as II, but a neural network discriminator can take a long time to train.
In the logistic location model, both MLE and the adversarial estimator with a logistic discriminator take less than a second, while the adversarial estimator with a neural network discriminator takes about 30 seconds on a laptop without a GPU or parallelization.
For this, we recommend pre\hyp{}estimation with a logistic discriminator or other existing methods to start with a good initial value.

The third drawback is a possible roughness of the loss surface.
As seen in \cref{sec:ill1}, a logistic discriminator tends to yield a very smooth objective function (\Cref{fig:logistic}) while a neural network discriminator may sometimes get bumpy and have spurious local minima (\Cref{fig:logistic2}).
Some degree of roughness can be smoothed with the choice of a training method or an increased number of iterations; additionally, we can estimate the discriminator several times and take their average and/or use an optimization method tailored for noisy functions.
If the initial value is good enough, we may also employ grid search in the neighborhood to skip estimation of the gradient.
At any rate, we recommend plotting the loss surface before outer minimization.

Fourth, being comparable with MLE, the asymptotic variance of the adversarial estimator depends on the score and Hessian (\cref{thm:theta:dist}), which is not easy to compute given intractable likelihood.
Therefore, we may resort to resampling methods like bootstrap to obtain a variance estimator, which can cost additional time.

\section{Statistical Properties} \label{statistics}

This section derives the asymptotic properties of the adversarial estimator.
For more general results, we refer the reader to our earlier version \citep{kmp2022}.

Let $\tilde{X}_i\sim\tilde{P}_0$ be a common random shock used in simulation.
The simulated observation $X_{i,\theta}\sim P_\theta$ is then constructed by transforming $\tilde{X}_i$ through a map, $X_{i,\theta}=T_\theta(\tilde{X}_i)$.
For a function $f$, the sample averages of $f(X_i)$ and $f(X_{i,\theta})$ are denoted by $\mathbb{P}_0 f\coloneqq\frac{1}{n}\sum_{i=1}^n f(X_i)$ and $\mathbb{P}_\theta f\coloneqq\frac{1}{m}\sum_{i=1}^m f(X_{i,\theta})$.
Their population counterparts are denoted as $P_0 f\coloneqq\int f(x)dP_0$ and $P_\theta f\coloneqq\int f(x)dP_\theta$.
We denote the population objective function by $M_\theta(D)\coloneqq P_0\log D+P_\theta\log(1-D)$ as well as the previously defined sample objective function $\mathbb{M}_\theta(D)\coloneqq\mathbb{P}_0\log D+\mathbb{P}_\theta\log(1-D)$.
We also define the distance on $\Theta$ by $h(\theta_1,\theta_2)\coloneqq\sqrt{\int(\sqrt{p_{\theta_1}}-\sqrt{p_{\theta_2}})^2}$.

Suppose that observables can be written as $X_i=(y_i,x_i)$ where $\theta$ affects only the conditional distribution of $y_i$ given $x_i$.
Such $x_i$ is called the covariate.
In the maximum likelihood literature, it is known that an efficient estimator is obtained by maximizing the conditional likelihood of $y_i$ given $x_i$, so the marginal distribution of $x_i$ can be left unspecified.
The same observation holds true in the adversarial framework.
Namely, the oracle discriminator $D_\theta$ does not depend on the marginal distribution of $x_i$, so the distributions $P_0$ and $P_\theta$ can be regarded as specifying only the conditional distribution of $y_i$ given $x_i$.
In our theory, we save notational complexity by allowing this implicitly.
One possible complication this might bring is the method to draw covariates for the simulated data.
In \cref{sec:toy:misspec}, we set $n=m$ and use the same sets of covariates in the actual data. %
Another possibility is to bootstrap the covariates.

\subsection{Consistency}

The adversarial estimator is consistent if the estimated loss $\mathbb{M}_\theta(\hat{D}_\theta)$ converges uniformly to the oracle loss $\mathbb{M}_\theta(D_\theta)$ and $\hat{\theta}$ finds a global minimizer.
As the maximized cross\hyp{}entropy loss is effectively bounded between $2\log(1/2)$ and $0$, uniform convergence on $\Theta$ is not an unreasonable assumption.

\begin{thm}[Consistency of generator] \label{thm:theta:consistency}
Suppose that
for every open $G\subset\Theta$ containing $\theta_0$, we have $\inf_{\theta\notin G}M_\theta(D_\theta)>M_{\theta_0}(D_{\theta_0})$, that
$\{\log D_\theta:\theta\in\Theta\}$ and $\{\log(1-D_\theta)\circ T_\theta:\theta\in\Theta\}$ are $P_0$- and $\tilde{P}_0$\hyp{}Glivenko--Cantelli respectively, that $\sup_{\theta\in\Theta}|\mathbb{M}_\theta(\hat{D}_\theta)-\mathbb{M}_\theta(D_\theta)|\to 0$ in probability, and that $\hat{\theta}$ satisfies
\(
	\mathbb{M}_{\hat{\theta}}(\hat{D}_{\hat{\theta}})\leq\inf_{\theta\in\Theta}\mathbb{M}_\theta(\hat{D}_\theta)+o_P^\ast(1)
\).
Then, %
$h(\hat{\theta},\theta_0)\to0$ in probability.
\end{thm}

This theorem does not assume that the generative model is parametric, so it also applies to possibly ``nonparametric'' generators.

\subsection{Rate of Convergence} \label{sec:rate}

To obtain a rate of convergence of the generator, we assume that the structural model is parametric.

\begin{asm}[Parametric generative model] \label{asm:maximal}
$\Theta$ is (a subset of) a Euclidean space;
$p_\theta$ is differentiable in $\theta$ at every $\theta\in\Theta$ for every $x\in\mathcal{X}$ with the derivative continuous in both $x$ and $\theta$;
the maximum eigenvalue of the Fisher information $I_\theta=P_\theta\dot{\ell}_\theta\dot{\ell}_\theta^\top$ is bounded uniformly in $\theta\in\Theta$;
the minimum eigenvalue of $I_\theta$ is bounded away from $0$ uniformly in $\theta\in\Theta$.
The same is assumed for the ``inverted'' structural model $\tilde{\mathcal{P}}_\theta=\{((p_0/p_\theta)\circ T_\theta)\tilde{p}_0:\theta\in\Theta\}$.
\end{asm}

To establish asymptotic results in terms of $n$, we next assume that the synthetic sample size $m$ grows as fast as $n$.
It is allowed (but not required) that $m$ diverges faster than $n$, in which case $n/m$ converges to $0$.

\begin{asm}[Growing synthetic sample size] \label{asm:m}
$n/m$ converges.
\end{asm}

The next assumption ensures that the estimation procedure finds a good minimum and that the derivative of the estimated loss converges to that of the oracle.
The first property hinges on the estimation procedure employed, the tolerance level, etc.
The second property is used in semiparametric $M$\hyp{}estimation to obtain a regular estimator orthogonal to nuisance estimation \citep[e.g.,][]{ks1993}. %
We revisit the plausibility of this condition in \cref{sec:asm:neyman}.

\begin{asm}[Approximately minimizing generator and orthogonality] \label{asm:neyman}
There exists a sequence of open balls $G_n\coloneqq\{\theta\in\Theta:h(\theta,\theta_0)<\eta_n\}$ such that $\eta_n\sqrt{n}\to\infty$,
\(
	\mathbb{M}_{\hat{\theta}}(\hat{D}_{\hat{\theta}})\leq\inf_{\theta\in G_n}\mathbb{M}_\theta(\hat{D}_\theta)+o_P^\ast(n^{-1})
\),
and
\(
	\inf_{\theta\in G_n}[\mathbb{M}_{\hat{\theta}}(\hat{D}_{\hat{\theta}})-\mathbb{M}_\theta(\hat{D}_\theta)]-[\mathbb{M}_{\hat{\theta}}(D_{\hat{\theta}})-\mathbb{M}_\theta(D_\theta)]=o_P^\ast(n^{-1})
\).
\end{asm}

The next assumption consists of three parts.
The first is a stronger identification condition than in \cref{thm:theta:consistency} that ensures a quadratic curvature at $\theta_0$; this is implied by the positive definiteness of $\tilde{I}_{\theta_0}$ in \cref{asm:dqm}.
The second assumes a degree of smoothness needed for $T_\theta$; this is trivial with $\tau_n\equiv 0$ if $n/m\to 0$ or \cref{asm:spec} holds; otherwise, if $T_\theta$ and $D_{\theta_0}$ are differentiable in $\theta$ and $x$ respectively and \cref{asm:dqm} holds, there is a closed\hyp{}form expression for $\tau_n$, which we derive in \cref{sec:asm:misspec}.
Third, we need that $P_0$ is ``close enough'' to $P_{\theta_0}$ in the sense that
convergence of $\theta$ to $\theta_0$
takes place on the support of $P_0$; this is also trivial under \cref{asm:spec}.

\begin{asm}[Smooth synthetic data generation and overlapping support] \label{asm:misspec}
There exists open $G\subset\Theta\subset\mathbb{R}^k$ containing $\theta_0$ in which $M_\theta(D_\theta)-M_{\theta_0}(D_{\theta_0})\gtrsim h(\theta,\theta_0)^2$.
There exists a sequence of $\mathbb{R}^k$\hyp{}valued random variables $\tau_n$ such that for every compact $K\subset\Theta$,
\(
	\sqrt{\tfrac{n}{m}}\sup_{h\in K}|\sqrt{m}(\tilde{\mathbb{P}}_0-\tilde{P}_0)(\sqrt{n}[\log(1-D_{\theta_0})\circ T_{\theta+h/\sqrt{n}}-\log(1-D_{\theta_0})\circ T_{\theta_0}]-h^\top\tau_n)|
	=o_P^\ast(1+\tfrac{n}{m})
\).
Also,
\(
	h(\theta,\theta_0)^2=O(\int D_{\theta_0}(\sqrt{p_{\theta_0}}-\sqrt{p_\theta})^2)
\)
as $\theta\to\theta_0$.
\end{asm}

\begin{thm}[Rate of convergence of generator] \label{thm:5}
Under \cref{asm:m,asm:neyman,asm:maximal,asm:misspec},
$h(\hat{\theta},\theta_0)=O_P^\ast(n^{-1/2})$.
\end{thm}

\subsubsection{On \cref{asm:neyman}} \label{sec:asm:neyman}

The second condition of \cref{asm:neyman}, which we call orthogonality, is essential in the rate of convergence for $\hat{\theta}$ in \cref{thm:5}.
Even in the best scenario, we can only expect $\mathbb{M}_\theta(\hat{D}_\theta)-\mathbb{M}_\theta(D_\theta)=O_P(n^{-1})$, so the convergence of $\hat{D}_\theta$ alone does not grant orthogonality.
The key to satisfying it is, therefore, some extent of the convergence of the {\em derivative} of $\hat{D}_\theta$ with respect to $\theta$ to that of $D_\theta$.
Note that this is different from the derivative of $\hat{D}_\theta$ with respect to $x$, so it does not follow from the convergence of the derivative of a nonparametrically estimated function.
Rather, %
it is the structure of the nested optimization that brings about orthogonality.

Take the logistic discriminator $D_\lambda(x)=\Lambda(x^\top\lambda)$ as considered in \cref{sec:toy}.
We can check that orthogonality holds if the following conditions are met. %
Let $\mathbb{E}_n f(X)\coloneqq\frac{1}{n}\sum_{i=1}^n f(X_i)$ and $\mathbb{E}_m f(X_\theta)\coloneqq\frac{1}{m}\sum_{i=1}^m f(X_{i,\theta})$ and denote the differentiation with respect to a row vector $\theta^\top$ by a dot, e.g., $\dot{\lambda}_\theta=\frac{\partial}{\partial\theta^\top}\lambda_\theta$.
\begin{enumerate}[noitemsep]
	\item (Smooth model) $T_\theta$ is continuously differentiable in $\theta$ for every $x\in\tilde{\mathcal{X}}$, so $X_\theta$ is continuously differentiable in $\theta$.
	\item (Finite moments) $\mathbb{E}[XX^\top]$ is positive definite;
	$\mathbb{E}[\|X\|^4]$, $\mathbb{E}[\|X_\theta\|^4]$, $\mathbb{E}[\|\dot{X}_\theta\|^2]$, and $\mathbb{E}[\|X_\theta\|^2\|\dot{X}_\theta\|^2]$ are bounded uniformly over $\theta$;
	$\mathbb{E}_m[\|X_\theta\|^2]$, $\mathbb{E}_m[\|\dot{X}_\theta\|]$, and $\mathbb{E}_m[\|X_\theta\|\|\dot{X}_\theta\|]$ converge uniformly in $\theta$.
	\item (Smooth discriminator) $\lambda_\theta$ is continuously differentiable in $\theta$.
	\item (Exact maximizer) $\hat{\lambda}_\theta$ is the exact maximizer of $\mathbb{M}_\theta(D_\lambda)$ in that the FOC for $\hat{\lambda}_\theta$ is exactly zero for every $\theta\in\Theta$. \label{asm:ortho:foc}
	\item (Uniform convergence rate of discriminator) $\sup_\theta\|\hat{\lambda}_\theta-\lambda_\theta\|=O_P(n^{-1/2})$.
\end{enumerate}

For ease of notation, we assume that $\lambda$ and $\theta$ are one\hyp{}dimensional; however, the argument below applies equally to the vector case.
The FOC for $\hat{\lambda}_\theta$ yields
\(
	\mathbb{E}_n[(1-\Lambda(X\hat{\lambda}_\theta))X]-\mathbb{E}_m[\Lambda(X_\theta\hat{\lambda}_\theta)X_\theta]=0
\).
This holds for every $\theta$, so we may differentiate both sides by $\theta$,
which can be solved for the derivative of $\hat{\lambda}_\theta$ with respect to $\theta$,
\begin{multline*}
	\dot{\hat{\lambda}}_\theta=-(\mathbb{E}_n[\Lambda(1-\Lambda)(X\hat{\lambda}_\theta)X^2]+\mathbb{E}_m[\Lambda(1-\Lambda)(X_\theta\hat{\lambda}_\theta)X_\theta^2])^{-1}\\
	(\mathbb{E}_m[\Lambda(1-\Lambda)(X_\theta\hat{\lambda}_\theta)X_\theta\dot{X}_\theta]\hat{\lambda}_\theta+\mathbb{E}_m[\Lambda(X_\theta\hat{\lambda}_\theta)\dot{X}_\theta]).
\end{multline*}
Note that $\lambda_\theta$ satisfies the population FOC, which leads to the population counterpart of the same expression, so $\dot{\hat{\lambda}}_\theta$ is consistent for $\dot{\lambda}_\theta$.
Moreover, by the uniform convergence assumptions, we deduce $\sup_\theta\|\dot{\hat{\lambda}}_\theta-\dot{\lambda}_\theta\|=O_P(n^{-1/2})$.
Thus, the derivative of the discriminator converges.

To derive orthogonality, we first Taylor\hyp{}expand it in $\lambda$ around $\hat{\lambda}_\theta$. %
In doing so, the first\hyp{}order term can be ignored thanks to Condition \ref{asm:ortho:foc}.
For arbitrary $\theta$,
\begin{align*}
	&\mathbb{M}_\theta(D_{\lambda_\theta})-\mathbb{M}_\theta(D_{\hat{\lambda}_\theta})\\
	&=[\mathbb{E}_n\log\Lambda(X\lambda_\theta)+\mathbb{E}_m\log(1-\Lambda)(X_\theta\lambda_\theta)]-[\mathbb{E}_n\log\Lambda(X\hat{\lambda}_\theta)+\mathbb{E}_m\log(1-\Lambda)(X_\theta\hat{\lambda}_\theta)]\\
	&=\tfrac{1}{2}(\hat{\lambda}_\theta-\lambda_\theta)^2[-\mathbb{E}_n\Lambda(1-\Lambda)(X\hat{\lambda}_\theta)X^2+\mathbb{E}_m\Lambda(1-\Lambda)(X_\theta\hat{\lambda}_\theta)X_\theta^2]+o_P((\hat{\lambda}_\theta-\lambda_\theta)^2).
\end{align*}
Next, we expand it further in $\theta$ around $\hat{\theta}$.
\begin{align*}
	&[\mathbb{M}_\theta(D_{\lambda_\theta})-\mathbb{M}_\theta(D_{\hat{\lambda}_\theta})]-[\mathbb{M}_{\hat{\theta}}(D_{\lambda_{\hat{\theta}}})-\mathbb{M}_{\hat{\theta}}(D_{\hat{\lambda}_{\hat{\theta}}})]\\
	&=-(\hat{\lambda}_{\hat{\theta}}-\lambda_{\hat{\theta}})(\dot{\hat{\lambda}}_{\hat{\theta}}-\dot{\lambda}_{\hat{\theta}})(\theta-\hat{\theta})[\mathbb{E}_n\Lambda(1-\Lambda)(X\hat{\lambda}_\theta)X^2-\mathbb{E}_m\Lambda(1-\Lambda)(X_\theta\hat{\lambda}_\theta)X_\theta^2]\\
	&\hphantom{={}}-\tfrac{1}{2}(\hat{\lambda}_\theta-\lambda_\theta)^2(\theta-\hat{\theta})\mathbb{E}_n\Lambda(1-\Lambda)(1-2\Lambda)(X\hat{\lambda}_{\hat{\theta}})X^3\dot{\hat{\lambda}}_{\hat{\theta}}\\
	&\hphantom{={}}+\tfrac{1}{2}(\hat{\lambda}_\theta-\lambda_\theta)^2(\theta-\hat{\theta})\mathbb{E}_m\Lambda(1-\Lambda)(1-2\Lambda)(X_{\hat{\theta}}\hat{\lambda}_{\hat{\theta}})X_{\hat{\theta}}^3\dot{\hat{\lambda}}_{\hat{\theta}}\\
	&\hphantom{={}}+\tfrac{1}{2}(\hat{\lambda}_\theta-\lambda_\theta)^2(\theta-\hat{\theta})\mathbb{E}_m\Lambda(1-\Lambda)(1-2\Lambda)(X_{\hat{\theta}}\hat{\lambda}_{\hat{\theta}})X_{\hat{\theta}}^2\dot{X}_{\hat{\theta}}\hat{\lambda}_{\hat{\theta}}\\
	&\hphantom{={}}+(\hat{\lambda}_\theta-\lambda_\theta)^2(\theta-\hat{\theta})\mathbb{E}_m\Lambda(1-\Lambda)(X_{\hat{\theta}}\hat{\lambda}_{\hat{\theta}})X_{\hat{\theta}}\dot{X}_{\hat{\theta}}
	+o_P((\hat{\lambda}_\theta-\lambda_\theta)^2(1+|\hat{\theta}-\theta|)).
\end{align*}
This is $O_P(n^{-1})$ for fixed $\theta$, so we can take a shrinking neighborhood of $\theta$ around $\theta_0$ that contains $\hat{\theta}$ to make the supremum of this $o_P(n^{-1})$, yielding orthogonality.
If the neighborhood shrinks only slightly slower than $n^{-1/2}$, then convergence of $\hat{\lambda}_\theta$ and $\dot{\hat{\lambda}}_\theta$ can be relaxed to as slow as $o_P(n^{-1/4})$ if possibly a few more degrees of differentiability and finite moments are granted.
It is also straightforward to relax the exact FOC condition to allow for errors of negligible order and to allow for nonlinear but parametric logistic discriminators, such as small neural networks.
An interesting conclusion of this is that the curvature of the estimated loss converges faster than the level, as observed throughout \cref{sec:toy}.

For a general nonparametric discriminator, it is not trivial to obtain a similar low\hyp{}level condition.
\cref{supp:theory} develops conditions for $\hat{D}_\theta$ to converge faster than $n^{-1/4}$ (pointwise in $\theta$), which seems necessary but is not sufficient to derive orthogonality.%
\footnote{In a similar situation where the derivative of the nuisance parameter identifies $\theta$, \citet{ks1993} exploit the structure of a kernel density estimator to show the convergence of the derivative, whereby obtaining a corresponding orthogonality condition.}
In \cref{sec:toy}, the plots of $\mathbb{M}_\theta(\hat{D}_\theta)$ confirm orthogonality in examples with or without differentiability.

\subsubsection{On \cref{asm:misspec}} \label{sec:asm:misspec}

We may derive a closed\hyp{}form expression for $\tau_n$ in \cref{asm:misspec} when $T_\theta$ is differentiable in $\theta$ and $D_{\theta_0}$ in $x$.
Suppose that $\mathcal{X}$ and $\tilde{\mathcal{X}}$ are Euclidean spaces; denote the differentiation with respect to an argument by a prime and with respect to a subscript by a dot, e.g., $f_\theta'(x)=\frac{\partial}{\partial x}f_\theta(x)$ and $\dot{f}_\theta(x)=\frac{\partial}{\partial\theta^\top}f_\theta(x)$. %
Observe that
\begin{align*}
	n(\mathbb{P}_\theta-\mathbb{P}_{\theta_0})\log(1-D_{\theta_0})
	&=n\tilde{\mathbb{P}}_0[\log(1-D_{\theta_0})\circ T_\theta-\log(1-D_{\theta_0})\circ T_{\theta_0}]\\
	&\approx
	n(\theta-\theta_0)^\top\tilde{\mathbb{P}}_0\bigl[\dot{T}_{\theta_0}^\top\bigl(D_{\theta_0}\bigl[\tfrac{p_{\theta_0}'}{p_{\theta_0}}-\tfrac{p_0'}{p_0}\bigr]\circ T_{\theta_0}\bigr)\bigr].
\end{align*}
Thus, we find
\(
	\tau_n=\dot{T}_{\theta_0}^\top(D_{\theta_0}[p_{\theta_0}'/p_{\theta_0}-p_0'/p_0]\circ T_{\theta_0})
\).

Note that this assumption only matters for the misspecified case.
\Cref{fig:lemma4} verifies this assumption for the misspecified model in \cref{sec:toy:misspec} in the regions relevant for the asymptotic distribution.
The red dashed lines plot the LHS of \cref{asm:misspec},
\(
	n(\tilde{\mathbb{P}}_0-\tilde{P}_0)(\log(1-D_{\theta_0})\circ T_{\theta+h/\sqrt{n}}-\log(1-D_{\theta_0})\circ T_{\theta_0})
\).
For the normally\hyp{}misspecified logistic location model in \cref{sec:toy:misspec}, we see that this line is already very linear and its slope corresponds to $\tau_n$. %

\begin{figure}
\centering
\includegraphics[page=1]{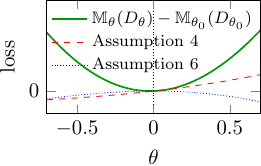}
\caption{\cref{asm:misspec,asm:model:smooth} for the normally\hyp{}misspecified logistic location model in \cref{sec:toy:misspec}. $n=m=300$.}
\label{fig:lemma4}
\end{figure}

\subsection{Asymptotic Distribution} \label{sec:dist}

To derive the asymptotic distribution of the adversarial estimator, we need the structural model to be differentiable as in maximum likelihood.

\begin{asm}[Twice differentiability] \label{asm:dqm}
The parameter space $\Theta$ is (a subset of) a Euclidean space $\mathbb{R}^k$.
The structural model $\{P_\theta:\theta\in\Theta\}$ has a likelihood that is twice differentiable in $\theta$ at $\theta_0$ for every $x\in\mathcal{X}$ with the derivatives continuous in both $x$ and $\theta$.
The Fisher information matrix $I_{\theta_0}\coloneqq P_{\theta_0}\dot{\ell}_{\theta_0}\dot{\ell}_{\theta_0}^\top=-P_{\theta_0}\ddot{\ell}_{\theta_0}$ and the matrix $\tilde{I}_{\theta_0}\coloneqq 2 P_{\theta_0}(D_{\theta_0}\dot{\ell}_{\theta_0}\dot{\ell}_{\theta_0}^\top+(\ddot{\ell}_{\theta_0}+\dot{\ell}_{\theta_0}\dot{\ell}_{\theta_0}^\top)\log(1-D_{\theta_0}))$ are positive definite.
\end{asm}

\begin{rem}
Under \cref{asm:spec}, the annoying term $(\mathbb{P}_\theta-\mathbb{P}_{\theta_0})\log(1-D_{\theta_0})$ in \cref{lem:lan:misspec} goes away, rendering twice differentiability unnecessary.
\end{rem}

We further impose a mild smoothness condition on $T_\theta$.

\begin{asm}[Smooth synthetic data generation] \label{asm:model:smooth}
For every compact $K\subset\Theta$,
\(
	\sqrt{\tfrac{n}{m}}\sup_{h\in K}\|\sqrt{m}[(\mathbb{P}_{\theta_0+h/\sqrt{n}}-P_{\theta_0+h/\sqrt{n}})-(\mathbb{P}_{\theta_0}-P_{\theta_0})]D_{\theta_0}\dot{\ell}_{\theta_0}\|=o_P^\ast(1)
\).
\end{asm}

Similarly to \cref{asm:misspec}, \cref{asm:model:smooth} is trivial if $n/m\to 0$.
Moreover, under the low\hyp{}level conditions (and notation) in \cref{sec:asm:misspec},
\begin{align*}
	\sqrt{n}(\mathbb{P}_\theta-\mathbb{P}_{\theta_0})D_{\theta_0}\dot{\ell}_{\theta_0}
	&=\sqrt{n}\tilde{\mathbb{P}}_0(D_{\theta_0}\dot{\ell}_{\theta_0}\circ T_\theta-D_{\theta_0}\dot{\ell}_{\theta_0}\circ T_{\theta_0})\\
	&\approx\tilde{\mathbb{P}}_0[\dot{T}_{\theta_0}^\top(D_{\theta_0}'\dot{\ell}_{\theta_0}+D_{\theta_0}\dot{\ell}_{\theta_0}')\circ T_{\theta_0}]\sqrt{n}(\theta-\theta_0).
\end{align*}
Therefore, if $\tilde{\mathbb{P}}_0[\dot{T}_{\theta_0}^\top(D_{\theta_0}'\dot{\ell}_{\theta_0}+D_{\theta_0}\dot{\ell}_{\theta_0}')\circ T_{\theta_0}]$ is consistent for $\tilde{P}_0[\dot{T}_{\theta_0}^\top(D_{\theta_0}'\dot{\ell}_{\theta_0}+D_{\theta_0}\dot{\ell}_{\theta_0}')\circ T_{\theta_0}]$, \cref{asm:model:smooth} holds even if $m\sim n$.

\Cref{fig:lemma4} verifies \cref{asm:model:smooth} for the normally\hyp{}misspecified model in \cref{sec:toy:misspec}.
The blue dotted lines plot the LHS of \cref{asm:model:smooth},
\(
	\sqrt{n}[(\mathbb{P}_{\theta_0+h/\sqrt{n}}-P_{\theta_0+h/\sqrt{n}})-(\mathbb{P}_{\theta_0}-P_{\theta_0})]D_{\theta_0}\dot{\ell}_{\theta_0}
\).

\begin{thm}[Asymptotic distribution of generator] \label{thm:theta:dist}
Under the conclusion of \cref{thm:5,asm:m,asm:neyman,asm:misspec,asm:model:smooth,asm:dqm},
\[
	\sqrt{n}(\hat{\theta}-\theta_0)=2\tilde{I}_{\theta_0}^{-1}\sqrt{n}[\mathbb{P}_0(1-D_{\theta_0})\dot{\ell}_{\theta_0}-\mathbb{P}_{\theta_0}D_{\theta_0}\dot{\ell}_{\theta_0}-\tilde{\mathbb{P}}_0\tau_n]+o_P^\ast(1)
	\leadsto N(0,\tilde{I}_{\theta_0}^{-1}V\tilde{I}_{\theta_0}^{-1}).
\]
where $V\coloneqq\lim_{n\to\infty}4[(P_{\theta_0}+\frac{n}{m}P_0)D_{\theta_0}(1-D_{\theta_0})\dot{\ell}_{\theta_0}\dot{\ell}_{\theta_0}^\top+\frac{n}{m}\tilde{P}_0[(D_{\theta_0}\dot{\ell}_{\theta_0}\circ T_{\theta_0})\tau_n^\top+\tau_n(D_{\theta_0}\dot{\ell}_{\theta_0}^\top\circ T_{\theta_0})+\tau_n\tau_n^\top]]$.
\end{thm}

A stronger efficiency result holds if the structural model is correctly specified.

\begin{asm}[Correct specification] \label{asm:spec}
The synthetic model $\{P_\theta:\theta\in\Theta\}$ is correctly specified, that is, $P_{\theta_0}=P_0$ and $D_{\theta_0}\equiv1/2$.
\end{asm}

\begin{coro}[Efficiency of generator] \label{thm:theta:efficiency}
Under the conclusion of \cref{thm:theta:dist,asm:spec},
\(
	\sqrt{n}(\hat{\theta}-\theta_0)
	=I_{\theta_0}^{-1}\sqrt{n}(\mathbb{P}_0-\mathbb{P}_{\theta_0})\dot{\ell}_{\theta_0}+o_P^\ast(1)
	\leadsto N(0,[1+\lim_{n\to\infty}\tfrac{n}{m}]I_{\theta_0}^{-1})
\).
\end{coro}

Thus, if $n/m\to0$, the adversarial estimator attains parametric efficiency.

\subsection{What If $\mathcal{D}$ Is Not Rich Enough?}

Our theory assumes that $\mathcal{D}$ is a sieve that eventually is capable of representing $D_\theta$.
In finite samples, however, we do not know how well $\mathcal{D}$ approximates $D_\theta$.
Therefore, it is interesting to see what happens when $\mathcal{D}$ is not a sieve but a fixed class of functions.
Although the complete treatment of this case is beyond our scope, we examine what happens to the population problem as we enrich $\mathcal{D}$, e.g., by gradually adding nodes and layers to the neural network.%
\footnote{The case where $\mathcal{D}$ is fixed to be the class of logistic discriminators is analyzed in \cref{sec:logistic}.}

For simplicity, we maintain \cref{asm:sieve:bound,asm:dqm,asm:spec} and assume that $\mathcal{D}$ contains a constant function $1/2$.
Let $\tilde{D}_\theta$ be the population maximizer of $M_\theta(D)$ in $\mathcal{D}$.
Since $M_\theta(D)-M_\theta(D_\theta)=-2d_\theta(D,D_\theta)^2+o(d_\theta(D,D_\theta)^2)$ by \cref{thm:obj:rate}, $\tilde{D}_\theta$ is equivalent to a minimizer of $d_\theta(D,D_\theta)^2$ in $\mathcal{D}$ up to $o(d_\theta(D,D_\theta)^2)$.
Under \cref{asm:spec}, $\tilde{D}_{\theta_0}=D_{\theta_0}\equiv1/2$ and $M_{\theta_0}(1/2)=M_\theta(1/2)$.
By \cref{thm:obj:rate},
\begin{multline*}
	M_{\theta_0}(\tilde{D}_{\theta_0})-M_\theta(\tilde{D}_\theta)=M_\theta(D_{\theta_0})-M_\theta(D_\theta)+M_\theta(D_\theta)-M_\theta(\tilde{D}_\theta)\\
	=-2d_\theta(D_{\theta_0},D_\theta)^2+2d_\theta(\tilde{D}_\theta,D_\theta)^2+o(d_\theta(D_{\theta_0},D_\theta)^2)+o(d_\theta(\tilde{D}_\theta,D_\theta)^2).
\end{multline*}
Note that by \cref{lem:average},
\(
	d_\theta(D_{\theta_0},D_\theta)^2
	=\tfrac{1}{2}\int\tfrac{p_0}{p_0+p_0}(\sqrt{p_0}-\sqrt{p_\theta})^2+\tfrac{1}{2}\int\tfrac{p_\theta}{p_\theta+p_\theta}(\sqrt{p_0}-\sqrt{p_\theta})^2+o(h(p_0,p_\theta)^2)
	=\tfrac{1}{2}h(p_0,p_\theta)^2+o(h(p_0,p_\theta)^2)
\).
Thus, we obtain
\[
	M_{\theta_0}(\tilde{D}_{\theta_0})-M_\theta(\tilde{D}_\theta)=-h(p_0,p_\theta)^2+2d_\theta(\tilde{D}_\theta,D_\theta)^2+o(h(p_0,p_\theta)^2).
\]
If $\mathcal{D}$ contains $D_\theta$, the second term is zero and the Hellinger curvature allows us to estimate $\theta$ efficiently;
if $\mathcal{D}$ is a singleton set that contains only $1/2$, the first and second terms cancel and the objective function becomes completely flat, rendering estimation of $\theta$ impossible.
Therefore, the second term represents the loss in efficiency due to the limited capacity of $\mathcal{D}$.
For the regular logit case, we know that $\mathcal{D}$ is already rich enough that the curvature admits $\sqrt{n}$\hyp{}estimation.
Then, as we enrich $\mathcal{D}$, it becomes more capable of minimizing $d_\theta(\tilde{D}_\theta,D_\theta)^2$, getting closer to efficiency.

\section{Empirical Application: ``Why Do the Elderly Save?''} \label{sec:application}

Using the adversarial framework, we examine the elderly's saving, following \citet{dfj} (henceforth \citetalias{dfj}).
The elderly save for various reasons---uncertainty on survival, bequest motive, or ever\hyp{}rising medical expenses as they age. Different motives for saving yield different implications on policy evaluation such as Medicaid and Medicare.
Hence, it is an important and active area of research.

The risk the elderly face is highly heterogeneous, depending on their gender, age, health status, and permanent income.
This implies potentially large heterogeneity in the saving motive across individuals; not accounting for this can bias the estimates of utility.
For example, the rich live several years more than the poor on average.
Failure to reflect this difference can make the rich look thriftier than they are.
On the other hand, existing estimation methods such as SMM may suffer from severe lack of precision when various heterogeneity is introduced.
This motivates adversarial estimation 
as a more efficient alternative to SMM.
Indeed, our adversarial estimates, using the same model and the same data as in \citetalias{dfj}, will see considerable gains in precision.

\subsection{Agent's Problem}

We focus on the behavior of single, retired individuals of age 70 and older.
In each period, a surviving single retired agent receives utility $u(c)$ from consumption $c$ and, if they die in that period, additional utility $\phi(e)$ from leaving estate $e$, where
\[
	u(c)\coloneqq\frac{c^{1-\nu}}{1-\nu}, \qquad\quad
	\phi(e)\coloneqq\vartheta\frac{(e+k)^{1-\nu}}{1-\nu},
\]
and $\nu$ is the relative risk aversion and $\vartheta$ and $k$ are the intensity and curvature of the bequest motive.
Each individual is associated with gender $g$ and permanent income $I$, and carries six state variables: age $t$, asset $a_t$, nonasset income $y_t$, health status $h_t$, medical expense shock $\zeta_t$, and survival $s_t$.
Health and survival are binary, where $h_t=1$ means they are healthy at age $t$, and $s_t=1$ they survive to the next period.

They face three channels of uncertainty: health, survival, and medical expenses.
Heath and survival evolve as Markov chains. We denote
\[
	\pi_{H}(g,h_t,I,t)\coloneqq\Pr(h_{t+1}=1\mid g,h_t,I,t), \quad
	\pi_{S}(g,h_t,I,t)\coloneqq\Pr(s_{t+1}=1\mid g,h_t,I,t).
\]
The medical expenses they incur are given by
\(
	\log m_{t}=m(g,h_t,I,t)+\sigma(g,h_t,I,t)\psi_{t}
\),
where $m$ and $\sigma$ are deterministic functions, %
$\psi_{t}=\zeta_{t}+\xi_{t}$, $\xi_{t}\sim N(0,\sigma_{\xi}^{2})$, $\zeta_{t}=\rho\zeta_{t-1}+\epsilon_{t}$, and $\epsilon_{t}\sim N(0,\sigma_{\epsilon}^{2})$.
The nonasset income evolves deterministically as $y_t=y(g,I,t)$.
The asset evolves as
\(
	a_{t+1}=a_{t}+y_{n}(ra_{t}+y_{t},\tau)+b_{t}-m_{t}-c_t
\),
where $b_{t}\geq0$ is the {\em government transfer}, $r$ the {\em risk\hyp{}free pretax rate of return}, $y_{n}(\cdot,\tau)$ the {\em posttax income}, and $\tau$ the {\em tax structure}.
The agent faces a borrowing constraint $a_t\geq0$ while social insurance guarantees minimum consumption $c_t\geq\underline{c}$; government transfer $b_t$ is positive only when both constraints cannot be satisfied without it.

The timing in each period is given as follows. Heath $h_t$ and medical expenses $m_t$ realize; then the individual chooses consumption $c_t$; then survival $s_t$ realizes; if $s_t=0$, they leave the remaining assets as bequest; if $s_t=1$, move on to the next period.

Denoting the {\em cash\hyp{}on\hyp{}hand} by $x_t\coloneqq c_t+a_{t+1}$, the agent's Bellman equation is
\[
	V_{t}(x,g,h,I,\zeta)=\max_{c,x'} \, u(c,h)+\beta[s\mathbb{E}_tV_{t+1}(x',g,h',I,\zeta')+(1-s)\phi(e)]
\]
subject to $x'=(x-c)+y_{n}(r(x-c)+y',\tau)+b'-m'$, $e=(x-c)-\max\{0,\tilde{\tau}(x-c-\tilde{x})\}$, and $x\geq c\geq\underline{c}$.
The first constraint is the budget constraint; the second the bequest (taxed at rate $\tilde{\tau}$ with deduction $\tilde{x}$); the last the borrowing and consumption constraints.

We also look at two transformations: the {\em marginal propensity to consume at the moment of death} $\text{MPC}\coloneqq(1+r)/(1+r+[\beta\vartheta(1+r)]^{1/\nu})$ and the {\em implied asset floor} $\underline{a}\coloneqq k/[\beta\vartheta(1+r)]^{1/\nu}$ above which individuals get utility from bequeathing.%
\footnote{The {\em marginal propensity to bequeath (MPB)} is defined by $1-\text{MPC}$.}

\subsection{Data}

We use the same data as \citetalias{dfj}, taken from {\em Assets and Health Dynamics Among the Oldest Old (AHEAD)}.
The sample consists of non\hyp{}institutionalized individuals of age 70 and older in 1994. It contains 8,222 individuals in 6,047 households (3,872 singles and 2,175 couples).
The survey took place biyearly from 1994 to 2006.
We focus on 3,259 single retired individuals, 592 of which are men and 2,667 women.%
\footnote{Single individuals are those who were neither married nor cohabiting at any point in the analysis.}
Of those, 884 were alive in 2006.
We drop the first survey in 1994 for reliability, following \citetalias{dfj}.

The survey collects information on age $t$, financial wealth $a_t$, nonasset income $y_t$, medical expenses $m_t$, and health status $h_t$. Financial wealth includes real estate, autos, several other liquid assets, retirement accounts, etc. Nonasset income includes social security benefits, veteran's benefits, and other benefits. Medical expenses are total out\hyp{}of\hyp{}pocket spending; the average yearly expenses are \$3,700 with standard deviation \$13,400. The permanent income is not observed, but we use as a proxy the ranking of individual average income over time.
The health status is a binary variable indicating whether the individual perceives herself as healthy.

\subsection{Estimation} \label{sec:estimation}

Following \citetalias{dfj}, we carry out estimation in two steps: (1) estimate $\pi_H$, $\pi_S$, $m$, $\sigma$, $\rho_m$, $\sigma_\xi$, $\sigma_\epsilon$ (in fact, we borrow numbers from \citetalias{dfj}), (2) estimate $\nu$, $\text{MPC}$, and $k$ using our adversarial approach. The parameters $r$, $\tau$, $\tilde{\tau}$, and $\tilde{x}$ are fixed as in the original paper, and $\beta = 0.971$.
For $\underline{c}$, we fix it at \$4,500 to reflect annual social security payments.%
\footnote{In their preferred specification, \citetalias{dfj} estimate $\beta$ and $\underline{c}$ in addition to $\nu$, MPC, and $k$. Instead, we fix $\beta$ and $\underline{c}$ to reasonable values in the literature. Changing $\underline{c}$ mostly affects estimates of $\nu$.}
After the second step, we can also recover $\vartheta$ and $\underline{a}$.

We consider two different sets of inputs to the discriminator.
The first set consists of the log age of an individual in 1996, permanent income (the aforementioned proxy), the profile (full history) of asset holdings, and the profile of survival indicators,%
\footnote{All individuals are alive in 1996, so we drop $s_{t_{1996}}$.}
\[
	X_1\coloneqq(1,\log t_{1996},I,a_{t_{1996}},\dots,a_{t_{2006}},s_{t_{1998}},\dots,s_{t_{2006}})\in\mathbb{R}^{14}.
\]
This is intended to capture similar identifying variation as \citetalias{dfj}.%
\footnote{\citetalias{dfj} use median assets as moments. Unlike in \cref{sec:toy:smm,sec:toy:case}, median moments cannot be translated into logistic inputs that yield an asymptotically equivalent estimator.}
The second set is augmented with gender and the profile of health status,
\[
	X_2\coloneqq(X_1,g,h_{t_{1996}},\dots,h_{t_{2006}})\in\mathbb{R}^{21},
\]
aiming to capture more variation for the bequest motive as explained in \cref{sec:health}.

We use cross validation to configure the discriminator. %
We focus on feed\hyp{}forward neural networks with sigmoid activation functions with at most two hidden layers.
We fix $\theta$ at a preliminary estimate; split the actual data into sample 1 (80\%) and sample 2 (20\%); estimate $D$ with sample 1, varying the numbers of nodes and layers; evaluate their classification accuracy with sample 2;%
\footnote{We use the classification accuracy provided by Keras's ADAM based on thresholding.}
pick the network configuration with the highest accuracy.
The selected neural network discriminator consists of two hidden layers, the first with 20 nodes and the second 10 nodes.

We compare our estimates with SMM in \citetalias{dfj}.
They use 150 moments of median assets of groups divided by the cohort and permanent income quintile in each year.
The cohort is defined on a four\hyp{}year window; Cohort 1 are those who were 72--6 years old in 1996; Cohort 2 were 77--81; Cohort 3 were 82--6; Cohort 4 were 87--91; Cohort 5 were 92 and older.
Details are in \citetalias{dfj}.
We note that accounting for health and gender is infeasible in SMM since it yields too many moments. %

\subsection{Results}

\Cref{results} gives the parameter estimates from \citetalias{dfj} and our adversarial method with specifications $X_1$ and $X_2$.
Parenthesized numbers are the standard errors; we use \citet{honore2017poor} to compute them for the adversarial estimates.
The first row is the SMM estimates in \citetalias{dfj}.
The second and third rows come from the adversarial estimation; the second uses $X_1$ (14 variables) and the third $X_2$ (21 variables).

\begin{table}[t!]
\caption{Estimates of the structural parameters. The choice of inputs $X_1$ to the discriminator is intended to capture similar identifying variation as \citetalias{dfj}. The inputs $X_2$ contain additional variation in gender and health. Standard errors for the adversarial estimates are obtained by the poor (wo)man's bootstrap.}
\label{results}
\centering
\small
\begin{tabular}{lcccccccc}
\toprule \midrule
& $\beta$ & $\underline{c}$ [\$] & $\nu$ & $\vartheta$ & $k$ [k\$] & MPC & $\underline{a}$ [\$] & Loss \\
\midrule
\citetalias[Table 3]{dfj} & 0.97 & 2,665 & 3.84 & 2,360 & 273 & 0.12 & 36,215 & $-0.67$ \\
& (0.05) & \hphantom{0}(353) & (0.55) & (8,122) & (446) & & & \\
Adversarial $X_1$ & 0.97 & 4,500 & 6.14 & 4,865 & 16.89 & 0.20 & 4,243 & $-0.67$ \\
& & & (.009) & (9.002) & (.030) & (.017) & (19.73) & \\
Adversarial $X_2$ & 0.97 & 4,500 & 5.99 & 192,676 & 10.02 & 0.12 & 1,320 & $-0.78$ \\
& & & (.005) & (8,112) & (.015) & (.014) & (3.66) & \\
\bottomrule
\end{tabular}
\end{table}

A major difference between our estimates and \citetalias{dfj}'s is the curvature of the utility of bequests $k$.
Ours is an order of magnitude smaller, which has an important implication: while \citetalias{dfj} conclude only the super rich would obtain utility from bequeathing, our estimate suggests bequeathing matters across the entire permanent income distribution.
A related number is the implied asset floor $\underline{a}$.
We obtain estimates of \$1,320 and \$4,243, which are on the lower side of the estimates known in the literature. However, they correspond respectively to the 22nd and 24th percentiles of the distribution of assets one period before deaths (see \cref{sec:counterfactual}) in our sample. We interpret these numbers as providing a sensible fit of the data. In contrast, \citetalias{dfj}'s implied asset floor is \$36,215, which corresponds to the 40th percentile. 

Overall, the intensity of the bequest motive is minor in \citetalias{dfj} and $X_1$ but non\hyp{}negligible in $X_2$.
While $k$ is low for both $X_1$ and $X_2$, MPC is almost twice as large in $X_1$ compared to $X_2$.
Consequently, individuals care about bequests less than their own consumption according to $X_1$.
\citetalias{dfj} and adversarial also differ in risk aversion $\nu$.
A large value of risk aversion rationalizes the observed saving patterns when the consumption floor $\underline{c}$ is fixed at \$4,500, a reasonable value in the literature.%
\footnote{\citetalias{dfj}'s risk aversion estimate increases from 3.84 to 6.04 in an alternative specification where $\underline{c}$ is fixed at \$5,000. However, according to their criterion, the fit of the model decreases substantially.}

In line with our theory, the adversarial estimation provides substantial gains in precision relative to \citetalias{dfj}.
The decrease in standard errors reflects that the data is sufficiently powerful to conclude the importance of the bequest motive, especially when exploiting additional variation in gender and health.

The last column reports the cross\hyp{}entropy loss for each method.
To make a fair comparison, we take the estimates of each method and train the discriminator using $X_2$ as the input.
The loss for adversarial $X_1$ does not improve over \citetalias{dfj} but does so substantially for adversarial $X_2$, which is consistent with our observation that gender and health provide useful identification for the bequest motive.
\section*{Appendix}

\appendix

\begin{proof}[Proof of \cref{thm:theta:consistency}]
Observe that $\mathbb{M}_{\hat{\theta}}(D_{\hat{\theta}})-\inf_{\theta\in\Theta}\mathbb{M}_\theta(D_\theta)$ is bounded by
\[
	\bigl[\mathbb{M}_{\hat{\theta}}(\hat{D}_{\hat{\theta}})-\inf_{\theta\in\Theta}\mathbb{M}_\theta(\hat{D}_\theta)\bigr]
	+[\mathbb{M}_{\hat{\theta}}(D_{\hat{\theta}})-\mathbb{M}_{\hat{\theta}}(\hat{D}_{\hat{\theta}})]
	+\sup_{\theta\in\Theta}\,[\mathbb{M}_\theta(\hat{D}_\theta)-\mathbb{M}_\theta(D_\theta)].
\]
The first difference is less than $o_P^\ast(1)$ and the latter two are $o_P^\ast(1)$ by assumption.
Therefore, $\mathbb{M}_{\hat{\theta}}(D_{\hat{\theta}})\leq\inf_{\theta\in\Theta}\mathbb{M}^\theta(D_\theta)+o_P^\ast(1)$.
Let $\mathcal{M}_1\coloneqq\{\log D_\theta:\theta\in\Theta\}$ and $\mathcal{M}_2\coloneqq\{\log(1-D_\theta)\circ T_\theta:\theta\in\Theta\}$.
By the assumption of Glivenko--Cantelli, $\|\mathbb{P}_0-P_0\|_{\mathcal{M}_1}\to0$ and $\|\tilde{\mathbb{P}}_0-\tilde{P}_0\|_{\mathcal{M}_2}\to0$ in outer probability as $n,m\to\infty$.
By \citet[Corollary 3.2.3 (i)]{vw1996}, it follows that $\hat{\theta}\to\theta_0$ in outer probability.
\end{proof}

Let
\(
	\tilde{h}(\theta_1,\theta_2)\coloneqq[\tilde{P}_0(\sqrt{\vphantom{T}\smash{(p_0/p_{\theta_1})\circ T_{\theta_1}}}-\sqrt{\vphantom{T}\smash{(p_0/p_{\theta_2})\circ T_{\theta_2}}})^2]^{1/2}
\).

\begin{proof}[Proof of \cref{thm:5}]
\Cref{asm:neyman} implies $\mathbb{M}_{\hat{\theta}}(D_{\hat{\theta}})\leq\mathbb{M}_{\theta_0}(D_{\theta_0})+O_P^\ast(n^{-1})$, so we apply \citet[Theorem 3.2.5]{vw1996} to $\mathbb{M}_\theta(D_\theta)$.
By \cref{asm:misspec}, $M_\theta(D_\theta)-M_{\theta_0}(D_{\theta_0})\gtrsim h(\theta,\theta_0)^2\wedge c$ for some $c>0$ globally in $\theta\in\Theta$.
By \cref{asm:maximal}, $\tilde{h}(\theta,\theta_0)^2=O(h(\theta,\theta_0))$ as $\theta\to\theta_0$.

Next, we show the convergence of the sample objective function. Note that
\[
	(\mathbb{M}_{\theta_0}-M_{\theta_0})(D_{\theta_0})-(\mathbb{M}_\theta-M_\theta)(D_\theta)
	=(\mathbb{P}_0-P_0)\log\tfrac{D_{\theta_0}}{D_\theta}+(\tilde{\mathbb{P}}_0-\tilde{P}_0)\log\tfrac{(1-D_{\theta_0})\circ T_{\theta_0}}{(1-D_\theta)\circ T_\theta}.
\]
By \cref{lem:distance},
\(
	\|\log\frac{D_{\theta_0}}{D_\theta}\|_{P_0,B}^2\leq4h(\theta,\theta_0)^2
\)
and
\(
	\|\log\frac{(1-D_{\theta_0})\circ T_{\theta_0}}{(1-D_\theta)\circ T_\theta}\|_{\tilde{P}_0,B}^2\leq4\tilde{h}(\theta,\theta_0)^2
\).
For $\delta>0$, define $\mathcal{M}_\delta^1\coloneqq\{\log\frac{D_{\theta_0}}{D_\theta}:h(\theta,\theta_0)\leq\delta\}$ and $\mathcal{M}_\delta^2\coloneqq\{\log\frac{(1-D_{\theta_0})\circ T_{\theta_0}}{(1-D_\theta)\circ T_\theta}:\tilde{h}(\theta,\theta_0)\leq\delta\}$.
By \citet[Lemma 3.4.3]{vw1996},
\[
	\mathbb{E}^\ast\sup_{h(\theta,\theta_0)<\delta}\bigl|\sqrt{n}(\mathbb{P}_0-P_0)\log\tfrac{D_{\theta_0}}{D_\theta}\bigr|\lesssim J_{[]}(2\delta,\mathcal{M}_\delta^1,\|\cdot\|_{P_0,B})\bigl[1+\tfrac{J_{[]}(2\delta,\mathcal{M}_\delta^1,\|\cdot\|_{P_0,B})}{4\delta^2\sqrt{n}}\bigr].
\]
Let $[\ell,u]$ be an $\varepsilon$\hyp{}bracket in $\{p_\theta\}$ with respect to $h$.
Since $u-\ell\geq0$ and $e^{|x|}-1-|x|\leq2(e^{x/2}-1)^2$ for every $x\geq0$,
\begin{align*}
	\bigl\|\log\tfrac{p_0+u}{p_0+p_{\theta_0}}-\log\tfrac{p_0+\ell}{p_0+p_{\theta_0}}\bigr\|_{P_0,B}^2
	\leq4\int\bigl(\sqrt{\tfrac{p_0+u}{p_0+\ell}}-1\bigr)^2p_0
	&\leq4\int(\sqrt{p_0+u}-\sqrt{p_0+{}\smash{\ell}})^2\\
	&\leq4h(u,\ell)^2\leq4\varepsilon^2.
\end{align*}
Thus, $[\log\frac{p_0+\ell}{p_0+p_{\theta_0}},\log\frac{p_0+u}{p_0+p_{\theta_0}}]$ makes a $2\varepsilon$\hyp{}bracket in $\mathcal{M}^1$.
Hence, $N_{[]}(2\varepsilon,\mathcal{M}_\delta^1,\|\cdot\|_{P_0,B})\leq N_{[]}(\varepsilon,\mathcal{P}_\delta,h)\lesssim(\delta/\varepsilon)^r$ by \cref{asm:maximal}.
This induces $J_{[]}(2\delta,\mathcal{M}_\delta^1,\|\cdot\|_{P_0,B})\lesssim\delta$.
Ergo,
\[
	\mathbb{E}^\ast\sup_{h(\theta,\theta_0)<\delta}\bigl|\sqrt{n}(\mathbb{P}_0-P_0)\log\tfrac{D_{\theta_0}}{D_\theta}\bigr|\lesssim\delta+\tfrac{1}{\sqrt{n}}.
\]
Similarly,
\(
	\mathbb{E}^\ast\sup_{\tilde{h}(\theta,\theta_0)<\delta}|\sqrt{m}(\tilde{\mathbb{P}}_0-\tilde{P}_0)\log\frac{1-D_{\theta_0}}{1-D_\theta}|
	\lesssim\delta+\frac{1}{\sqrt{m}}
\).
Then, the result follows by \citet[Theorem 3.2.5]{vw1996}.
\end{proof}

\begin{proof}[Proof of \cref{thm:theta:dist}]
By \cref{thm:5}, $\hat{\theta}$ is consistent and $\sqrt{n}(\hat{\theta}-\theta_0)$ is uniformly tight.
\cref{asm:neyman} implies
\(
	\mathbb{M}_{\hat{\theta}}(D_{\hat{\theta}})\leq\inf_{\theta\in G_n}\mathbb{M}_\theta(D_\theta)+o_P^\ast(n^{-1})
\).
Let $\mathbb{G}_{\theta_0}\dot{\ell}_{\theta_0}\coloneqq\sqrt{n}(\mathbb{P}_0-P_0)(1-D_{\theta_0})\dot{\ell}_{\theta_0}-\sqrt{n}(\mathbb{P}_{\theta_0}-P_{\theta_0})D_{\theta_0}\dot{\ell}_{\theta_0}-\sqrt{n}(\tilde{\mathbb{P}}_0-\tilde{P}_0)\tau_n$.
With \cref{asm:m,asm:misspec,asm:model:smooth,asm:dqm}, \cref{lem:lan:misspec} implies that uniformly in $h\in K$ compact,
\[
	n[\mathbb{M}_{\theta_0+h/\sqrt{n}}(D_{\theta_0+h/\sqrt{n}})-\mathbb{M}_{\theta_0}(D_{\theta_0})]
	=-h^\top\mathbb{G}_{\theta_0}\dot{\ell}_{\theta_0}+\tfrac{h^\top\tilde{I}_{\theta_0}h}{4}+o_P(1+\tfrac{n}{m}).
\]
In particular, this holds for both $\hat{h}\coloneqq\sqrt{n}(\hat{\theta}-\theta_0)$ and $\breve{h}\coloneqq 2\tilde{I}_{\theta_0}^{-1}\mathbb{G}_{\theta_0}\dot{\ell}_{\theta_0}$, so
\begin{gather*}
	n[\mathbb{M}_{\theta_0+\hat{h}/\sqrt{n}}(D_{\theta_0+\hat{h}/\sqrt{n}})-\mathbb{M}_{\theta_0}(D_{\theta_0})]
	=-\hat{h}^\top\mathbb{G}_{\theta_0}\dot{\ell}_{\theta_0}+\tfrac{1}{4}\hat{h}^\top\tilde{I}_{\theta_0}\hat{h}+o_P^\ast(1+\tfrac{n}{m}),\\
	n[\mathbb{M}_{\theta_0+\breve{h}/\sqrt{n}}(D_{\theta_0+\breve{h}/\sqrt{n}})-\mathbb{M}_{\theta_0}(D_{\theta_0})]=-\mathbb{G}_{\theta_0}\dot{\ell}_{\theta_0}^\top\tilde{I}_{\theta_0}^{-1}\mathbb{G}_{\theta_0}\dot{\ell}_{\theta_0}+o_P(1+\tfrac{n}{m}).
\end{gather*}
Since $G_n$ shrinks slower than $1/\sqrt{n}$, $\theta_0+\breve{h}/\sqrt{n}$ is eventually contained in $G_n$.
Since $\hat{h}$ minimizes $\mathbb{M}_\theta(D_\theta)$ up to $o_P^\ast(1/n)$ in $G_n$, the LHS of the first equation is larger than that of the second up to $o_P^\ast(1)$. Subtracting the two,
\[
	\tfrac{1}{4}(\hat{h}-2\tilde{I}_{\theta_0}^{-1}\mathbb{G}_{\theta_0}\dot{\ell}_{\theta_0})^\top\tilde{I}_{\theta_0}(\hat{h}-2\tilde{I}_{\theta_0}^{-1}\mathbb{G}_{\theta_0}\dot{\ell}_{\theta_0})+o_P^\ast(1+\tfrac{n}{m})\leq0.
\]
Since $\tilde{I}_{\theta_0}$ is positive definite, $\hat{h}-2\tilde{I}_{\theta_0}^{-1}\mathbb{G}_{\theta_0}\dot{\ell}_{\theta_0}=o_P^\ast(\sqrt{1+n\smash{/}m})$, proving the expression of $\sqrt{n}(\hat{\theta}-\theta_0)$.
The asymptotic variance is $4\tilde{I}_{\theta_0}^{-1}\Var(\mathbb{G}_{\theta_0}\dot{\ell}_{\theta_0})\tilde{I}_{\theta_0}^{-1}$.
Since $\mathbb{P}_0$ and $\mathbb{P}_{\theta_0}$ are independent,
\begin{align*}
	\Var(\mathbb{G}_{\theta_0}\dot{\ell}_{\theta_0})
	&=P_0(1-D_{\theta_0})^2\dot{\ell}_{\theta_0}\dot{\ell}_{\theta_0}^\top+\lim_{n\to\infty}\tfrac{n}{m}\tilde{P}_0(D_{\theta_0}\dot{\ell}_{\theta_0}\circ T_{\theta_0}+\tau_n)(D_{\theta_0}\dot{\ell}_{\theta_0}\circ T_{\theta_0}+\tau_n)^\top\\
	&=P_{\theta_0}D_{\theta_0}(1-D_{\theta_0})\dot{\ell}_{\theta_0}\dot{\ell}_{\theta_0}^\top+\lim_{n\to\infty}\tfrac{n}{m}P_0 D_{\theta_0}(1-D_{\theta_0})\dot{\ell}_{\theta_0}\dot{\ell}_{\theta_0}^\top\\
	&\hspace{50pt}+\lim_{n\to\infty}\tfrac{n}{m}\tilde{P}_0[(D_{\theta_0}\dot{\ell}_{\theta_0}\circ T_{\theta_0})\tau_n^\top+\tau_n(D_{\theta_0}\dot{\ell}_{\theta_0}^\top\circ T_{\theta_0})+\tau_n\tau_n^\top].
	\tag*\qedhere
\end{align*}
\end{proof}

\begin{singlespacing}
\bibliographystyle{ecta}
\bibliography{reference}
\end{singlespacing}

%% file: supptext.tex
\title{\MakeUppercase{An Adversarial Approach to Structural Estimation}\\{\bf Online Appendix}}

\setcounter{footnote}{0}
\maketitle

\setcounter{page}{1}

\setcounter{section}{0}
\setcounter{lem}{0}
\setcounter{thm}{0}
\setcounter{exa}{0}
\setcounter{asm}{0}

\renewcommand{\thesection}{S.\arabic{section}}
\renewcommand{\thelem}{S.\arabic{lem}}
\renewcommand{\thethm}{S.\arabic{thm}}
\renewcommand{\theexa}{S.\arabic{exa}}
\renewcommand{\theprop}{S.\arabic{prop}}
\renewcommand{\theasm}{S.\arabic{asm}}

\renewcommand{\theHsection}{Supplement.\thesection}
\renewcommand{\theHlem}{Supplement.\thelem}
\renewcommand{\theHthm}{Supplement.\thethm}
\renewcommand{\theHexa}{Supplement.\theexa}
\renewcommand{\theHasm}{Supplement.\theasm}

\defcitealias{dfj}{DFJ}

\section{Equivalence to SMM When $D$ is Logistic} \label{sec:logistic}

We show that the adversarial estimator with a logistic discriminator is asymptotically equivalent to SMM.
Importantly, we do not assume that the logistic discriminator is ``correctly specified'' so the oracle discriminator $D_\theta$ may not take the form of a logistic classifier.
In turn, we assume that the moments are correctly specified, $\mathbb{E}[X_i]=\mathbb{E}[X_{i,\theta_0}]$; however, the structural model may still be misspecified.
As in \cref{sec:toy}, let $D_\lambda(x)=\Lambda(x^\top\lambda)$ be the logistic discriminator and $\lambda_\theta$ and $\hat{\lambda}_\theta$ be the population parameter and its estimator for each $\theta$, respectively.
We employ the same notation as \cref{sec:asm:neyman}. %

In particular, we show that the adversarial estimator $\hat{\theta}$ with this discriminator is asymptotically equivalent to the following SMM estimator,
\[
	\tilde{\theta}\coloneqq\argmin_{\theta\in\Theta}\,(\mathbb{E}_n[X]-\mathbb{E}_m[X_\theta])^\top\Omega(\mathbb{E}_n[X]-\mathbb{E}_m[X_\theta])
	\ \ \text{for} \ \ 
	\Omega\coloneqq\bigl(\tfrac{\mathbb{E}[X X^\top]+\mathbb{E}[X_{\theta_0}X_{\theta_0}^\top]}{2}\bigr)^{-1}.
\]
This is optimally weighted when $X$ and $X_\theta$ contain a constant term and the second\hyp{}order moments are also correctly specified (viz.\ $\mathbb{E}[X X^\top]=\mathbb{E}[X_{\theta_0}X_{\theta_0}^\top]$), in which case $\Omega$ reduces to $\mathbb{E}[X X^\top]^{-1}$.
For simplicity, we ignore estimation of $\Omega$.
To show their equivalence, we assume the following.
\begin{enumerate}[noitemsep]
	\item (Growing synthetic sample size) $n/m$ converges.
	\item (Smooth model) $T_\theta$ is twice continuously differentiable in $\theta$ for every $x\in\tilde{\mathcal{X}}$.
	\item (Finite moments) $\mathbb{E}[X X^\top]$ is positive definite;
	$\mathbb{E}[\|\dot{X}_\theta\|^2]$ and $\mathbb{E}[\|\ddot{X}_\theta\|]$ are bounded uniformly in $\theta$;
	$\mathbb{E}_m[\|X_\theta\|^2]$ and $\mathbb{E}_m[\|\dot{X}_\theta\|^2]$ converge uniformly in $\theta$.
	\item (Correctly specified moments) $\mathbb{E}[X]=\mathbb{E}[X_{\theta_0}]$. \label{c:spec}
	\item (Identification of $\lambda_{\theta_0}$) $\lambda_{\theta_0}$ is unique. \label{c:id1}
	\item (Smooth discriminator) $\lambda_\theta$ is continuously differentiable in $\theta$.
	\item (Exact maximizer) $\hat{\lambda}_\theta$ is the exact maximizer of $\mathbb{M}_\theta(D_\lambda)$ in that the FOC for $\hat{\lambda}_\theta$ is exactly zero for every $\theta\in\Theta$.
	\item (Uniform convergence rate of discriminator) $\sup_\theta\|\hat{\lambda}_\theta-\lambda_\theta\|=O_P(n^{-1/2})$.
	\item (Identification of $\theta_0$) %
	$\mathbb{E}[\dot{X}_{\theta_0}]$ is of full row rank.
	\item (Exact minimizer) $\hat{\theta}$ is the exact minimizer of $\mathbb{M}_\theta(D_{\hat{\lambda}_\theta})$ in that the FOC for $\hat{\theta}$ is exactly zero.
	\item (Consistency) $\hat{\theta}$ and $\tilde{\theta}$ are consistent for $\theta_0$.
\end{enumerate}

The FOC for $\lambda_{\theta_0}$ gives $\mathbb{E}[(1-\Lambda(X^\top\lambda_{\theta_0}))X]=\mathbb{E}[\Lambda(X_{\theta_0}^\top\lambda_{\theta_0})X_{\theta_0}]$.
Conditions \ref{c:spec} and \ref{c:id1} imply $\lambda_{\theta_0}=0$.
The Taylor expansion of the FOC for $\hat{\lambda}_{\theta_0}$ yields
\(
	\sqrt{n}(\hat{\lambda}_{\theta_0}-0)=\Omega\sqrt{n}(\mathbb{E}_n[X]-\mathbb{E}_m[X_\theta])+o_P(1)
	\leadsto N(0,V_{\lambda})
\)
for
\(
	V_{\lambda}\coloneqq\Omega[\Var(X)+\lim\tfrac{n}{m}\Var(X_{\theta_0})]\Omega
\).
Also, by the same reasoning as \cref{sec:asm:neyman}, $\sup_\theta\|\dot{\hat{\lambda}}_\theta-\dot{\lambda}_\theta\|=O_P(n^{-1/2})$.

Next, the envelope theorem simplifies the FOC for $\hat{\theta}$ to
\(
	\mathbb{E}_m[\Lambda(X_{\hat{\theta}}^\top\hat{\lambda}_{\hat{\theta}})\dot{X}_{\hat{\theta}}^\top\hat{\lambda}_{\hat{\theta}}]=0
\),
whose Taylor expansion gives
\begin{multline*}
	0=\mathbb{E}_m[\Lambda(X_{\theta_0}^\top\hat{\lambda}_{\theta_0})\dot{X}_{\theta_0}^\top\hat{\lambda}_{\theta_0}]
	+\mathbb{E}_m[\Lambda(X_{\theta_0}^\top\hat{\lambda}_{\theta_0})[(1-\Lambda(X_{\theta_0}^\top\hat{\lambda}_{\theta_0}))\dot{X}_{\theta_0}^\top\hat{\lambda}_{\theta_0}\hat{\lambda}_{\theta_0}^\top\dot{X}_{\theta_0}\\
	+A+\dot{X}_{\theta_0}^\top\dot{\hat{\lambda}}_{\theta_0}]](\hat{\theta}-\theta_0)+o_P(n^{-1/2})
\end{multline*}
where
\(
	A=[
	(\frac{\partial}{\partial\theta_1}\dot{X}_{\theta_0})^\top\hat{\lambda}_{\theta_0},
	\cdots,
	(\frac{\partial}{\partial\theta_d}\dot{X}_{\theta_0})^\top\hat{\lambda}_{\theta_0}
	]
\).
As $\hat{\lambda}_{\theta_0}\to 0$ and $\dot{\hat{\lambda}}_{\theta_0}\to\dot{\lambda}_{\theta_0}$, this becomes
\[
	\sqrt{n}(\hat{\theta}-\theta_0)=-\mathbb{E}[\dot{X}_{\theta_0}^\top\dot{\lambda}_{\theta_0}]^{-1}\mathbb{E}[\dot{X}_{\theta_0}^\top]\sqrt{n}(\hat{\lambda}_{\theta_0}-0)+o_P(1).
\]
As in \cref{sec:asm:neyman}, we have
\(
	\dot{\lambda}_{\theta_0}
	=-\Omega\mathbb{E}[\dot{X}_{\theta_0}]
\),
which yields
\(
	\sqrt{n}(\hat{\theta}-\theta_0)\leadsto N(0,V_\theta)
\)
for
\(
	V_\theta\coloneqq
	(\mathbb{E}[\dot{X}_{\theta_0}^\top]\Omega\mathbb{E}[\dot{X}_{\theta_0}])^{-1}
	\mathbb{E}[\dot{X}_{\theta_0}^\top]V_\lambda\mathbb{E}[\dot{X}_{\theta_0}]
	(\mathbb{E}[\dot{X}_{\theta_0}^\top]\Omega\mathbb{E}[\dot{X}_{\theta_0}])^{-1}
\).

Meanwhile, the FOC for SMM,
\(
	\mathbb{E}_m[\dot{X}_{\tilde{\theta}}^\top]\Omega(\mathbb{E}_n[X]-\mathbb{E}_m[X_{\tilde{\theta}}])=0
\),
expands as
\[
	0
	=\mathbb{E}_m[\dot{X}_{\theta_0}^\top]\Omega(\mathbb{E}_n[X]-\mathbb{E}_m[X_{\theta_0}])
	+(B-\mathbb{E}_m[\dot{X}_{\theta_0}^\top]\Omega\mathbb{E}_m[\dot{X}_{\theta_0}])(\tilde{\theta}-\theta_0)+o_P(n^{-1/2})
\]
where
\(
	B=[
	\mathbb{E}_m[\tfrac{\partial\dot{X}_{\theta_0}}{\partial\theta_1}]^\top\Omega(\mathbb{E}_n[X]-\mathbb{E}_m[X_{\theta_0}]),\dots,
	\mathbb{E}_m[\tfrac{\partial\dot{X}_{\theta_0}}{\partial\theta_d}]^\top\Omega(\mathbb{E}_n[X]-\mathbb{E}_m[X_{\theta_0}])]
\).
Thus,
\(
	\sqrt{n}(\tilde{\theta}-\theta_0)=-(\mathbb{E}[\dot{X}_{\theta_0}^\top]\Omega\mathbb{E}[\dot{X}_{\theta_0}])^{-1}\mathbb{E}[\dot{X}_{\theta_0}^\top]\Omega\sqrt{n}(\mathbb{E}_n[X]-\mathbb{E}_m[X_{\theta_0}])+o_P(1)
\),
which shows that $\hat{\theta}$ and $\tilde{\theta}$ are asymptotically equivalent in probability as well as in distribution.

\begin{rem}
If $X$ and $X_\theta$ have a constant term and the second\hyp{}order moments are correctly specified, $V_\theta$ simplifies to
\(
	[1+\lim\frac{n}{m}](\mathbb{E}[\dot{X}_{\theta_0}^\top]\mathbb{E}[X X^\top]^{-1}\mathbb{E}[\dot{X}_{\theta_0}])^{-1}
\).
\end{rem}

\section{Convergence Rates of the Discriminator} \label{supp:theory}

This section establishes the rate of convergence of the discriminator.
In addition to results on a general nonparametric discriminator, we present results specific to a neural network discriminator.

The distance of discriminators is measured by a Hellinger\hyp{}like distance
\[
	d_\theta(D_1,D_2)\coloneqq\sqrt{h_\theta(D_1,D_2)^2+h_\theta(1-D_1,1-D_2)^2}
\]
where
\(
	h_\theta(D_1,D_2)\coloneqq\sqrt{(P_0+P_\theta)(\sqrt{\vphantom{D}\smash{D_1}}-\sqrt{\vphantom{D}\smash{D_2}})^2}
\).

The size of the neural network sieve is usually measured by the uniform and bracketing entropies.
Conceptually, the bracketing entropy gives a stronger bound than the uniform entropy and yields a tighter convergence rate. It also goes nicely with the Bernstein norm that is useful for maximal inequalities for the log likelihood ratio (as well as our discriminators).
For this, we go with the bracketing entropy.
See \cite{vw2011} for more comparison of the two entropy notions.

\begin{defn}[Bracketing number and bracketing entropy integral]
The {\em $\varepsilon$\hyp{}bracketing number} $N_{[]}(\varepsilon,\mathcal{F},d)$ of a set $\mathcal{F}$ with respect to a premetric $d$ is the minimal number of $\varepsilon$\hyp{}brackets in $d$ needed to cover $\mathcal{F}$.
The {\em $\delta$\hyp{}bracketing entropy integral} of $\mathcal{F}$ with respect to $d$ is
\(
	J_{[]}(\delta,\mathcal{F},d)\coloneqq\int_0^\delta\sqrt{1+\log N_{[]}(\varepsilon,\mathcal{F},d)}d\varepsilon
\).
\end{defn}

The results on convergence of the discriminator are stated pointwise in $\theta\in\Theta$, so the discussion is made for fixed $\theta$.
Let $\delta_n$ be a nonnegative sequence.

\subsection{General Nonparametric Discriminator} \label{supp:theory:discriminator}

Let $\mathcal{D}_{\theta,\delta}\coloneqq\{D\in\mathcal{D}_n:d_\theta(D,D_\theta)\leq\delta\}$.
We first assume that the sieve does not grow too fast.

\begin{asm}[Entropy of sieve] \label{asm:sieve:entropy}
The entropy integral satisfies $J_{[]}(\delta_n,$ $\mathcal{D}_{\theta,\delta_n},d_\theta)\lesssim\delta_n^2\sqrt{n}$.
Also, there exists $\alpha<2$ such that $J_{[]}(\delta,\mathcal{D}_{\theta,\delta},d_\theta)/\delta^\alpha$ has a majorant decreasing in $\delta>0$.
\end{asm}

The estimated discriminator need not be the exact maximizer of the loss but is required to maximize it up to some rate.

\begin{asm}[Approximately maximizing discriminator] \label{asm:discriminator}
The trained discriminator $\hat{D}_\theta$ satisfies
\(
	\mathbb{M}_\theta(\hat{D}_\theta)\geq\mathbb{M}_\theta(D_\theta)-O_P(\delta_n^2)
\).
\end{asm}

In a sense, we can interpret \cref{asm:sieve:entropy} as a requirement that the sieve be not too rich and \cref{asm:discriminator} that the sieve be rich enough.
For example, if $\mathcal{D}_{\theta,\delta_n}$ is an empty set, \cref{asm:sieve:entropy} is trivially satisfied, but there is no way to attain \cref{asm:discriminator}. %
On the contrary, if $\mathcal{D}_n$ contains every function, there would exist an element in $\mathcal{D}_n$ that satisfies \cref{asm:discriminator} but \cref{asm:sieve:entropy} will be violated.
Both assumptions collectively require that the sieve is small but good enough for $D_\theta$.
With these, we obtain the rate of convergence of the discriminator.

\begin{thm}[Rate of convergence of discriminator] \label{thm:D:rate}
Under \cref{asm:sieve:entropy,asm:m,asm:discriminator},
$d_\theta(\hat{D}_\theta,D_\theta)=O_P^\ast(\delta_n)$.
\end{thm}

One interesting observation is that \cref{thm:D:rate} does not require convergence of the objective function.
This is reminiscent of the nonparametric maximum likelihood literature.
To prove it without requiring convergence of the objective function, we think in terms of a pseudo\hyp{}objective function.
Let $m^p_q\coloneqq\log\frac{p+q}{2q}$ and
\[
	\tilde{M}_\theta(D)\coloneqq P_0 m^D_{D_\theta}+P_\theta m^{1-D}_{1-D_\theta}, \qquad
	\tilde{\mathbb{M}}_\theta(D)\coloneqq\mathbb{P}_0 m^D_{D_\theta}+\mathbb{P}_\theta m^{1-D}_{1-D_\theta}.
\]

\begin{proof}%
The concavity of the logarithm and \cref{asm:discriminator} imply
\(
	\tilde{\mathbb{M}}_\theta(\hat{D}_\theta)-\tilde{\mathbb{M}}_\theta(D_\theta)\geq\tfrac{1}{2}[\mathbb{M}_\theta(\hat{D}_\theta)-\mathbb{M}_\theta(D_\theta)]\geq-O_P(\delta_n^2)
\).
Then, apply \citet[Theorem 3.4.1]{vw1996} with \cref{lem:maximal,asm:sieve:entropy}.
\end{proof}

The following is a maximal inequality used to prove \cref{thm:D:rate}.
Let $\mathcal{M}_{\theta,\delta}^1\coloneqq\{m^{D}_{D_\theta}:D\in\mathcal{D}_{\theta,\delta}\}$ and $\mathcal{M}_{\theta,\delta}^2\coloneqq\{m^{1-D}_{1-D_\theta}:D\in\mathcal{D}_{\theta,\delta}\}$.

\begin{lem}[Maximal inequality for pseudo\hyp{}cross\hyp{}entropy discriminator] \label{lem:maximal}
For every $D\in\mathcal{D}$,
\(
	\tilde{M}_\theta(D)-\tilde{M}_\theta(D_\theta)\leq-d_\theta(D,D_\theta)^2/(1+\sqrt{2})^2
\).
For every $\delta>0$,
\begin{multline*}
	\mathbb{E}^\ast\sup_{D\in\mathcal{D}_{\theta,\delta}}\sqrt{n}\bigl|(\tilde{\mathbb{M}}_\theta-\tilde{M}_\theta)(D)-(\tilde{\mathbb{M}}_\theta-\tilde{M}_\theta)(D_\theta)\bigr|\\
	\lesssim J_{[]}(\delta,\mathcal{D}_{\theta,\delta},d_\theta)\bigl[1+\sqrt{\tfrac{n}{m}}+(1+\tfrac{n}{m})\tfrac{J_{[]}(\delta,\mathcal{D}_{\theta,\delta},d_\theta)}{\delta^2\sqrt{n}}\bigr].
\end{multline*}
\end{lem}

\begin{proof}%
Since $\log x\leq2(\sqrt{x}-1)$ for every $x>0$,
\begin{multline*}
	P_0\log\tfrac{D}{D_\theta}\leq2P_0\bigl(\sqrt{\tfrac{D}{D_\theta}}-1\bigr)
	=\Bigl[2P_0\tfrac{\sqrt{D(p_0+p_\theta)}}{\sqrt{p_0}}-\int D(p_0+p_\theta)-\int p_0\Bigr]\\
	+(P_0+P_\theta)(D-D_\theta)
	=-h_\theta(D,D_\theta)^2+(P_0+P_\theta)(D-D_\theta).
\end{multline*}
Similarly,
\(
	P_\theta\log\frac{1-D}{1-D_\theta}\leq-h_\theta(1-D,1-D_\theta)^2-(P_0+P_\theta)(D-D_\theta)
\).
Replacing $D$ and $1-D$ with $(D+D_\theta)/2$ and $(1-D+1-D_\theta)/2$ and summing them up yield
\[
	P_0 m^D_{D_\theta}+P_\theta m^{1-D}_{1-D_\theta}\leq-h_\theta\bigl(\tfrac{D+D_\theta}{2},D_\theta\bigr)^2-h_\theta\bigl(\tfrac{1-D+1-D_\theta}{2},1-D_\theta\bigr)^2.
\]
Since $\sqrt{2}h_\theta(\frac{p+q}{2},q)\leq h_\theta(p,q)\leq(1+\sqrt{2})h_\theta(\frac{p+q}{2},q)$ \citep[Problem 3.4.4]{vw1996}, we obtain the first inequality.
For the second inequality, observe that
\[
	\sqrt{n}\bigl[(\tilde{\mathbb{M}}_\theta-\tilde{M}_\theta)(D)-(\tilde{\mathbb{M}}_\theta-\tilde{M}_\theta)(D_\theta)\bigr]
	=\sqrt{n}(\mathbb{P}_0-P_0)m^D_{D_\theta}+\sqrt{n}(\mathbb{P}_\theta-P_\theta)m^{1-D}_{1-D_\theta}.
\]
Therefore, it suffices to separately bound
\[
	\mathbb{E}^\ast\sup_{D\in\mathcal{D}_{\theta,\delta}}\bigl|\sqrt{n}(\mathbb{P}_0-P_0)m^D_{D_\theta}\bigr|
	\quad\text{and}\quad
	\sqrt{\tfrac{n}{m}}\,\mathbb{E}^\ast\sup_{D\in\mathcal{D}_{\theta,\delta}}\bigl|\sqrt{m}(\mathbb{P}_\theta-P_\theta)m^{1-D}_{1-D_\theta}\bigr|.
\]
Since $m^D_{D_\theta},m^{1-D}_{1-D_\theta}\geq\log(1/2)$ and $e^{|x|}-1-|x|\leq4(e^{x/2}-1)^2$ for every $x\geq\log(1/2)$,
\begin{gather*}
	\|m^D_{D_\theta}\|_{P_0,B}^2\leq 8 P_0(e^{m^D_{D_\theta}/2}-1)^2\leq 8 h_\theta(\tfrac{D+D_\theta}{2},D_\theta)^2\leq4h_\theta(D,D_\theta)^2,\\
	\|m^{1-D}_{1-D_\theta}\|_{P_\theta,B}^2\leq4h_\theta(1-D,1-D_\theta)^2.
\end{gather*}
By \citet[Lemma 3.4.3]{vw1996}, the first supremum is bounded by
\(
	J_{[]}(2\delta,\mathcal{M}_{\theta,\delta}^1,\|\cdot\|_{P_0,B})[1+J_{[]}(2\delta,\mathcal{M}_{\theta,\delta}^1,\|\cdot\|_{P_0,B})/(4\delta^2\sqrt{n})].
\)
Let $[\ell,u]$ be an $\varepsilon$\hyp{}bracket in $\mathcal{D}$ with respect to $d_\theta$.
Since $u-\ell\geq0$ and $e^{|x|}-1-|x|\leq2(e^{x/2}-1)^2$ for $x\geq0$,
\begin{align*}
	\bigl\|m^u_{D_\theta}-m^\ell_{D_\theta}\bigr\|_{P_0,B}^2\leq4\int\Bigl(\sqrt{\tfrac{u+D_\theta}{\ell+D_\theta}}-1\Bigr)^2p_0
	&\leq4\int\bigl(\sqrt{\vphantom{D}\smash{u+D_\theta}}-\sqrt{\vphantom{D}\smash{\ell+D_\theta}}\bigr)^2(p_0+p_\theta)\\
	&\leq4h_\theta(u,\ell)^2\leq4\varepsilon^2.
\end{align*}
Thus, $[m^\ell_{D_\theta},m^u_{D_\theta}]$ makes a $2\varepsilon$\hyp{}bracket in $\mathcal{M}_{\theta,\delta}^1$ with respect to $\|\cdot\|_{P_0,B}$, so $J_{[]}(2\delta,\mathcal{M}_{\theta,\delta}^1,\|\cdot\|_{P_0,B})\leq 2 J_{[]}(\delta,\mathcal{D}_{\theta,\delta},d_\theta)$.
Analogous argument for the second supremum yields the second inequality.
\end{proof}

\subsection{Cross\hyp{}Entropy Loss}

To show convergence of the objective function, we need to make an additional assumption that the tails of the discriminators in the sieve are not too thin.
This assumption would be trivial if we assume a compact support for the observables $X_i$ and $X_{i,\theta}$, which is standard in the neural network literature.

\begin{asm}[Support compatibility] \label{asm:sieve:bound}
Define $P(X|A)$ to be $P(X\mathbbm{1}\{A\})/P(A)$ if $P(A)>0$ and $0$ otherwise.
There exists $M$ such that
\[
	\sup_{D\in\mathcal{D}_{\theta,\delta_n}}P_0\bigl(\tfrac{D_\theta}{D}\bigm|\tfrac{D_\theta}{D}\geq\tfrac{25}{16}\bigr)<M, \quad
	\sup_{D\in\mathcal{D}_{\theta,\delta_n}}P_\theta\bigl(\tfrac{1-D_\theta}{1-D}\bigm|\tfrac{1-D_\theta}{1-D}\geq\tfrac{25}{16}\bigr)<M.
\]
Also, the brackets $\{\ell\leq D\leq u\}$ in \cref{asm:sieve:entropy} can be taken so that
\(
	(P_0+P_\theta)(\frac{D_\theta}{\ell}(\sqrt{u}-\sqrt{\ell})^2)
\)
and
\(
	(P_0+P_\theta)(\frac{1-D_\theta}{1-u}(\sqrt{1-\ell}-\sqrt{1-u})^2)
\)
are $O(d_\theta(u,\ell)^2)$.
\end{asm}

With this, we obtain the rate for the estimated cross\hyp{}entropy loss.

\begin{thm}[Rate of convergence of objective function] \label{thm:obj:rate}
Under \cref{asm:sieve:entropy,asm:sieve:bound,asm:m,asm:discriminator},
\(
	\mathbb{M}_\theta(\hat{D}_\theta)-\mathbb{M}_\theta(D_\theta)=O_P^\ast(\delta_n^2)
\).
\end{thm}

\begin{proof}%
Since $\mathbb{M}_\theta(\hat{D}_\theta)-\mathbb{M}_\theta(D_\theta)\geq-O_P(\delta_n^2)$ by \cref{asm:discriminator}, we need only to prove the reverse inequality.
With $\log(x)\leq2(\sqrt{x}-1)$ for $x>0$, for every $D$,
\begin{multline*}
	\mathbb{M}_\theta(D)-\mathbb{M}_\theta(D_\theta)\\
	\leq 2 P_0\bigl(\sqrt{\tfrac{D}{D_\theta}}-1\bigr)+2P_\theta\bigl(\sqrt{\tfrac{1-D}{1-D_\theta}}-1\big)
	+(\mathbb{P}_0-P_0)\log\tfrac{D}{D_\theta}+(\mathbb{P}_\theta-P_\theta)\log\tfrac{1-D}{1-D_\theta}.
\end{multline*}
As in \cref{lem:maximal}, the first two terms are equal to $-d_\theta(D,D_\theta)^2$.
Since \cref{thm:D:rate} implies $d_\theta(\hat{D}_\theta,D_\theta)^2=O_P^\ast(\delta_n^2)$, it remains to show that the last two terms are of the same order.
We bound the suprema,
\[
	\mathbb{E}^\ast\sup_{D\in\mathcal{D}_{\theta,\delta_n}}\bigl|\sqrt{n}(\mathbb{P}_0-P_0)\log\tfrac{D}{D_\theta}\bigr|
	\quad\text{and}\quad
	\mathbb{E}^\ast\sup_{D\in\mathcal{D}_{\theta,\delta_n}}\bigl|\sqrt{m}(\mathbb{P}_\theta-P_\theta)\log\tfrac{1-D}{1-D_\theta}\bigr|.
\]
Under \cref{asm:sieve:bound}, it follows from (the remark after) \cref{lem:vaart} that for $D\in\mathcal{D}_{\theta,\delta_n}$,
\[
	\bigl\|\tfrac{1}{2}\log\tfrac{D}{D_\theta}\bigr\|_{P_0,B}^2\leq2(1+M)h_\theta(D,D_\theta)^2, \
	\bigl\|\tfrac{1}{2}\log\tfrac{1-D}{1-D_\theta}\bigr\|_{P_\theta,B}^2\leq2(1+M)h_\theta(1-D,1-D_\theta)^2.
\]
\cref{asm:sieve:bound} also implies that an $\varepsilon$\hyp{}bracket in $\mathcal{M}_{\theta,\delta}^1$ induces
\begin{gather*}
	\bigl\|\log\tfrac{u}{D_\theta}-\log\tfrac{\ell}{D_\theta}\bigr\|_{P_0,B}^2\leq4P_0\bigl(\sqrt{\tfrac{u}{\ell}}-1\bigr)^2=4(P_0+P_\theta)\tfrac{D_\theta}{\ell}(\sqrt{u}-\sqrt{\ell})^2\leq C d_\theta(u,\ell)^2,\\
	\bigl\|\log\tfrac{1-\ell}{1-D_\theta}-\log\tfrac{1-u}{1-D_\theta}\bigr\|_{P_\theta,B}^2\leq4(P_0+P_\theta)\tfrac{1-D_\theta}{1-u}(\sqrt{1-\ell}-\sqrt{1-u})^2\leq C d_\theta(u,\ell)^2,
\end{gather*}
for some $C>0$.
By similar arguments as in the proof of \cref{lem:maximal}, the two suprema are of orders $\sqrt{n}\delta_n^2$ and $\sqrt{m}\delta_n^2$.%
\footnote{We can write $\|\frac{1}{2}\log\frac{D}{D_\theta}\|_{P_0,B}^2\leq[2(1+M)\vee C]h_\theta(D,D_\theta)^2$ and $\|\log\frac{u}{D_\theta}-\log\frac{\ell}{D_\theta}\|_{P_0,B}^2\leq[2(1+M)\vee C]d_\theta(u,\ell)^2$ to apply the same argument as \cref{thm:D:rate}.}
With \cref{asm:m} follows the theorem.
\end{proof}

\subsection{Neural Network Discriminator}

The results in \cref{supp:theory:discriminator} apply to any nonparametric sieve discriminator.
Given a particular sieve, the specific convergence rate is determined by the $\delta_n$ that satisfies \cref{asm:sieve:entropy}.
In the nonparametric estimation literature, it is often observed that $\delta_n$ gets slower as the dimension $d$ of the input $X_i$ increases.
In the context of nonparametric regression, \citet{bk2019} show that a particular type of neural network estimator does not have a rate that slows with $d$ but only with $d^\ast$, the ``underlying dimension'' of the target function.%
\footnote{\citet{bk2019} call $d^\ast$ the {\em order}.}
We believe that the structure they impose on the target function arises very naturally in economic models, and want to incorporate the ``remedy for the curse of dimensionality'' aspect into our theory.%

In light of this, we develop the ``classification counterpart'' of the results in \citet{bk2019}.
Instead of the target regression function, we exploit the low\hyp{}dimensional composite structure on the log likelihood ratio $\log(p_0/p_\theta)$.
We note that our theory does not {\em require} that there is such a low\hyp{}dimensional structure; if there is none, we have $d^\ast=d$ and our result reduces to a regular nonparametric rate with the curse of dimensionality.

Intuitively, the low\hyp{}dimensional composite structure is described as follows.
Note that the log likelihood ratio $\log(p_0/p_\theta)$ takes a $d$\hyp{}dimensional input $X$ as its argument, where $d$ can be large.
We need that this ratio admits a representation as a nested composition of smooth functions, each of which takes a possibly smaller number $d^\ast$ of arguments.
In the first layer of composition, we assume a linear index structure to reduce $d$ arguments into $d^\ast$ intermediate outputs.

To develop a precise definition, we start with the notion of smoothness we use.

\begin{defn}[$(p,C)$\hyp{}smoothness; {\citealp[Definition 1]{bk2019}}]
Let $p=q+s$ for some $q\in\mathbb{N}_0$ and $0<s\leq 1$. A function $m:\mathbb{R}^d\to\mathbb{R}$ is called {\em $(p,C)$\hyp{}smooth} if for every $\alpha=(\alpha_1,\dots,\alpha_d)\in\mathbb{N}_0^d$ with $\sum_{j=1}^d\alpha_j=q$, the partial derivative $\frac{\partial^q m}{\partial x_1^{\alpha_1}\cdots\partial x_d^{\alpha_d}}$ exists and satisfies
\[
	\biggl|\frac{\partial^q m}{\partial x_1^{\alpha_1}\cdots\partial x_d^{\alpha_d}}(x)-\frac{\partial^q m}{\partial x_1^{\alpha_1}\cdots\partial x_d^{\alpha_d}}(z)\biggr|\leq C\|x-z\|^s
\]
for every $x,z\in\mathbb{R}^d$ where $\|\cdot\|$ denotes the Euclidean norm.
\end{defn}

With this, the nested composition structure is defined as follows.

\begin{defn}[Generalized hierarchical interaction model; {\citealp[Definition 2]{bk2019}}]
Let $d \in\mathbb{N}$, $d^\ast \in \{1,\dots,d\}$, and $m : \mathbb{R}^d \rightarrow \mathbb{R}$.
We say that the function $m$ satisfies a {\em generalized hierarchical interaction model of order $d^\ast$ and level $0$}, if there exist $a_1\in\mathbb{R}^d,\ldots, a_{d^\ast}\in\mathbb{R}^d$, and $f : \mathbb{R}^{d^\ast} \rightarrow \mathbb{R}$ such that
\[
m(x) = f(a_1^\top x,\ldots,a_{d^\ast}^\top x)
\]
for every $x \in \mathbb{R}^d$.
We say that $m$ satisfies a {\em generalized hierarchical interaction model of order $d^\ast$ and level $l + 1$ with $K$ components} if there exist $g_k: \mathbb{R}^{d^\ast} \rightarrow \mathbb{R}$ and $f_{1,k},\dots,f_{d^\ast,k}: \mathbb{R}^{d} \rightarrow \mathbb{R}$ $(k = 1,\ldots,K)$ such that $f_{1,k},\ldots,f_{d^\ast,k}$ $(k = 1,\dots,K)$ satisfy a generalized hierarchical model of order $d^\ast$ and level $l$ and 
\[
m(x) = \sum_{k=1}^K g_k(f_{1,k}(x),\dots,f_{d^\ast,k}(x))
\]
for every $x \in \mathbb{R}^d$.
We say that the generalized hierarchical interaction model is {\em $(p,C)$\hyp{}smooth} if all functions occurring in its definition are $(p,C)$\hyp{}smooth.
\end{defn}

For example, a conditional binary choice model yields a log likelihood ratio that satisfies a generalized hierarchical interaction model of order $d^\ast\leq 3$ and level $0$, irrespectively of the dimension of the covariates.

\begin{exa}[Binary choice model]
Let $y_i=\mathbbm{1}\{x_i^\top\alpha+\varepsilon_i>0\}$, $\varepsilon_i\sim P_\varepsilon$, be the true DGP and $y_i=\mathbbm{1}\{x_i^\top\beta+\tilde{\varepsilon}_i>0\}$, $\tilde{\varepsilon}_i\sim\tilde{P}_\varepsilon$, be the structural model.
Then,
\[
	\log\frac{p_0(y,x)}{p_\theta(y,x)}=y\log\frac{1-P_\varepsilon(-x^\top\alpha)}{1-\tilde{P}_\varepsilon(-x^\top\beta)}+(1-y)\log\frac{P_\varepsilon(-x^\top\alpha)}{\tilde{P}_\varepsilon(-x^\top\beta)}.
\]
Therefore, we can write this as $f(a_1^\top z,a_2^\top z,a_3^\top z)$ where $z=(y,x^\top)^\top$, $a_1=(1,0,\dots,0)^\top$, $a_2=(0,-\alpha^\top)^\top$, $a_3=(0,-\beta^\top)^\top$, and %
\(
	f(y,x_1,x_2)=y[\log(1-P_\varepsilon(x_1))-\log(1-\tilde{P}_\varepsilon(x_2))]+(1-y)[\log P_\varepsilon(x_1)-\log\tilde{P}_\varepsilon(x_2)]
\).
\end{exa}

Neural networks approximate functions by a nested composition of activation functions.
For theoretical development, we define the following structure on the neural network estimator.

\begin{defn}[Hierarchical neural network; {\citealp[Section 2]{bk2019}}]
Let $\sigma:\mathbb{R}\to\mathbb{R}$ be a $q$\hyp{}admissible activation function.
For $M^\ast\in\mathbb{N}$, $d\in\mathbb{N}$, $d^\ast\in\{1,\dots,d\}$, and $\alpha>0$, let $\mathcal{F}_{M^\ast,d^\ast,d,\alpha}$ be the class of functions $f:\mathbb{R}^d\to\mathbb{R}$ such that
\[
	f(x)=\sum_{i=1}^{M^\ast}\mu_i\sigma\biggl(\sum_{j=1}^{4d^\ast}\lambda_{i,j}\sigma\biggl(\sum_{v=1}^d\theta_{i,j,v}x_v+\theta_{i,j,0}\biggr)+\lambda_{i,0}\biggr)+\mu_0
\]
for some $\mu_i,\lambda_{i,j},\theta_{i,j,v}\in\mathbb{R}$, where $|\mu_i|\leq\alpha$, $|\lambda_{i,j}|\leq\alpha$, and $|\theta_{i,j,v}|\leq\alpha$.
For $l=0$, define the set of neural networks with two hidden layers by $\mathcal{H}_{M^\ast,d^\ast,d,\alpha}^{(0)}\coloneqq\mathcal{F}_{M^\ast,d^\ast,d,\alpha}$; for $l>0$, define the set of neural networks with $2l+2$ hidden layers by
\begin{multline*}
	\mathcal{H}_{M^\ast,d^\ast,d,\alpha}^{(l)}\coloneqq\\
	\biggl\{h:\mathbb{R}^d\to\mathbb{R}:h(x)=\sum_{k=1}^K g_k(f_{1,k}(x),\dots,f_{d^\ast,k}(x)),
	g_k\in\mathcal{F}_{M^\ast,d^\ast,d^\ast,\alpha},\ f_{j,k}\in\mathcal{H}^{(l-1)}\biggr\}.
\end{multline*}
\end{defn}

Now, we assume that the log likelihood ratio admits a hierarchical representation and that the neural network has a corresponding hierarchical structure.

\begin{asm}[Neural network discriminator] \label{asm:nn}
Let $P_0$ and $P_\theta$ have subexponential tails and finite first moments.%
\footnote{We say that $P$ on $\mathbb{R}^d$ has {\em subexponential tails} if $\log P(\|X\|_\infty>a)\lesssim-a$ for large $a$.}
Let $\log(p_0/p_\theta)$ satisfy a $(p,C)$\hyp{}smooth generalized hierarchical interaction model of order $d^\ast$ and finite level $l$ with $K$ components for $p=q+s$, $q\in\mathbb{N}_0$, and $s\in(0,1]$.
Let $\mathcal{H}_{M^\ast,d^\ast,d,\alpha}^{(l)}$ be the class of neural networks with the Lipschitz activation function with Lipschitz constant $1$
for
\begin{gather*}
	M_\ast=\biggl\lceil{d^\ast+q\choose d^\ast}(q+1)\biggl(\biggl[\frac{(\log\delta_n)^{2(2q+3)}}{\delta_n}\biggr]^{\frac{1}{p}}+1\biggr)^{d^\ast}\biggr\rceil,\\
	\alpha=\biggl[\frac{(\log\delta_n)^{2(2q+3)}}{\delta_n}\biggr]^{\frac{d^\ast+p(2q+3)+1}{p}}\frac{\log n}{\delta_n^2},
\end{gather*}
and
\(
	\delta_n=[(\log n)^{\frac{p+2d^\ast(2q+3)}{p}}/n]^{\frac{p}{2p+d^\ast}}
\).
Denote by $\mathcal{D}_n\coloneqq\{\Lambda(f):f\in\mathcal{H}_{M^\ast,d^\ast,d,\alpha}^{(l)}\}$ the sieve of neural network discriminators for the standard logistic cdf $\Lambda$.
\end{asm}

\cref{asm:nn} gives a sufficient condition for \cref{asm:sieve:entropy}, %
so we use this to derive the rate of convergence of the neural network discriminator.
If, in addition, $d^\ast<2p$, we have $\delta_n=o_P(n^{-1/4})$; this is easier to satisfy if the underlying dimension of the log likelihood ratio is low, regardless of the dimension of the input.

\begin{prop}[Rate of convergence of neural network discriminator] \label{prop:nn:rate}
Under \cref{asm:nn,asm:m,asm:discriminator}, $d_\theta(\hat{D}_\theta,D_\theta)=O_P^\ast(\delta_n)$.
\end{prop}

\begin{proof}%
We use \cref{lem:bracket} to bound the bracketing number in \cref{asm:sieve:entropy}.
For now, let us assume that $\mathcal{D}_n$ in \cref{asm:nn} satisfies the network structure of \cref{lem:bracket}; later, we calibrate the constants in reflection of the network structure in \cref{asm:nn}.
Since $D$ is nonnegative, we can extend $d_\theta$ to accommodate arbitrary functions $f_1$ and $f_2$ by $d_\theta(f_1,f_2)\coloneqq d_\theta(0\vee f_1,0\vee f_2)$.
In the notation of \cref{lem:bracket},
\begin{align*}
	\|\varepsilon^2F\|_{d_\theta}^2&=\sup_{D\in\mathcal{D}}d_\theta(D-\varepsilon^2F/2,D+\varepsilon^2F/2)^2\leq h_\theta(0,\varepsilon^2F)^2+h_\theta(0,\varepsilon^2F)^2\\
	&=2\varepsilon^2(P_0+P_\theta)F=2\varepsilon^2[2\sigma_0+(P_0+P_\theta)\|X\|_\infty]\eqqcolon B\varepsilon^2.
\end{align*}
Since $P_0$ and $P_\theta$ have bounded first moments, $B<\infty$.
Replacing $\varepsilon$ with $\varepsilon/\sqrt{B}$ yields $\|\frac{\varepsilon^2}{B}F\|_{d_\theta}\leq\varepsilon$.
Therefore, with \cref{lem:bracket},
\[
	\log N_{[]}(\varepsilon,\mathcal{D}_n,d_\theta)\leq\log N_{[]}(\|\tfrac{\varepsilon^2}{B}F\|_{d_\theta},\mathcal{D}_n,d_\theta)\leq S\log\bigl\lceil\tfrac{2 B(L+1)(\tilde{U}C)^{L+1}d}{\varepsilon^2}\bigr\rceil.
\]
Observe that for $0<\delta\leq e^a$,
\[
	\int_0^\delta\sqrt{1+a-\log\varepsilon}d\varepsilon=\tfrac{\sqrt{\pi}e^a}{2}\erfc\bigl(\sqrt{1+a-\log\delta}\bigr)+\delta\sqrt{1+a-\log\delta}\lesssim\delta\sqrt{1+a-\log\delta}.
\]
Therefore,
\begin{multline*}
	J_{[]}(\delta,\mathcal{D}_n,h_\theta)
	\lesssim\int_0^\delta\sqrt{1+S[\log(2 B(L+1)(\tilde{U}C)^{L+1}d)-2\log\varepsilon]_+}\,d\varepsilon\\
	\lesssim\delta\sqrt{1+S[\log(2 B(L+1)(\tilde{U}C)^{L+1}d)-2\log\delta]_+}
	\lesssim\delta\sqrt{1\vee[SL\log(\tilde{U}C)-S\log\delta]}.
\end{multline*}
Therefore, if we set
\begin{equation} \label{eq:delta}
	\delta_n=O\Bigl(\sqrt{\tfrac{SL\log(\tilde{U}C)+S\log n}{n}}\Bigr),
\end{equation}
$\mathcal{D}_n$ satisfies \cref{asm:sieve:entropy} with $\alpha=1.5$.
Now, we must choose $S$, $L$, $\tilde{U}$, and $C$ so that this rate is attainable and fast.
For the rate to be attainable, we must also have \cref{asm:discriminator}, for which we need that $\mathcal{D}_{\theta,\delta}$ is nonempty.
That is, the sieve $\mathcal{D}_n$ must contain an element in the $\delta_n$\hyp{}neighborhood of $D_\theta$, i.e., $\inf_{D\in\mathcal{D}_n}d_\theta(D,D_\theta)\lesssim\delta_n$.

Since $\mathcal{D}_n=\Lambda(\mathcal{H}^{(l)})$, we use \citet[Theorem 3]{bk2019} to find the network configuration that attains this inequality.
For this, we need to choose ``$N$, $\eta_n$, $a_n$, $M_n$'' in their notation; in doing so, we find ``$S$, $L$, $\tilde{U}$, $C$'' in our notation.
First, we set $N=q$ and $\eta_n=\delta_n^2$.
By subexponentiality, we have $\log P_0(\|X\|_\infty>a)+\log P_\theta(\|X\|_\infty>a)\lesssim-a$ for large $a$.
Therefore, we want $a_n\gg-2\log\delta_n$ so that the remainder term in \citet[Theorem 3]{bk2019} is small enough, that is, $(P_0+P_\theta)(\|X\|_\infty>a_n)\lesssim\delta_n^2$.%
\footnote{If we set $a_n\sim-2\log\delta_n$, we can only say $(P_0+P_\theta)(\|X\|_\infty>a_n)\lesssim\delta_n^c$ for some $c$.}
We can do this by setting, e.g., $a_n=(-\log\delta_n)^2$.
Finally, we want to choose $M_n$ so that $a_n^{N+q+3}M_n^{-p}\sim\delta_n$ since then \citet[Theorem 3]{bk2019} can bound the supremum term that appears below; set $M_n=(\log\delta_n)^{2(N+q+3)/p}/\delta_n^{1/p}$.
Let $A\subset[-a_n,a_n]^d$ be the set for which $(P_0+P_\theta)(A)\leq c\eta_n$ in \citet[Theorem 3]{bk2019}.
Then,
\begin{multline*}
	h_\theta(D,D_\theta)^2\leq\Bigl(\int_{\|x\|_\infty>a_n}+\int_A+\int_{\{\|x\|_\infty\leq a_n\}\setminus A}\Bigr)(\sqrt{D}-\sqrt{\vphantom{D}\smash{D_\theta}})^2(p_0+p_\theta)\\
	\leq(P_0+P_\theta)(\|X\|_\infty>a_n)+(P_0+P_\theta)(A)+\int_{\{\|x\|_\infty\leq a_n\}\setminus A}(\sqrt{D}-\sqrt{\vphantom{D}\smash{D_\theta}})^2(p_0+p_\theta).
\end{multline*}
The first two terms are bounded by $\delta_n^2+c\delta_n^2$.
For $D=\Lambda(f)$,
\begin{multline*}
	\int_{\{\|x\|_\infty\leq a_n\}\setminus A}(\sqrt{D}-\sqrt{\vphantom{D}\smash{D_\theta}})^2(p_0+p_\theta)
	=\int_{\{\|x\|_\infty\leq a_n\}\setminus A}\bigl(\sqrt{\Lambda\smash{(f)}}-\sqrt{\Lambda\smash{(\Lambda^{-1}\circ D_\theta)}}\bigr)^2(p_0+p_\theta)\\
	\leq\tfrac{2}{27}\|f-\Lambda^{-1}\circ D_\theta\|_{\infty,\{\|x\|_\infty\leq a_n\}\setminus A}^2=\tfrac{2}{27}\|f-\log\tfrac{p_0}{p_\theta}\|_{\infty,\{\|x\|_\infty\leq a_n\}\setminus A}^2,
\end{multline*}
since $\sqrt{\Lambda\smash{(\cdot)}}$ is Lipschitz with constant $1/(3\sqrt{3})$.
We may likewise bound $h_\theta(1-D,1-D_\theta)^2$. %
By \citet[Theorem 3]{bk2019}, $\inf_{f\in\mathcal{H}^{(l)}}\|f-\log\frac{p_0}{p_\theta}\|_{\infty,\{\|x\|_\infty\leq a_n\}\setminus A}\lesssim\delta_n$.
Thus, we obtain
\(
	\inf_{D\in\mathcal{D}_n}d_\theta(D,D_\theta)\lesssim\delta_n
\).

These configurations can be translated into our constants as $S=O(dd^\ast M_\ast K^l)\sim M_\ast$, $\tilde{U}=M_\ast\vee(4d^\ast)\vee K\sim M_\ast$, $C=\alpha$, and $L=2+3l=O(1)$, where \citet[Theorem 3]{bk2019} define
\begin{gather*}
	M_\ast={d^\ast+N\choose d^\ast}(N+1)(M_n+1)^{d^\ast}\sim M_n^{d^\ast}=\frac{(\log\delta_n)^{2d^\ast(N+q+3)/p}}{\delta_n^{d^\ast/p}},\\
	\alpha=\frac{M_n^{d^\ast+p(2N+3)+1}}{\eta_n}\log n=\frac{(\log\delta_n)^{2(N+q+3)[d^\ast+p(2N+3)+1]/p}}{\delta_n^{2+[d^\ast+p(2N+3)+1]/p}}\log n.
\end{gather*}
With these, (\ref{eq:delta}) becomes
\(
	\delta_n^2\sim M_\ast\frac{\log(M_\ast\alpha)+\log n}{n}
	\sim[(\log n)^{\frac{p+2d^\ast(N+q+3)}{p}}/n]^{\frac{p}{2p+d^\ast}}
\).
The result follows by substituting $N=q$ and invoking \cref{thm:D:rate}.
\end{proof}

The following lemma bounds the bracketing number of a (possibly sparse) neural network with bounded weights and Lipschitz activation functions.
The notation of the neural network is defined as follows.
Denote the hidden\hyp{}layer activation function by $\sigma:\mathbb{R}\to\mathbb{R}$ and the output activation function by $\Lambda:\mathbb{R}\to\mathbb{R}$.
Let $L$ be the number of hidden and output layers.
Let $w_{\ell ij}$ be the weight for the $i$th node in the $(\ell+1)$th layer on the $j$th node in the $\ell$th layer; for example, the input to the second node in the first layer is $w_{021}x_1+\cdots+w_{02U}x_U$, where $X=(x_1,\dots,x_U)$ is the input to the network.
Let $w_{\ell i}=(w_{\ell i1},\dots,w_{\ell iU})^\top$ be the column vector of weights for the $i$th node in the $(\ell+1)$th layer.
Let $w_\ell=(w_{\ell1},\dots,w_{\ell U})$ be the matrix with columns $w_{\ell i}$; note that for $\ell=L$, $w_L$ is just a column vector as there is only one output.
Let $w$ be the vector of all parameters.
Then, the discriminator is given by%
\footnote{If we include a constant input and a constant node (also known as the ``bias'' term), it is assumed to be already incorporated in $X$ and $w$.}
\[
	D(X;w)=\Lambda(w_L^\top\sigma(w_{L-1}^\top\sigma(\cdots w_1^\top\sigma(w_0^\top X)))),
\]
where $\sigma(\cdot)$ for a vector argument is elementwise application.

\begin{lem}[Bracketing number of neural network with bounded weights] \label{lem:bracket}
Let $\mathcal{F}$ be a class of neural networks defined as above. %
Denote the total number of nonzero weights by $S$ and the maximum number of nonzero weights in each node (except for the first layer taking inputs) by $\tilde{U}$.%
\footnote{The number of nonzero elements in each row of each matrix $w_\ell$, $\ell\geq1$, is bounded by $\tilde{U}$.}
Assume that $\sigma$ and $\Lambda$ are Lipschitz with constant $1$ and $\|w\|_\infty\leq C$ for some $C$.
Assume innocuously that $\tilde{U}C\geq2$ and let $\sigma_0\coloneqq|\sigma(0)|$.
Define an envelope $F:\mathbb{R}^d\to\mathbb{R}$ by $F(x)\coloneqq\sigma_0+\|x\|_\infty$.
Then, for every premetric $d_{\mathcal{F}}$ and $\|f\|_{d_{\mathcal{F}}}\coloneqq\sup_{g\in\mathcal{F}}d_{\mathcal{F}}(g-f/2,g+f/2)$,
\[
	N_{[]}(\|\varepsilon F\|_{d_{\mathcal{F}}},\mathcal{F},d_{\mathcal{F}})\leq\biggl\lceil\frac{2(L+1)(\tilde{U}C)^{L+1}d}{\varepsilon}\biggr\rceil^S.
\]
For a fully connected network, $\tilde{U}=U$ and $S=(LU+1)U+(d-U)U$. 
For a hierarchical network in \citet{bk2019}, $S=O(\tilde{U}^{(L+4)/3}d)$.
\end{lem}

\begin{proof}%
The neural network is expressed as
\(
	f(x;w)=\Lambda(w_L^\top\sigma(w_{L-1}^\top\sigma(\cdots w_1^\top\sigma(w_0^\top x))))
\).
We can bound the outputs of the $\ell$th layer by
\begin{align*}
	\|\sigma(w_{\ell-1}^\top\sigma(\cdots))\|_\infty&\leq\sigma_0+\|w_{\ell-1}^\top\sigma(\cdots)\|_\infty\leq\sigma_0+\tilde{U}C\|\sigma(\cdots)\|_\infty\\
	&\leq[1+\tilde{U}C+\cdots+(\tilde{U}C)^{\ell-1}]\sigma_0+\tilde{U}^{\ell-1}C^\ell d\|x\|_\infty\\
	&\leq\tilde{U}^{\ell-1}C^\ell(\tilde{U}\sigma_0+d\|x\|_\infty)\leq(\tilde{U}C)^\ell d(\sigma_0+\|x\|_\infty),
\end{align*}
where the fourth inequality holds for $\tilde{U}C\geq2$.
For two sets of weights, $w$ and $\tilde{w}$,
\begin{align*}
	|f(x;w)-f(x;\tilde{w})|
	&\leq\tilde{U}\|w_L-\tilde{w}_L\|_\infty(\|\sigma(w_{L-1}^\top\sigma(\cdots))\|_\infty\vee\|\sigma(\tilde{w}_{L-1}^\top\sigma(\cdots))\|_\infty)\\
	&\hspace{80pt}+\tilde{U}C\|\sigma(w_{L-1}^\top\sigma(\cdots))-\sigma(\tilde{w}_{L-1}^\top\sigma(\cdots))\|_\infty\\
	&\leq\tilde{U}^{L+1}C^Ld\|w_L-\tilde{w}_L\|_\infty(\sigma_0+\|x\|_\infty)+\cdots\\
	&{}+\tilde{U}^{L+1}C^Ld\|w_1-\tilde{w}_1\|_\infty(\sigma_0+\|x\|_\infty)+\tilde{U}^LC^Ld\|w_0-\tilde{w}_0\|_\infty\|x\|_\infty\\
	&\leq(L+1)\tilde{U}^{L+1}C^Ld\|w-\tilde{w}\|_\infty(\sigma_0+\|x\|_\infty).
\end{align*}
Let $A\coloneqq(L+1)\tilde{U}^{L+1}C^Ld$.
Partitioning the weight space $[-C,C]^S$ into cubes of length $2\varepsilon/A$ creates $\lceil C A/\varepsilon\rceil^S$ cubes.
Hence, the covering number is bounded as $N(\varepsilon,[-C,C]^S,\|\cdot\|_\infty)\leq\lceil C A/\varepsilon\rceil^S$.
The bound on the bracketing number then follows from \citet[Theorem 2.7.11]{vw1996}, observing that the proof thereof works for a premetric with modification of $2\varepsilon\|F\|$ to $\|2\varepsilon F\|_{d_{\mathcal{F}}}$.

For a fully connected network, the number of all weights is $dU$ (weights for the first layer) plus $(L-1)U^2$ (weights for the remaining hidden layers) plus $U$ (weights in the output layer), summing to $(LU+1)U+(d-U)U$.%
\footnote{If the network has a bias term, the actual variable weights are slightly fewer, but it does not change the order.}
For a network $\mathcal{H}^{(0)}$ in \citet{bk2019} (in their notation), the number of all weights is $A^{(0)}\coloneqq d(4d^\ast M_\ast)+4d^\ast M_\ast+M_\ast=4(1+d)d^\ast M_\ast+M_\ast$.
For $\mathcal{H}^{(1)}$, $A^{(1)}\coloneqq A^{(0)}K+K(4d^\ast M_\ast)+4d^\ast M_\ast+M_\ast=A^{(0)}K+4(1+K)d^\ast M_\ast+M_\ast$.
For $\mathcal{H}^{(l)}$,
\(
	A^{(l)}\coloneqq A^{(l-1)}K+4(1+K)d^\ast M_\ast+M_\ast
	=A^{(0)}K^l+\sum_{j=0}^{l-1}K^j[4(1+K)d^\ast M_\ast+M_\ast]
	=4d^\ast M_\ast[(1+d)K^l+\frac{1-K^l}{1-K}(1+K)]+M_\ast\frac{1-K^{l+1}}{1-K}=O(dd^\ast M_\ast K^l)
\).
Then use $L=2+3l$ and $\tilde{U}=M_\ast\vee(4d^\ast)\vee K$.
\end{proof}

\begin{rem}
\cref{lem:bracket} assumes a Lipschitz property for the activation and output functions, which accommodates ReLU, softplus, and sigmoid, but not perceptron.
\end{rem}

\section{Supporting Lemmas for the Main Text}

The following lemma shows local convergence of the loss needed for \cref{thm:theta:dist}.

\begin{lem}[Asymptotic distribution of objective function] \label{lem:lan:misspec}
Under \cref{asm:m,asm:dqm}, for every compact $K\subset\Theta$, uniformly in $h\in K$,
\begin{multline*}
	\!\!\!\!
	n[\mathbb{M}_{\theta_0+h/\sqrt{n}}(D_{\theta_0+h/\sqrt{n}})-\mathbb{M}_{\theta_0}(D_{\theta_0})]
	=-\sqrt{n}\mathbb{P}_0h^\top\dot{\ell}_{\theta_0}+\sqrt{n}(\mathbb{P}_0+\mathbb{P}_{\theta_0+h/\sqrt{n}})D_{\theta_0+h/\sqrt{n}}h^\top\dot{\ell}_{\theta_0}\\
	+n[(\mathbb{P}_{\theta_0+h/\sqrt{n}}-P_{\theta_0+h/\sqrt{n}})-(\mathbb{P}_{\theta_0}-P_{\theta_0})]\log(1-D_{\theta_0})
	+\tfrac{h^\top\tilde{I}_{\theta_0}h}{4}+o_P(1).
\end{multline*}
With \cref{asm:model:smooth,asm:misspec}, this reduces to
\[
	-\sqrt{n}\mathbb{P}_0 h^\top\dot{\ell}_{\theta_0}+\sqrt{n}(\mathbb{P}_0+\mathbb{P}_{\theta_0})D_{\theta_0}h^\top\dot{\ell}_{\theta_0}
	+\sqrt{n}\tilde{\mathbb{P}}_0 h^\top\tau_n+\tfrac{h^\top\tilde{I}_{\theta_0}h}{4}+o_P(1).
\]
\end{lem}

\begin{proof}%
Let $\theta\coloneqq\theta_0+h/\sqrt{n}$, $W\coloneqq\sqrt{\vphantom{D}\smash{D_\theta/D_{\theta_0}}}-1$, $\tilde{W}\coloneqq\sqrt{\vphantom{D}\smash{p_{\theta_0}/p_\theta}}-1$.
Observe that
\[
	n[\mathbb{M}_\theta(D_\theta)-\mathbb{M}_{\theta_0}(D_{\theta_0})]
	=n(\mathbb{P}_0+\mathbb{P}_\theta)\log\tfrac{D_\theta}{D_{\theta_0}}-n\mathbb{P}_\theta\log\tfrac{p_{\theta_0}}{p_\theta}+n(\mathbb{P}_\theta-\mathbb{P}_{\theta_0})\log(1-D_{\theta_0}).
\]
We examine each term separately.
By \cref{asm:dqm},
\begin{align*}
	n(P_\theta-P_{\theta_0})\log(1-D_{\theta_0})&=n\int(\sqrt{p_\theta}+\sqrt{p_{\theta_0}})(\sqrt{p_\theta}-\sqrt{p_{\theta_0}})\log(1-D_{\theta_0})\\
	&=\int\bigl(\sqrt{n}h^\top\dot{\ell}_{\theta_0}+\tfrac{h^\top\ddot{\ell}_{\theta_0}h}{2}+\tfrac{h^\top\dot{\ell}_{\theta_0}\dot{\ell}_{\theta_0}^\top h}{2}\bigr)p_{\theta_0}\log(1-D_{\theta_0})+o(1).
\end{align*}
The first term is zero since $M_\theta(D_\theta)-M_{\theta_0}(D_{\theta_0})\geq0$ and $M_\theta(D_\theta)-M_{\theta_0}(D_{\theta_0})=2\int D_{\theta_0}(\sqrt{p_\theta}-\sqrt{p_{\theta_0}})^2+o(h(\theta,\theta_0)^2)+(P_\theta-P_{\theta_0})\log(1-D_{\theta_0})$.%
\footnote{The term $P_{\theta_0}h^\top\dot{\ell}_{\theta_0}\log(1-D_{\theta_0})$ is the only term that is linear in $h=h(\theta,\theta_0)$, so if it is not zero, then $M_\theta(D_\theta)-M_{\theta_0}(D_{\theta_0})\geq0$ is violated.}
Therefore,
\(
	n(P_\theta-P_{\theta_0})\log(1-D_{\theta_0})=\frac{1}{2}P_{\theta_0}(h^\top\ddot{\ell}_{\theta_0}h+h^\top\dot{\ell}_{\theta_0}\dot{\ell}_{\theta_0}^\top h)\log(1-D_{\theta_0})+o(1)
\).

Using $\log x=2(\sqrt{x}-1)-(\sqrt{x}-1)^2+(\sqrt{x}-1)^2R(\sqrt{x}-1)$ for $R(x)=O(x)$,
\[
	n(\mathbb{P}_0+\mathbb{P}_\theta)\log\tfrac{D_\theta}{D_{\theta_0}}=2 n(\mathbb{P}_0+\mathbb{P}_\theta)W-n(\mathbb{P}_0+\mathbb{P}_\theta)W^2+n(\mathbb{P}_0+\mathbb{P}_\theta)W^2R(W_n).
\]
Let $\breve{I}_{\theta_0}\coloneqq 2 P_{\theta_0}D_{\theta_0}\dot{\ell}_{\theta_0}\dot{\ell}_{\theta_0}^\top$.
Observe that
\[
	(P_0+P_\theta)\bigl(\sqrt{n}W+\tfrac{h^\top\dot{\ell}_{\theta_0}}{2}(1-D_\theta)\bigr)^2
	=n\int\bigl[\sqrt{p_0+p_{\theta_0}}-\sqrt{p_0+p_\theta}+\tfrac{h^\top\dot{\ell}_{\theta_0}}{2\sqrt{n}}\sqrt{(1-D_\theta)p_\theta}\bigr]^2,
\]
which is $o(\|h\|^2/n)$ by \cref{lem:average,asm:dqm}.
Thus, the RHS converges to zero uniformly over every compact $K\subset\Theta$.
We draw two observations: (i) the mean and variance of $(\sqrt{n}W+(1-D_\theta)h^\top\dot{\ell}_{\theta_0}/2)(X_i)$, $X_i\sim(P_0+P_{\theta_n})/2$, converge to zero and so does the variance of $\sqrt{n}(\mathbb{P}_0+\mathbb{P}_\theta)(\sqrt{n}W+(1-D_\theta)h^\top\dot{\ell}_{\theta_0}/2)$ under \cref{asm:m};%
\footnote{This does not imply that the mean of $\sqrt{n}(\mathbb{P}_0+\mathbb{P}_\theta)(\sqrt{n}W+(1-D_\theta)h^\top\dot{\ell}_{\theta_0}/2)$ converges to zero.}
(ii) $(P_0+P_\theta)|nW^2-(1-D_\theta)^2(h^\top\dot{\ell}_{\theta_0}/2)^2|\to0$, so $n(\mathbb{P}_0+\mathbb{P}_\theta)W^2=(\mathbb{P}_0+\mathbb{P}_\theta)(1-D_\theta)^2(h^\top\dot{\ell}_{\theta_0}/2)^2+o_P(1)\to h^\top I_{\theta_0}h/4-h^\top\breve{I}_{\theta_0}h/8$.
Next,
\begin{gather*}
	n(P_0+P_\theta)W
	=-\tfrac{n}{2}h(p_0+p_{\theta_0},p_0+p_\theta)^2\conv-\tfrac{h^\top I_{\theta_0}h}{8}+\tfrac{h^\top\breve{I}_{\theta_0}h}{16},\\
	\sqrt{n}(P_0+P_\theta)(1-D_\theta)\tfrac{h^\top\dot{\ell}_{\theta_0}}{2}=\sqrt{n}P_\theta\tfrac{h^\top\dot{\ell}_{\theta_0}}{2}=\sqrt{n}(P_\theta-P_{\theta_0})\tfrac{h^\top\dot{\ell}_{\theta_0}}{2}\to\tfrac{h^\top I_{\theta_0}h}{2}.
\end{gather*}
This implies that the mean of
\(
	\sqrt{n}(\mathbb{P}_0+\mathbb{P}_\theta)(\sqrt{n}W+(1-D_\theta)h^\top\dot{\ell}_{\theta_0}/2)
\)
converges to $3h^\top I_{\theta_0}h/8+h^\top\breve{I}_{\theta_0}h/16$.
Combining with (i), we find
\[
	n(\mathbb{P}_0+\mathbb{P}_\theta)W=-\sqrt{n}(\mathbb{P}_0+\mathbb{P}_\theta)(1-D_\theta)\tfrac{h^\top\dot{\ell}_{\theta_0}}{2}+\tfrac{3h^\top I_{\theta_0}h}{8}+\tfrac{h^\top\breve{I}_{\theta_0}h}{16}+o_P(1).
\]
The remainder term $n(\mathbb{P}_0+\mathbb{P}_\theta)W^2R(W_n)$ vanishes by the same logic as \citet[Theorem 7.2]{v1998}.

Next, observe that
\(
	n\mathbb{P}_\theta\log\tfrac{p_{\theta_0}}{p_\theta}
	=2n\mathbb{P}_\theta\tilde{W}-n\mathbb{P}_\theta\tilde{W}^2+n\mathbb{P}_\theta\tilde{W}^2R(\tilde{W})
\)
and
\[
	P_\theta\bigl(\sqrt{n}\tilde{W}+\tfrac{h^\top\dot{\ell}_{\theta_0}}{2}\bigr)^2=n\int\bigl[\sqrt{p_{\theta_0}}-\sqrt{p_\theta}+\tfrac{h^\top\dot{\ell}_\theta}{2\sqrt{n}}\sqrt{p_\theta}\bigr]^2=o\bigl(\tfrac{\|h\|^2}{n}\bigr).
\]
Again, (i) the mean and variance of $(\sqrt{n}\tilde{W}+h^\top\dot{\ell}_{\theta_0}/2)(X_i)$, $X_i\sim P_\theta$, converge to zero and so does the variance of $\sqrt{n}\mathbb{P}_\theta(\sqrt{n}\tilde{W}+h^\top\dot{\ell}_{\theta_0}/2)$ under \cref{asm:m}; (ii) $P_\theta|n\tilde{W}^2-(h^\top\dot{\ell}_{\theta_0}/2)^2|\to0$, so $n\mathbb{P}_\theta\tilde{W}^2\to P_\theta(h^\top\dot{\ell}_{\theta_0}/2)^2\to h^\top I_{\theta_0}h/4$.
Next,
\(
	nP_\theta\tilde{W}=-nh(\theta,\theta_0)^2/2\conv-h^\top I_{\theta_0}h/8
\)
and
\(
	\sqrt{n}P_\theta h^\top\dot{\ell}_{\theta_0}/2\conv h^\top I_{\theta_0}h/2
\).
This implies that the mean of $\sqrt{n}\mathbb{P}_\theta(\sqrt{n}\tilde{W}+h^\top\dot{\ell}_{\theta_0}/2)$ converges to $3h^\top I_{\theta_0}h/8$.
Thus, we find
\[
	n\mathbb{P}_\theta\tilde{W}=-\sqrt{n}\mathbb{P}_\theta\tfrac{h^\top\dot{\ell}_{\theta_0}}{2}+\tfrac{3h^\top I_{\theta_0}h}{8}+o_P(1).
\]
We may once again ignore the remainder term $n\mathbb{P}_\theta\tilde{W}^2R(\tilde{W})$.
Altogether, with $\tilde{I}_{\theta_0}$ defined in \cref{asm:dqm},
\begin{multline*}
	n[\mathbb{M}_\theta(D_\theta)-\mathbb{M}_{\theta_0}(D_{\theta_0})]
	=-\sqrt{n}\mathbb{P}_0 h^\top\dot{\ell}_{\theta_0}+\sqrt{n}(\mathbb{P}_0+\mathbb{P}_\theta)D_\theta h^\top\dot{\ell}_{\theta_0}+\tfrac{h^\top\tilde{I}_{\theta_0}h}{4}\\
	+n[(\mathbb{P}_\theta-\mathbb{P}_{\theta_0})-(P_\theta-P_{\theta_0})]\log(1-D_{\theta_0})+o_P(1).
\end{multline*}
For the second claim, it remains to show that with \cref{asm:model:smooth},
\[
	\sqrt{n}(\mathbb{P}_0+\mathbb{P}_\theta)D_\theta h^\top\dot{\ell}_{\theta_0}-\sqrt{n}(\mathbb{P}_0+\mathbb{P}_{\theta_0})D_{\theta_0}h^\top\dot{\ell}_{\theta_0}=o_P(1).
\]
Note that $(P_0+P_\theta)D_\theta h^\top\dot{\ell}_{\theta_0}-(P_0+P_{\theta_0})D_{\theta_0}h^\top\dot{\ell}_{\theta_0}=0$.
Write
\[
	\sqrt{n}(\mathbb{P}_0+\mathbb{P}_\theta)(D_\theta-D_{\theta_0})h^\top\dot{\ell}_{\theta_0}+\sqrt{n}(\mathbb{P}_\theta-\mathbb{P}_{\theta_0})D_{\theta_0}h^\top\dot{\ell}_{\theta_0}.
\]
The second term converges to $h^\top\breve{I}_{\theta_0}h/2$ under \cref{asm:model:smooth}.
Since $p/(p+x)$ is convex in $x\geq0$ for $p>0$,
\(
	D_{\theta_0}\frac{p_{\theta_0}-p_\theta}{p_0+p_{\theta_0}}\leq D_\theta-D_{\theta_0}\leq D_\theta\frac{p_{\theta_0}-p_\theta}{p_0+p_\theta}
\)
by Taylor's theorem.
Therefore, by \cref{asm:dqm},
\begin{align*}
	{-(\mathbb{P}_0+\mathbb{P}_\theta)}D_{\theta_0}(1-D_{\theta_0})(h^\top\dot{\ell}_{\theta_0})^2+o_P(1)
	&\leq\sqrt{n}(\mathbb{P}_0+\mathbb{P}_\theta)(D_\theta-D_{\theta_0})h^\top\dot{\ell}_{\theta_0}\\
	&\leq-(\mathbb{P}_0+\mathbb{P}_\theta)D_\theta(1-D_\theta)(h^\top\dot{\ell}_{\theta_0})^2+o_P(1).
\end{align*}
Thus, the first term converges to $-P_{\theta_0}D_{\theta_0}(h^\top\dot{\ell}_{\theta_0})^2=-h^\top\breve{I}_{\theta_0}h/2$ in probability.
\end{proof}

The Bernstein ``norm'' of a function $f$ is defined as $\|f\|_{P,B}\coloneqq\sqrt{2P(e^{|f|}-1-|f|)}$; this induces a premetric without the triangle inequality \citep[p.\ 324]{vw1996}.%
\footnote{A {\em premetric} on a class of functions $\mathcal{F}$ is a function $d:\mathcal{F}\times\mathcal{F}\to\mathbb{R}$ that satisfies $d(f,f)=0$ and $d(f,g)=d(g,f)\geq0$ for every $f,g\in\mathcal{F}$.} %
The next lemma bounds the Bernstein ``norm'' of a log likelihood ratio by the Hellinger distance without assuming a bounded likelihood ratio.

\begin{lem}[Bernstein ``norm'' of log likelihood ratio; {\citealp[Lemma 2.1 (iv)]{kr2022}}] \label{lem:vaart}
For any pair of probability measures $P$ and $P_0$ such that $P_0(p_0/p)<\infty$,
\[
	\bigl\|\tfrac{1}{2}\log\tfrac{p}{p_0}\bigr\|_{P_0,B}^2
	\leq2h(p,p_0)^2\bigl[1+P_0\bigl(\tfrac{p_0}{p}\bigm|\tfrac{p_0}{p}\geq\tfrac{25}{16}\bigr)\bigr],
\]
where $P_0(p_0/p\mid p_0/p\geq a)=0$ if $P_0(p_0/p\geq a)=0$.
\end{lem}

\begin{rem}
Similarly, we have
\begin{gather*}
	\bigl\|\tfrac{1}{2}\log\tfrac{D}{D_\theta}\bigr\|_{P_0,B}^2\leq2h_\theta(D,D_\theta)^2\bigl[1+P_0\bigl(\tfrac{D_\theta}{D}\bigm|\tfrac{D_\theta}{D}\geq\tfrac{25}{16}\bigr)\bigr],\\
	\bigl\|\tfrac{1}{2}\log\tfrac{1-D}{1-D_\theta}\bigr\|_{P_\theta,B}^2\leq2h_\theta(1-D,1-D_\theta)^2\bigl[1+P_\theta\bigl(\tfrac{1-D_\theta}{1-D}\bigm|\tfrac{1-D_\theta}{1-D}\geq\tfrac{25}{16}\bigr)\bigr].
\end{gather*}
\end{rem}

\begin{lem}[Bernstein ``norm'' of log discriminator ratio] \label{lem:distance}
For every $\theta_1,\theta_2\in\Theta$,
\[
	\bigl\|\log\tfrac{D_{\theta_1}}{D_{\theta_2}}\bigr\|_{P_0,B}^2\leq 8 h(\theta_1,\theta_2)^2,\qquad
	\bigl\|\log\tfrac{(1-D_{\theta_1})\circ T_{\theta_1}}{(1-D_{\theta_2})\circ T_{\theta_2}}\bigr\|_{\tilde{P}_0,B}^2\leq8\tilde{h}(\theta_1,\theta_2)^2.
\]
\end{lem}

\begin{proof}%
Since $e^{|x|}-1-|x|\leq2(e^{x/2}-1)^2$ for $x\geq0$,
\begin{align*}
	\bigl\|\log\tfrac{D_{\theta_1}}{D_{\theta_2}}\bigr\|_{P_0,B}^2
	&\leq4P_0\Bigl(\sqrt{\tfrac{D_{\theta_1}}{D_{\theta_2}}}-1\Bigr)^2\mathbbm{1}\{D_{\theta_1}\geq D_{\theta_2}\}+4P_0\Bigl(\sqrt{\tfrac{D_{\theta_2}}{D_{\theta_1}}}-1\Bigr)^2\mathbbm{1}\{D_{\theta_1}<D_{\theta_2}\}\\
	&\leq4P_0\bigl(\sqrt{\tfrac{p_0+p_{\theta_2}}{\smash{p_0+p_{\theta_1}}}}-1\bigr)^2+4P_0\bigl(\sqrt{\tfrac{p_0+p_{\theta_1}}{\smash{p_0+p_{\theta_2}}}}-1\bigr)^2\\
	&\leq8\int(\sqrt{p_0+p_{\theta_1}}-\sqrt{p_0+p_{\theta_2}})^2
	\leq8\int(\sqrt{p_{\theta_1}}-\sqrt{p_{\theta_2}})^2
	\leq8h(\theta_1,\theta_2)^2.
\end{align*}
Similarly,
\[
	\bigl\|\log\tfrac{(1-D_{\theta_1})\circ T_{\theta_1}}{(1-D_{\theta_2})\circ T_{\theta_2}}\bigr\|_{\tilde{P}_0,B}^2
	\leq4\tilde{P}_0\Bigl(\sqrt{\tfrac{(1-D_{\theta_1})\circ T_{\theta_1}}{\smash{(1-{}}D_{\theta_2}\smash{)\circ T_{\theta_2}}}}-1\Bigr)^2+4\tilde{P}_0\Bigl(\sqrt{\tfrac{(1-D_{\theta_2})\circ T_{\theta_2}}{\smash{(1-{}}D_{\theta_1}\smash{)\circ T_{\theta_1}}}}-1\Bigr)^2
	\leq8\tilde{h}(\theta_1,\theta_2)^2
\]
since
\begin{align*}
	\tilde{P}_0\Bigl(\sqrt{\tfrac{(1-D_{\theta_1})\circ T_{\theta_1}}{\smash{(1-{}}D_{\theta_2}\smash{)\circ T_{\theta_2}}}}-1\Bigr)^2
	&\leq\tilde{P}_0\bigl(\tfrac{1}{\sqrt{\vphantom{D^\theta}\smash{(1-D_{\theta_2})\circ T_{\theta_2}}}}-\tfrac{1}{\sqrt{\vphantom{D^\theta}\smash{(1-D_{\theta_1})\circ T_{\theta_1}}}}\bigr)^2\\
	&\leq\tilde{P}_0\bigl(\sqrt{\smash{\tfrac{p_0}{p_{\theta_2}}}\circ T_{\theta_2}}-\sqrt{\smash{\tfrac{p_0}{p_{\theta_1}}}\circ T_{\theta_1}}\bigr)^2=\tilde{h}(\theta_1,\theta_2)^2.
	\tag*\qedhere
\end{align*}
\end{proof}

\begin{lem}[Hellinger distance of sums of densities] \label{lem:average}
For arbitrary densities $p$, $p_0$, $p_1$,
\[
	h(p+p_0,p+p_1)^2=\int\tfrac{p_0}{p+p_0}(\sqrt{p_0}-\sqrt{p_1})^2+o(h(p_0,p_1)^2).
\]
\end{lem}

\begin{proof}%
Since $\sqrt{p+x^2}$ is uniformly differentiable in $x$ with derivative $x/\sqrt{p+x^2}$, %
the result follows by expanding $\sqrt{p_1}$ around $\sqrt{p_0}$.
\end{proof}

\section{Additional Notes on the Empirical Application}

\subsection{Identifying Role of Health Status} \label{sec:health}

The health status is a variable that was not used in the moments of \citetalias{dfj}; we argue that this gives additional variation to identify the bequest motive. %
Disentangling the bequest motive from the medical expenditure risk is a challenging task. As the bequest is a luxury good, we may expect that its identifying power comes from wealthy individuals. However, wealthy individuals are also ones with the longest life expectancy, being motivated to save for medical expenses.
Indeed, \citetalias{dfj} document that the medical expenditure for the rich skyrockets after age 95, reaching \$15,000 by age 100.
However, if the health condition diminishes their life expectancy, those with shorter horizons would face much less incentive to save for the coming medical expenses while as much incentive to save for bequests.

\begin{figure}[t!]
\centering
\begin{subfigure}[t]{0.45\textwidth}
\centering \includegraphics[page=1]{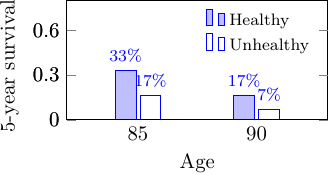}
\caption{Men's five\hyp{}year survival rates.} \label{fig:1g}
\end{subfigure}
\qquad
\begin{subfigure}[t]{0.45\textwidth}
\centering \includegraphics[page=2]{fig_c3s.pdf}
\caption{Women's five\hyp{}year survival rates.} \label{fig:1h}
\end{subfigure}
\\\vspace{6pt}
\begin{subfigure}[t]{0.45\textwidth}
\centering \includegraphics[page=1]{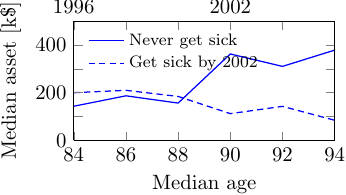}
\caption{Men's asset.} \label{fig:1a}
\end{subfigure}
\qquad
\begin{subfigure}[t]{0.45\textwidth}
\centering \includegraphics[page=5]{fig_c3.pdf}
\caption{Women's asset.} \label{fig:1b}
\end{subfigure}
\\\vspace{6pt}
\begin{subfigure}[t]{0.45\textwidth}
\centering \includegraphics[page=2]{fig_c3.pdf}
\caption{Men's medical expenses.} \label{fig:1e}
\end{subfigure}
\qquad
\begin{subfigure}[t]{0.45\textwidth}
\centering \includegraphics[page=6]{fig_c3.pdf}
\caption{Women's medical expenses.} \label{fig:1f}
\end{subfigure}
\\\vspace{6pt}
\begin{subfigure}[t]{0.45\textwidth}
\centering \includegraphics[page=3]{fig_c3.pdf}
\caption{Men's permanent income.} \label{fig:1c}
\end{subfigure}
\qquad
\begin{subfigure}[t]{0.45\textwidth}
\centering \includegraphics[page=7]{fig_c3.pdf}
\caption{Women's permanent income.} \label{fig:1d}
\end{subfigure}
\caption{Profiles by gender and health. (\subref{fig:1a}) to (\subref{fig:1f}) are for 4--5th PIqs in Cohort 3. Solid lines are for those who stay healthy for the duration of their observation; dashed lines for those who are healthy in 1996 and become unhealthy by 2002.}
\label{fig_health_asset}
\end{figure}

We find some evidence of this in our dataset.
\Cref{fig:1g,fig:1h} are the proportions of individuals who survive for the next five years at ages 85 and 90, conditional on gender and health.
We see that the health status, along with gender, is a strong predictor of life expectancy in years when the medical expenditure soars.

Heterogeneity in the survival materializes as the difference in savings.
\Cref{fig:1a,fig:1b} give the trajectories of the median assets for the 4th and 5th PI quintiles in Cohort 3.
The solid lines are those who were healthy throughout the survey periods and the dashed lines are those who were healthy in 1996 but reported unhealthy in 1998, 2000, or 2002.
We see that men who were exposed to the health shock (hence the survival shock) dig into their savings much more than healthy men. With higher survival rates, women exhibit the trend to a much lesser degree.

Such difference in the asset profiles seems driven neither by the difference in medical expenses nor by survival selection among the rich.
\Cref{fig:1e,fig:1f} show the median medical expenses during the same periods; we observe similar trajectories across gender and health.
\Cref{fig:1c,fig:1d} show the median PI quantiles of the survivors; if there is attrition of rich or poor individuals that affects the median assets, we expect to see a change in the median PI quantiles. However, they do not differ much by at least age 90 while bifurcation of the asset profiles begins at age 90.

These findings suggest that the difference in the asset profiles is attributable to the change in the saving behaviors. 
The health status changes the exposure to the medical expenditure risk through the survival probability, which then induces changes in the saving behavior by shifting the balance between the bequest motive and the medical expenditure risk.

\subsection{Estimation and Inference Procedure}\label{sec:estimation_al}

Estimation of original GAN for a generative model of images is known to be challenging \citep{arjovsky2017towards}.
Two main issues are (i) the mode\hyp{}seeking behavior of the discriminator due to imbalance between $n$ and $m$ and (ii) the flat gradient of the loss in $\theta$ when actual and synthetic samples are easily distinguishable.

Imbalance in the sample sizes arises naturally in our context; in order to reduce the variance of $\hat{\theta}$, we may want to have $m \gg n$.
With this, however, there is a risk that $\hat{D}_\theta$ becomes near zero everywhere.
We follow the literature recommendation to perform data augmentation of the actual data; we resample the histories of assets of individuals with replacement until $n$ and $m$ are even.

The flat gradient seems not nearly as pervasive when the generative model is a typical structural economic model. %
\cite{arjovsky2017towards} show that the flat gradient is related to the support problem in the generative model (see Lemma 1 and Theorem 2.1 therein) where the set of realizable images is of measure zero in the space of all images.
However, economic models are much lower\hyp{}dimensional and economic data often span the same (subset of a) Euclidean space across parameter values.
In our empirical application, outcomes are continuous and disjoint supports are not a first\hyp{}order problem. Nonetheless, gradients of the loss can approach $0$ when the conditional distributions of the simulated outcomes and of the actual outcomes are far apart, slowing the naive gradient descent.
We implement two speeding strategies popular in training neural networks: the Nesterov accelerated gradient (NAG), an accelerated gradient descent featuring momentum 
\citep{nesterov27method}, and resilient propagation (RPROP), an adaptive learning rate algorithm \citep{riedmiller1993direct}.

Finally, we give details on the tuning parameters for the discriminator.
Recall that $\mathcal{D}$ is the set of feedforward neural networks with two hidden layers with 20 and 10 neurons, respectively, with the sigmoid activation function.
We use the R Keras package, namely, the default ADAM optimization algorithm that incorporates the stochastic gradient descent and backpropagation for fast computation of the gradient.
For the stochastic gradient descent, we select a small batch of 120 samples per gradient calculation and a large number of epochs (2000).
In contrast to other implementations of GAN, we train the discriminator ``to completion'' and fix the seed of the stochastic gradient before each training to preserve non\hyp{}randomness of the criterion.
We find that this strategy delivers the most reliable estimates, albeit computationally intensive.
To avoid overfitting, we make use of callback options that track the evolution of out\hyp{}of\hyp{}sample accuracy measures over epochs.

For standard errors, we use poor (wo)man's bootstrap based on 50 replications.
For each replication, we solve nine one\hyp{}dimensional optimization problems in the directions defined in \citet[Corollary 2]{honore2017poor}.
We treat the network configuration as fixed, so we do not repeat cross validation for each bootstrap sample.

\subsection{Fit and Counterfactual Simulations} \label{sec:counterfactual}

Similarly as \citetalias{dfj}, we look at the assets one period before deaths to compare the fit and counterfactuals.
Individuals who passed away during the survey periods are divided into five groups of permanent income quintiles (PIqs).
We take the assets in the last survey when they were alive and sum these across individuals in each group.

\Cref{table_beq} shows the actual and simulated assets one period before deaths.
Adversarial $X_2$ baseline and \citetalias{dfj} baseline rows are the simulations with parameters equal to the estimates of our preferred specification and of \citetalias{dfj}.
Our estimates fit the assets for low PIqs well but overestimates high PIqs, while \citetalias{dfj} show the opposite pattern.%
\footnote{Trimming observations above top 1\% of mean assets significantly decreases discrepancy between observed assets and the predicted assets with $X_2$ of the actual data. Results are available upon request. In addition, the gap in the fit between the poor and the rich might be attributed to the rich doing inter vivos transfers more often than the poor, biasing the assets of the rich downwards toward the end of their lives \citep{mcgarry1999inter}.}

\begin{table}[t!]

\caption{Fit of the savings and counterfactual simulation without bequest motive and medical expense risk. ``No bequest'' rows are the simulation of the model with $\vartheta=0$ (so $\phi\equiv0$). ``No medical risk'' rows are the simulation of the model with $\sigma\equiv0$ (so $\log m_t=m$). Each number is a cross\hyp{}sectional sum of assets of individuals one period before their deaths in the units of k\$, a proxy for their intended bequest. Percentages are relative to the corresponding baselines.}

\label{table_beq}
\centering
\small
\begin{tabular}{lccccc}
\toprule \midrule
& \multicolumn{5}{c}{Permanent income quintile} \\
\cmidrule{2-6}
& 1st & 2nd & 3rd & 4th & 5th \\
\midrule
Actual data & 18,191 & 25,266 & 42,006 & 50,495 & 85,814 \\
[0.4em]
Adversarial $X_2$ baseline & 20,441 & 26,366 & 51,339 & 62,662 & \llap{1}10,385 \\
[0.2em]
\quad No bequest & 17,644 & 21,587 & 42,586 & 50,631 & 95,212 \\
\quad (\% difference) & (13.7\%) & (18.1\%) & (17.1\%) & (19.2\%) & (13.7\%) \\
[0.2em]
\quad No medical risk & 18,890 & 23,252 & 43,789 & 49,385 & 90,204 \\
\quad (\% difference) & (\hphantom{0}7.6\%) & (11.8\%) & (14.7\%) & (21.2\%) & (18.3\%) \\
[0.4em]
\citetalias{dfj} baseline & 16,527 & 19,672 & 38,157 & 42,737 & 83,814 \\
[0.2em]
\quad No bequest & 16,342 & 19,605 & 37,387 & 42,425 & 83,563 \\
\quad (\% difference) & (\hphantom{0}1.1\%) & (\hphantom{0}0.3\%) & (\hphantom{0}2.1\%) & (\hphantom{0}0.7\%) & (\hphantom{0}0.5\%) \\
[0.2em]
\quad No medical risk & 16,440 & 19,242 & 36,157 & 38,053 & 76,080 \\
\quad (\% difference) & (\hphantom{0}0.5\%) & (\hphantom{0}2.2\%) & (\hphantom{0}5.4\%) & (11.0\%) & (\hphantom{0}9.4\%) \\
\bottomrule
\end{tabular}
\end{table}

Next, we perform two counterfactual simulations to measure the elderly's saving motive in terms of (i) bequest and (ii) medical expenditure risk.
We simulate the model with the same parameters except that we kill either the bequest incentive, $\phi\equiv0$, or the medical expenditure risk, $\sigma\equiv0$.
The ``(\% difference)'' rows give the difference of the baseline and counterfactual relative to the baseline.

The contribution of the bequest motive to the savings differs substantially between our estimates and \citetalias{dfj}.
In our estimates, the lack of the bequest motive decreases the savings by 13.7\% to 19.2\%, while \citetalias{dfj} estimates suggest at most 2.1\% decrease.
This is largely due to the difference in the estimates of the curvature $k$.
According to our estimates, the bequest motive is an important and substantial source of savings for both the poor and the rich.
This finding is consistent with \citet{lockwood2018incidental} who uses additional data on annuity takeup to identify the bequest motive.

The contribution of the medical expenditure risk looks much more in line for the two models.
The amount of savings to prepare for uncertain medical expenses is substantial in both predictions.
This is because rich individuals live long and hence are at high risk of large medical expenses.
Poor individuals do not survive long enough and are more likely to be covered by social insurance programs.

To summarize, our adversarial estimates reveal with precision that the bequest motive contributes in similar magnitudes to the slow decrease in the elderly's savings across PIqs. 
The uncertainty in medical expenses contribute less for poor individuals.

\begin{singlespacing}
\bibliographystyle{ecta}
\bibliography{reference}
\end{singlespacing}